\documentclass[12pt]{amsart}
\usepackage[margin=1in]{geometry}
\usepackage{amsfonts}
\usepackage{amsmath,amsthm}
\usepackage{amssymb}
\usepackage{epsfig}
\usepackage{dsfont}
\usepackage{euscript}
\usepackage{color}
\usepackage{anysize}
\usepackage[breaklinks,bookmarks=false]{hyperref}
\hypersetup{colorlinks, linkcolor=blue, citecolor=blue,urlcolor=blue, plainpages=false}

\usepackage{paralist}
\usepackage{bm}
\usepackage{url}
\usepackage{psfrag}
\usepackage{listings}
\usepackage{moreverb}
\usepackage{float}

\usepackage{stmaryrd}
\usepackage{verbatim}
\usepackage[normalem]{ulem}

 \newtheorem{theorem}{Theorem}[section]
 \newtheorem{corollary}[theorem]{Corollary}
 \newtheorem{lemma}[theorem]{Lemma}
 \newtheorem{proposition}[theorem]{Proposition}
 \theoremstyle{definition}

 \theoremstyle{remark}
 \newtheorem{remark}[theorem]{Remark}

 \numberwithin{equation}{section}

\usepackage{setspace} 

\usepackage{epstopdf}
\usepackage{color}
\usepackage{wrapfig}
\usepackage{dsfont}
\usepackage[center]{caption}
\usepackage{multirow}
\usepackage[all,cmtip]{xy}
\usepackage[titletoc]{appendix}
\usepackage{rotating}
\usepackage{fancyhdr}
\usepackage[charter]{mathdesign}
\usepackage[english]{minitoc} 
\usepackage{relsize}
\usepackage{tikz}
\usetikzlibrary{positioning}
\usepackage{graphicx}
\newcommand\numberthis{\addtocounter{equation}{1}\tag{\theequation}}

\usepackage{ae,aecompl,aeguill}

\newcommand{\C}{\mathbb{C}}
\newcommand{\R}{\mathbb{R}}

\newcommand{\Z}{\mathbb{Z}}

\newcommand{\SpC}{\mathrm{Sp}_{\mathrm{N}}^*(\C)}
\newcommand{\one}{\mathrm{I}_{\mathrm{L}}}

\newcommand{\cred}[1]{{\color{red}#1}}
\definecolor{Plum}{rgb}{.5,0,1}

\numberwithin{equation}{section}

\newcommand{\hakim}[1]{\cred{{\it Remark by Hakim:} #1 {\it End of Remark.}}}

\begin{document}

\title{Dynamical localization for a random scattering zipper}

\author[H. Boumaza]{Hakim Boumaza}
\address[H. Boumaza]{LAGA, Universit\'e Sorbonne Paris Nord, 99 avenue J.B. Cl\'ement, F-93430 Villetaneuse}
\email[H. Boumaza]{boumaza@math.univ-paris13.fr}
\author[A. Khouildi]{Amine Khouildi}
\address[A. Khouildi]{LAGA, Universit\'e Sorbonne Paris Nord, 99 avenue J.B. Cl\'ement, F-93430 Villetaneuse}
\email[A. Khouildi]{khouildi@math.univ-paris13.fr}

\thanks{The research of H.B was partially funded by the ANR Project RAW, ANR-20-CE40-0012 \\    Orcid number of H.B. : 0000-0001-5255-3887}
\date{}

\begin{abstract}
This article establishes a proof of dynamical localization for a random scattering zipper model. The scattering zipper operator is the product of two unitary by blocks operators, multiplicatively perturbed on the left and right by random unitary phases. One of the operator is shifted so that this configuration produces a random 5-diagonal unitary operator per blocks. To prove the dynamical localization for this operator, we use the method of fractional moments. We first prove the continuity and strict positivity of the Lyapunov exponents in an annulus around the unit circle, which leads to the exponential decay of a power of the norm of the products of transfer matrices. We then establish an explicit formula of the coefficients of the finite resolvent in terms of the coefficients of the transfer matrices using Schur's complement. From this we deduce, through two reduction results, the exponential decay of the resolvent, from which we get the dynamical localization.
\end{abstract}

\maketitle





\DeclareGraphicsExtensions{{.pdf}}




\section{Introduction : model and main results}\label{sec_intro}

In this paper we aim at proving a result of dynamical localization for a random scattering zipper of arbitray size. Let $L\geq 1$ an integer. A scattering zipper is a system obtained by concatenation of elementary unitary diffusion with a fixed number $2L$ of incoming and outgoing channels each. The deterministic scattering zipper operator was introduced in \cite{marin2013scattering} and a random version of this model was introduced in \cite{boumaza2015absence}. We now precise this latter random version of the random scattering zipper adding some hypotheses on the randomness which allow to prove a dynamical localization result for this new model.

Let $U(L)$ be the unitary subgroup of $\mathrm{GL}_{\mathrm{L}}(\C)$ of order $L$. Let 
\begin{equation}\label{def_Omegatilde}    
\tilde{\Omega}=\left( U(L)\times\left\{-1,1\right\}^L\times[0,2\pi]^L\right)^2
\end{equation}
endowed with the probability measure
\begin{equation}\label{def_Ptilde}
\tilde{\mathbb{P}}=\left( \nu_L\otimes(\mathcal{B}(p))^{\otimes L} \otimes(\mathcal{L}eb_{[0,2\pi]})^{\otimes L}\right)^{\otimes 2},    
\end{equation}
\noindent defined on $\tilde{\mathcal{B}}$, the Borel $\sigma$-algebra on $\tilde{\Omega}$ for the usual topology on the Lie group. Here $\nu_L$ denotes the Haar measure on the compact Lie group $U(L)$, $\mathcal{B}(p)$ denotes the Bernoulli distribution of parameter $p\in (0,1)$ and $\mathcal{L}eb_{[0,2\pi]}$ denotes the Lebesgue measure on the interval $[0,2\pi]$.

We then define the product probability space: 
\begin{equation}\label{def_OmegaP}
    (\Omega,\mathcal{B},\mathbb{P})=(\tilde{\Omega}^{\mathbb{Z}},\otimes_{n\in\mathbb{Z}}\tilde{\mathcal{B}},\otimes_{n\in\mathbb{Z}}\tilde{\mathbb{P}}).
\end{equation}

Let $\omega\in \Omega$ and $n\in \Z$. Then $\omega_n\in \tilde{\Omega}$ and we set
\begin{equation}\label{def_omegan}
    \omega_n = (\widetilde{V}^{(n)}_\omega,d^{(n)}_\omega, \theta_\omega^{(n)}, \widetilde{U}^{(n)}_\omega,D^{(n)}_{\omega}, \Theta_\omega^{(n)} )
\end{equation}
with :
\begin{itemize}
    \item[(i)] $\widetilde{U}^{(n)}_\omega\in U(L)$ and $\widetilde{V}^{(n)}_\omega \in U(L)$ ;
    \item[(ii)] $d^{(n)}_\omega=(d^{(n)}_{\omega,1},\hdots,d^{(n)}_{\omega,L})\in \{-1,1\}^L$ which will represent indifferently the $L$-uple or the diagonal matrix whose diagonal elements are the $d^{(n)}_{\omega,1},\hdots,d^{(n)}_{\omega,L}$ and the same for $D^{(n)}_\omega=(D^{(n)}_{\omega,1},\hdots,D^{(n)}_{\omega,L})\in \{-1,1\}^L$ ;
\item[(iii)] $\theta^{(n)}_\omega=(\theta^{(n)}_{\omega,1},\hdots,\theta^{(n)}_{\omega,L})\in [0,2\pi]^L$ which will represent indifferently the $L$-uple or the diagonal matrix whose diagonal elements are the $\theta^{(n)}_{\omega,1},\hdots,\theta^{(n)}_{\omega,L}$ and the same for $\Theta^{(n)}_\omega=(\Theta^{(n)}_{\omega,1} ,\hdots,\Theta^{(n)}_{\omega,L})\in [0,2\pi]^L$.
\end{itemize}

With these notations we define, for every $\omega \in \Omega$ and every $n\in \Z$, the following phases
\begin{equation}\label{def_phasesUV}  
U^{(n)}_\omega=e^{i\Theta^{(n)}_\omega}\widetilde{U}^{(n)}_\omega D^{(n)}_\omega(\widetilde{U}^{(n)}_\omega)^*:=e^{i\Theta^{(n)}_\omega} \widehat{U}^{(n)}_\omega \quad \mbox{ and }\quad  V^{(n)}_\omega=\widetilde{V}^{(n)}_\omega d^{(n)}_\omega(\widetilde{V}^{(n)}_\omega)^* e^{i\theta^{(n)}_\omega}:=\widehat{V}^{(n)}_\omega e^{i\theta^{(n)}_\omega} 
\end{equation}
where the reduced notations $\widehat{U}^{(n)}_\omega$ and $\widehat{V}^{(n)}_\omega$ will be used in the proofs of Theorem \ref{thm_unif_bounded}, Lemma \ref{lem1alphainversible} and Theorem \ref{thm_second_order}. 

Note that for every $\omega\in \Omega$ and every $n\in \Z$, $\widehat{U}^{(n)}_\omega$ and $\widehat{V}^{(n)}_\omega$ are elements of $U(L)\cap H_L(\C)$ where $H_L(\C)$ is the vector space of the Hermitian matrices of size $L\times L$.

\begin{remark}
In view of the definition of the probability space $ (\Omega,\mathcal{B},\mathbb{P})$ one can see $(\widetilde{U}_n^\omega)_{n\in \mathbb{Z}}$ and $(\widetilde{V}_n^\omega)_{n\in \mathbb{Z}}$ as sequences of independent and identically distributed (\emph{i.i.d.} for short) random variables in $U(L)$ of uniform law according to the Haar measure on $U(L)$. The sequences $(d^{(n)}_{\omega,i})_{n\in \mathbb{Z}}$ and $(D_{\omega,j}^{(n)})_{n\in \mathbb{Z}}$ for $i,j\in \llbracket 1, L\rrbracket$ can be considered as sequences of \emph{i.i.d.} random variables in $\{-1,1\}$ of common law the Bernoulli law $\mathcal{B}(p)$ of parameter $p\in (0,1)$ and  the sequences $(\theta_{\omega,i}^{(n)})_{n\in \mathbb{Z}}$ and $(\Theta_{\omega,j}^{(n)})_{n\in \mathbb{Z}}$ for $i,j\in \llbracket1, L\rrbracket$ can be considered as sequences of \emph{i.i.d.} random variables in $[0,2\pi]$ of common law the uniform law on $[0,2\pi]$.

Moreover, all these sequences of random variables are independent from each others. 
\end{remark}

\bigskip

With these notations, we can now introduce our model of random scattering zipper. Consider the random family of unitary operators $\{ \mathbb{U}_\omega \}_{\omega \in \Omega}$ where, for every $\omega \in \Omega$, the operator $\mathbb{U}_\omega$  acts on $\ell^2(\mathbb{Z}, \mathbb{C}^L)$ and is defined by: 
\begin{equation}\label{def_scatzipp1}
 \mathbb{U}_\omega = \mathbb{V}_\omega\mathbb{W}_\omega   
\end{equation}
where
\begin{equation}\label{def_scattzip2}
    \mathbb{V}_{\omega} = \begin{pmatrix}
    \ddots & & & \\
           & S^{(0)}_\omega & & \\
           & & S^{(2)}_\omega & \\
           & & & \ddots \\
\end{pmatrix} \circ s_g^{L}  \quad \mbox{ , }\quad 
\mathbb{W}^{\omega} = \begin{pmatrix}
    \ddots & & & \\
           & S^{(-1)}_\omega & & \\
           & & S^{(1)}_\omega & \\
           & & & \ddots \\
\end{pmatrix}
\end{equation}
and $s_g$ is the shift operator to the left $(v_n)_{n\in \Z} \mapsto  (v_{n+1})_{n\in \Z}$. The $2L\times 2L$ blocks $S^{(n)}_\omega$ are unitary matrices in the unitary group $U(2L)$ of the particular form
\begin{equation}\label{def_scatmat}
S^{(n)}_\omega=S(\alpha, U^{(n)}_\omega, V^{(n)}_\omega) = \begin{pmatrix}
    \alpha & \rho(\alpha)U^{(n)}_\omega \\
    V^{(n)}_\omega \tilde{\rho}(\alpha) & -V^{(n)}_\omega \alpha^* U^{(n)}_\omega
\end{pmatrix},    
\end{equation}
for a fixed $\alpha\in \mathcal{M}_L(\C)$ such that $\|\alpha\| < 1$, and with $\rho(\alpha) = (1 - \alpha\alpha^*)^{\frac{1}{2}}$ and $\tilde{\rho}(\alpha) = (1 - \alpha^*\alpha)^{\frac{1}{2}}$. 

Throughout this article, $\|\cdot \|$ denotes any subordinate norm on $\mathcal{M}_L(\C)$. We introduce the set 
\begin{equation}
\label{eq-U(2L)subset}
\mbox{U}(2L)_{\mathrm{inv}}
\; =\;
\left\{\left. \left( \begin{smallmatrix} \alpha & \beta \\ \gamma & \delta \end{smallmatrix} \right) \in \mbox{U}(2L) \;\right|\; 
\alpha,\gamma,\delta \in \mathcal{M}_{L}(\C) \mbox{ and } \beta \in \mathrm{GL}_{\mathrm{L}}(\C)\; 
\right\}.
\end{equation}
which also has the representation (see \cite{marin2013scattering}):
\begin{equation}
\label{eq-CMVrep}
\mbox{U}(2L)_{\mathrm{inv}} = \left\{\left. S(\alpha,U,V) \in \mbox{\rm U}(2L) \;\right|\; \|\alpha^*\alpha\|<1 \;\mbox{ and }U,V\in\mbox{U}(L) \right\},
\end{equation}
where for any $U,V \in \mathrm{U}(L)$ and any $\alpha \in  \mathcal{M}_{L}(\C)$ with  $\|\alpha^*\alpha\|<1$,
\begin{equation}\label{eq-Smatdef}
S(\alpha,U,V)\; =\; 
\begin{pmatrix} \alpha & \rho(\alpha) U 
\\ 
V\widetilde{\rho}(\alpha) & -V\alpha^*U 
\end{pmatrix}.
\end{equation}

Note that the family $\{ \mathbb{U}_{\omega} \}_{\omega \in \Omega}$ has a very important property of $2\Z$-ergodicity which implies the existence of an almost sure spectrum for this family. See \cite{carmona2012spectral} or \cite{kirsch2007invitation} for a definition of this ergodicity property and for the existence of the almost-sure spectrum.  Note that if $\Sigma$ denotes the almost-sure spectrum of  $\{ \mathbb{U}_{\omega} \}_{\omega \in \Omega}$ , $\Sigma \subset \mathbb{S}^1$ since the operators $\mathbb{U}_{\omega}$ are unitary. The $2\mathbb{Z}$-ergodicity also implies the existence of almost-sure pure point, absolutely continuous and singular continuous spectra.

The block shifting between $\mathbb{V}_{\omega}$ and $\mathbb{W}_{\omega}$ in the  product defining $\mathbb{U}_{\omega}$ allows us to consider the scattering zipper  model as a matrix-valued version of the Blatter-Browne model introduced in \cite{blatter1988zener} to understand Zener tunneling effect in diodes. It can also be considered as a matrix-valued version of the Chalker-Coddington model on a band studied in \cite{asch2010localization}. Indeed, $\mathbb{U}_{\omega}$ has the same band structure as that found in the Blatter-Browne and quasi-one-dimensional Chalker-Coddington models, but with an arbitrarily large diagonal bandwidth. Note that since $\mathbb{U}_{\omega}$ is a unitary band operator, the general results of \cite{bourget2003spectral} apply to it. To avoid any confusion, note that the zipper scattering model is very different from the Chalker-Coddington model over the entire $\Z^2$ network, since it is one-dimensional rather than two-dimensional. The scattering zipper model can also be seen as a version of CMV matrices with matrix coefficients. CMV matrices are the unitary analog of Jacobi matrices and were originally introduced in the study of orthogonal polynomials on the unit circle (see \cite{simon2005orthogonal} for a comprehensive review of this vast subject). On the other hand, Schulz-Baldes studied in \cite{SB07} a generalization of Jacobi matrices with matrix coefficients. The scattering zipper was introduced by Marin and Schulz-Baldes in \cite{marin2013scattering} with the aim of defining a unitary analog to Jacobi matrices with matrix coefficients that shares with CMV matrices their spectral properties and a simple way of representing their coefficients. In particular, the block factorization of CMV matrices into two diagonal operators (see \cite{CMV03}) is used to define the scattering zipper, without having to resort to an interpretation in terms of orthogonal polynomials on the unit disk.

We can summarize this discussion in a diagram in which the horizontal arrows represent the transition from a scalar-valued operator to a matrix-valued operator, and the vertical arrows represent the transition from a self-adjoint model to its unitary analogue.
$$\begin{array}{ccl}
\mbox{Jacobi matrices} & \longrightarrow & \mbox{Jacobi matrices with matrix coefficients} \\
\downarrow		&		& \quad \quad \quad \quad \downarrow \\
\mbox{CMV matrices}	& \longrightarrow & \mbox{Scattering zipper} 
\end{array}$$

\bigskip

Let us now present the main results obtained in this paper. Let $\{e_k\}_{k\in \mathbb{Z}}$ be the canonical basis of $\ell^2(\mathbb{Z})\otimes \C^L$. For $i,j\in \Z$, we set $e_{\{i, j\}} := e_{iL+j}$, which corresponds to the j-th component in the i-th L-block.  We now state the main result of this article.

\begin{theorem}\label{thm_DL}
There exists \(r_0 > 0\) such that for every \(\alpha \in \mathrm{GL}_{\mathrm{L}}(\mathbb{C})\) with \(\|\alpha\| \leq r_0\), there exists \(C_{r_0} > 0\) and \(b > 0\) such that for all \(\{k,p\}\) and \(\{l,q\}\) in \(\mathbb{Z} \times \llbracket1, L\rrbracket\),
\begin{equation}\label{eq_thm_DL}
\mathbb{E}\left[\sup_{n \in \mathbb{Z}}\left|\left\langle e_{\{k,p\}}, (\mathbb{U}_{\omega})^{n} e_{\{l,q\}} \right\rangle\right|\right] \leq C_{r_0} e^{-b |k-l|}.
\end{equation}
\end{theorem}
The estimate \eqref{eq_thm_DL} means that the family $\{\mathbb{U}_\omega\}_{\omega\in\Omega}$ satisfies the condition of dynamical localization as defined for example in \cite{hamza2009dynamical}. This property is dynamic in nature and follows the evolution of wave packets over discrete time $n\in \Z$. It tells us that the solutions of the Schr\"odinger equation are localized in space in the vicinity of their initial position and this, uniformly over time. This reflects the absence of quantum transport.

The proof of Theorem \ref{thm_DL} is based upon a number of intermediate results. First, following the result of positivity of the Lyapunov exponents on the unit circle \(\mathbb{S}^1\) obtained by Boumaza and Marin in \cite{boumaza2015absence}, we prove their strict positivity on an annulus 
\begin{equation}\label{def_Sepsilon}
  \mathbb{S}_{\epsilon}:= \{z \in \mathbb{C}; 1-\epsilon < |z| < 1+\epsilon\}  
\end{equation}
for some \(\epsilon\in (0,1]\). The proof is done by combining the positivity of the Lyapunov exponents on the unit circle and their continuity on \(\C\setminus \{ 0 \}\).

For $z\in \C\setminus \mathbb{S}^1$ let $G_{\omega}(z)$ be the resolvent at $z$ of the operator $\mathbb{U}_{\omega}$ and let $G_{\omega}(z,\cdot,\cdot)$ be its Green kernel. For $a, b \in \mathbb{Z}\cup \{\pm \infty\}$, $a < b$, we also denote by $G_{\omega}^{[a,b]}(z)$ the resolvent of the restriction of the operator $\mathbb{U}_{\omega}$ to the interval $[a,b]$ and by $G_{\omega}^{[a,b]}(z,\cdot,\cdot)$ its Green kernel : 
\begin{equation}\label{def_resolvent_ab}
    G_\omega^{[a,b]}(z) = \left(\mathbb{U}_\omega^{[a,b]} - z\right)^{-1},
\end{equation}
and for $k, l \in \mathbb{Z}$,
\begin{equation}\label{def_resolvent_ab_coef}
       G_\omega^{[a,b]}(z, k, l) = \langle e_k, \left(\mathbb{U}_\omega^{[a,b]} - z\right)^{-1} e_l \rangle.
\end{equation}
The precise definition of the restricted operator $\mathbb{U}_\omega^{[a,b]}$ is given at Section \ref{sec_finitezippers}.

We study the fractional moments of $G_{\omega}^{[a,b]}(z,\cdot,\cdot)$ for $z\in \C \setminus \mathbb{S}^1 $ and we first prove that they are uniformly bounded.

\begin{theorem}\label{thm_unif_bounded}
For every \(s \in (0, \frac14)\) there exists \( C(s)>O\) such that, for every \(z\in \C \setminus \mathbb{S}^1\), every $a, b \in \mathbb{Z}\cup \{\pm \infty\}$, $a < b$,  and every $k,l\in \mathbb{Z}$ such that \(|k-l| > 4\), 
  \begin{equation}\label{eq_thm_unif_bounded}
  \mathbb{E}\left(\|G_{\omega}^{[a,b]}(z,k,l)\|^s\right) \leq C(s).
  \end{equation}  
\end{theorem}

Once we have obtained this uniform bound, combining it with the positivity of the smallest Lyapunov exponent and proving estimates on the blocks of the products of transfer matrices, we prove the exponential decay of some blocks of the restricted Green kernel to suitable intervals.

\begin{theorem}\label{thm_exp_decay_even_odd}
There exists $r_0>0$, $\epsilon_0>0$, \(s_0 \in (0,1)\), $p_0>1$, \(C_{s_0,r_0}> 0\) and \(\gamma > 0\) such that, for every $\alpha\in GL_L(\mathbb{C})$ with $\|\alpha\|\leq r_0,$ every $s\in (0,s_0]$ and every $\epsilon\in (0,\epsilon_0]$,
\begin{equation}\label{eq_exp_decay_reduced}
    \mathbb{E}\left(\|G_{\omega}^{[2n,2m+1]}(z, 2n, 2m+1)\|^s\right) \leq C e^{-\gamma |m-n|}
\end{equation} 
for every \(z \in \mathbb{S}_\epsilon \setminus\mathbb{S}^1 \) and for every $m$ and $n$ in $\mathbb{Z}$ such that \(|m-n| > p_0\).
\end{theorem}
Theorem \ref{thm_exp_decay_even_odd} is central in the proof of Theorem \ref{thm_DL}. Unlike traditional approaches, like the one used by Hamza, Joye, and Stolz (see \cite{hamza2009dynamical}) which relies on expressing the Green's function in terms of eigenfunctions with particular boundary conditions to obtain exponential decay through transfer matrices, we must adopt a different method. The block configuration of our model does not allow us to directly apply this strategy since the expression for the Green kernel obtained in this case is too complicated to estimate (see \cite{damanik2008analytictheorymatrixorthogonal} for such an expression of the Green kernel). Therefore, we use the Schur complement as an alternative. This method allows us to identify a term whose norm expectation must be bounded to demonstrate the decay of the reduced case.

In order to be able to reduce our analysis to the blocks for which we have the exponential decaying \eqref{eq_exp_decay_reduced}, we prove two reduction results. We start by showing that it suffices to deal with even-sized scattering zipper operators.

 \begin{proposition}\label{prop_reduc_paire}
    Assume $\alpha\in \mathrm{GL}_{\mathrm{L}}(\C) $ is such that $\|\alpha\|<1$. Let $s\in (0,\frac{1}{4})$, $\epsilon\in (0,1)$ and let $k,l\in \mathbb{Z}$ such that $|k-l|>4$. There exists $C(s,\epsilon)>0 $ such that:
\begin{equation}\label{ineq_resolv_temrespairs}
    \mathbb{E}\left( \|G_{\omega}^{[a,b]}(z,k,l) \|^s\right)^2\leq C(s,\epsilon) \sum_{i,j=0}^1 \mathbb{E} (\|G_{\omega}^{[a,b]}(z,2n+2i,2m+1+2j)  \|^{4s})^{\frac{1}{2}},
\end{equation}
for all $z \in \mathbb{S}_\epsilon$ and $n,m$ such that $k\in \{2n,2n+1\}$ and $l\in \{2m,2m+1\}$. 
  \end{proposition}
\noindent The proof involves bounding the norm of the "even" blocks $\|G_{\omega}^{[a,b]}(z,2n,2m) \|$ and~$\|G_{\omega}^{[a,b]}(z,2n+1,2m) \|$ by the norms of the "odd" blocks $\|G_{\omega}^{[a,b]}(z,2n,2m+1) \|$ and~$\|G_{\omega}^{[a,b]}(z,2n+1,2m+1) \|$. A necessary condition for establishing this is the invertibility of $\alpha$. This invertibility is necessary to prove the following lemma.

\begin{lemma}\label{lem1alphainversible}
 If $\alpha\in \mathrm{GL}_{\mathrm{L}}(\mathbb{C})$, then for every $\epsilon >0$, for every $n\in \mathbb{Z}$ and for every $z\in\mathbb{S}_\epsilon$, the matrix 
     $\alpha+zV^{(n)}_\omega\alpha U_\omega^{(n)}$ is invertible almost surely, and for every $ s\in (0,1)$, there exists $C(s)>0$ such that
\begin{equation}
\forall n\in \mathbb{Z}, \ \ \mathbb{E}\left( \| (\alpha+ z V_\omega^{(n)}\alpha U_\omega^{(n)})^{-1}\|^s\right)\leq C(s).
\end{equation}     
\end{lemma}
Once we get Lemma \ref{lem1alphainversible}, using H\"older inequality, one gets Proposition \ref{prop_reduc_paire}.

The second reduction result show that it suffices to control the Green kernel of the restricted operator to some finite interval in order to control the Green kernel of $\mathbb{U}_\omega$.

\begin{proposition}\label{casfini}
    Let  $s \in (0,\frac12)$ and $\epsilon >0$. One has 
\begin{equation}\label{eq_lemmacasfini}
\mathbb{E}\left( \|G_{\omega}(z,k,l) \|^s\right)^2 \leq C(s)\mathbb{E} \left( \|G_{\omega}^{[k,l]}(z,k,l)^{2s}\|\right) 
\end{equation}
for every $z \in \mathbb{S}_\epsilon\setminus\mathbb{S}^1$ and every $k,l\in\mathbb{Z}$ such that $|k-l|>4$.
\end{proposition}

The proof of Lemma \ref{casfini} relies on the geometric resolvent identity.

Combining Proposition  \ref{prop_reduc_paire} and Proposition \ref{casfini}, as well as Theorem \ref{thm_exp_decay_even_odd} on the exponential decay of the reduced case, we obtain the exponential decay of the fractional moments.

\begin{theorem} \label{thm_greengeneral}
 There exists $r_0>0$,\ $s\in(0,1)$, $\epsilon_0>0$, $C_{s,r_0}>0$ and $\gamma>0$, such that for every $\alpha \in \mathrm{GL}_{\mathrm{L}}(\mathbb{C}) $  with~$\|\alpha\|< r_0$, 
\begin{equation}\label{eq_thm_greengeneral}
     \mathbb{E}\left(\| G_\omega(z,k,l)\|^s\right)\leq C_{s,r_0}e^{-\gamma |k-l|}
\end{equation}  
 for every $\epsilon\in (0,\epsilon_0]$, every $k,l\in \mathbb{Z}$ and every $z\in \mathbb{S}_\epsilon \setminus \mathbb{S}^1$.
\end{theorem}

The proof of Theorem \ref{thm_DL} follows from Theorem \ref{thm_greengeneral} and from the following estimate on the moments of order two of the coefficients of the resolvent.

\begin{theorem}\label{thm_second_order}
There exists $\epsilon_0>0$, $r_0>0$, \(C_{r_0}> 0\) and \(\gamma > 0\) such that for every $\epsilon\in(0,\epsilon_0]$ and every $z\in\mathbb{S}_\epsilon\setminus\mathbb{S}^1,$
for every $\alpha\in GL_L(\mathbb{C})$ with $\|\alpha\|\leq r_0,$ and every $\{k, p\}$ and $\{l, q\}$ in $\mathbb{Z} \times \llbracket 1, L \rrbracket$:  
  $$\mathbb{E}\left( (1 - |z|^2) \left|  \langle e_{\{k, p\}}| \left(\mathbb{U}_\omega - z\right)^{-1} e_{\{l, q\}} \rangle\right|^2 \right) \leq C_{r_0} e^{-\gamma(k-l)}$$
\end{theorem}

Theorem \ref{thm_second_order} is a consequence of Theorem \ref{thm_greengeneral} using second order perturbation theory.

\begin{remark}
The parameter \(\epsilon_0\) is chosen to guarantee strictly positive Lyapunov exponents over the annulus \(\mathbb{S}_{\epsilon_0}\) and  is introduced in Corollary \ref{corollaire_pos_Lyap_couronne}.
\end{remark}

Results of strict positivity of the positive Lyapunov exponent are already known for models of unit-band operators (see \cite{bourget2003spectral,hamza2006localization}). Bourget, Howland and Joye's article \cite{bourget2003spectral} was one of the first to present the study of unitary models with one diffusion channel, and was followed by other articles by Alain Joye on the subject, such as \cite{J04} and \cite{J05}. For a synthetic presentation of these results and the associated physical models, we refer to Alain Joye's text, \cite{J11}. Note that the unitary models considered in these first articles are all with scalar coefficients, and not with matrix coefficients as in the case of the scattering zipper.

One can see the random  scattering zipper model as a unitary version of the quasi-one-dimensional Anderson model. There is already a unitary version of the scalar-valued Anderson model for which Hamza, Joye and Stolz have proven dynamical localization in \cite{hamza2009dynamical}. The random scattering zipper can therefore also be seen as a quasi-one-dimensional version of the unitary Anderson model. This is why, as we will see in the rest of the paper, it is possible to follow the strategy of the proof of the dynamical localization result of \cite{hamza2009dynamical}. Based upon the Fractional Moments Method  (see \cite{aizenman2015random} for an overview of this method), this strategy is also the one used in the proof of localization in \cite{asch2010localization,asch2012dynamical} for the Chalker-Coddington model.

However, it would be possible to consider a multiscale analysis approach to study this unitary model, as Cedzich and Werner did in \cite{CW21} in the context of unitary quantum walks. This would then make it possible to address the question of localization for a random scattering zipper model in which the unitary phases would reveal a much more singular randomness than a uniform Haar law. A last possible approach would be to adapt to the quasi-one-dimensional case the proofs of localization results for CMV matrices in \cite{zhu2021localization, bucaj2019localization} which are based on large deviations inequalities.

\bigskip

The rest of the paper is devoted to the proof of the results we have just staten in this introductory section. In Section \ref{sec_transmat_lyap} we introduce the transfer matrices associated to the quasi-one-dimensional model $\{ \mathbb{U}_{\omega}\}_{\omega\in \Omega}$ of unitary type (see \cite{boumaza2023localization}). We also define the Lyapunov exponents associated to the sequence of transfer matrices and we prove their continuity, hence their positivity on some annulus $\mathbb{S}_{\epsilon}$. Section \ref{sec_exp_decay_general} is the core of the article. In this section we prove the exponential decay of the fractional moments of the Green kernel. First, we show that it suffices to prove this exponential decay for some particular blocks of the resolvent by proving Proposition \ref{prop_reduc_paire} and Proposition \ref{prop_reduc_suitable_volume}. To prove the exponential decay of the  fractional moments of these particular blocks we look precisely at the behavior of the four $L\times L$ blocks of the products of transfer matrices and prove boundedness of some "Wronskian" involving these blocks. Finally in Section \ref{sec_DL}, we prove the exponential decay of moments of order $2$ and we deduce from the spectral theorem for unitary operators the proof of the dynamical localization for $\{ \mathbb{U}_{\omega} \}_{\omega \in \Omega}$.


\section{Transfer matrices and Lyapunov exponents}\label{sec_transmat_lyap}

\subsection{Transfer matrices.}

Using the transfer matrix formalism, we reduce the study of the asymptotic behavior of a solution $\phi$  of 
\begin{equation}\label{eq_vp_scatzip}
\mathbb{U}_{\omega}\phi=z\phi,\quad \mbox{for}\ z\in \C,  
\end{equation}
to the asymptotic behavior of a product of random matrices. 

In the unitary setting, the role played by the symplectic group for the quasi-one-dimensional models of Schr\"odinger type will be played by the group $\mbox{U}(L,L)$ of the matrices of size $2L \times 2L$ which preserve the form $\mathcal{L} = \left(\begin{smallmatrix}
 \one & 0 \\
0 & -\one
\end{smallmatrix} \right)$ in the sense that $T$ is in $\mbox{U}(L,L)$ if and only if $T^* \mathcal{L} T = \mathcal{L}$. 

To compute the transfer matrices we proceed as follows. Instead of looking at the input-output relations of the scattering matrix $S_\omega^{(n)}$, we look for a new matrix which allows to express $\left( \begin{smallmatrix}           \phi_{n+1} \\ \psi_{n+1}  \end{smallmatrix}\right)$ in terms of $\left( \begin{smallmatrix}     \phi_n \\
\psi_n   \end{smallmatrix} \right)$ for $\phi$ a solution of \eqref{eq_vp_scatzip} and $\psi= \mathbb{W}_{\omega} \phi$. This is done by transforming the scattering matrices $S_\omega^{(n)}$ belonging to $\mbox{U}(2L)_{\mathrm{inv}}$ into elements of $U(L,L)$ via the bijection:
$$\varphi\ :\ \begin{array}{ccl}
            \mbox{U}(2L)_{\mathrm{inv}} & \to & \mathrm{U}(L,L) \\[2mm]
	    \left( \begin{smallmatrix}
	    \alpha & \beta \\
	    \gamma & \delta
	    \end{smallmatrix} \right) & \mapsto & \left(\begin{smallmatrix}
	    \gamma-\delta \beta^{-1} \alpha & \delta \beta^{-1} \\
	    -\beta^{-1} \alpha & \beta^{-1}
	    \end{smallmatrix} \right)
            \end{array}$$
Let $z\in \C$. We have the following relations, proven in \cite{marin2013scattering}:
\begin{equation}\label{eq_trans_mat_odd_even}
\forall n\in \Z,\  \left( \begin{smallmatrix}     \phi_{2n} \\ \psi_{2n} \end{smallmatrix} \right) = \varphi(z^{-1} S^{(2n)}_\omega)  \left( \begin{smallmatrix} \psi_{2n-1} \\
\phi_{2n-1} \end{smallmatrix} \right) \quad \mbox{and} \quad \left( \begin{smallmatrix} \psi_{2n+1} \\ \phi_{2n+1}    \end{smallmatrix} \right) = \varphi( S^{(2n+1)}_\omega)  \left( \begin{smallmatrix}  \phi_{2n} \\ \psi_{2n}  \end{smallmatrix} \right).
\end{equation}

These relations lead to introduce the application $T(z,\cdot ): \Omega \to \mathrm{GL}_{2L}(\C)$, 
\begin{equation}\label{eq_def_trans_mat}
 \forall \omega \in \Omega,\ T(z,\omega) =  \left(\begin{smallmatrix} V_\omega^{(0)}  & 0 \\ 0 & (U_\omega^{(0)})^* \end{smallmatrix} \right)
T_0(z) \left( \begin{smallmatrix} V^{(1)}_\omega  & 0 \\ 0 & (U^{(1)}_\omega)^* \end{smallmatrix} \right) T_1
\end{equation}
with
$$
T_0(z) = \left( \begin{smallmatrix}
z^{-1} (\widetilde{\rho}(\alpha))^{-1} & (\widetilde{\rho}(\alpha))^{-1}\alpha^* \\
\alpha (\widetilde{\rho}(\alpha))^{-1} & z (\rho(\alpha))^{-1}
\end{smallmatrix} \right) \quad \mbox{ and }\quad T_1 = \left( \begin{smallmatrix}
 (\widetilde{\rho}(\alpha))^{-1} & (\widetilde{\rho}(\alpha))^{-1}\alpha^* \\
\alpha (\widetilde{\rho}(\alpha))^{-1} &  (\rho(\alpha))^{-1}
\end{smallmatrix} \right).
$$
We define the 2-shift transformation $\tau_2:\Omega\mapsto\Omega$ by:
$$ \forall\omega\in \Omega, \forall n\in \mathbb{Z},\ (\tau_2(\omega))_n=\omega_{n+2}.$$
Then, $\tau_2$ is ergodic on $(\Omega,\mathcal{B},\mathbb{P})$ and
\begin{equation}\label{eq_link_trans_scat}
 \forall \omega \in \Omega,\ \forall z\in \C,\ \forall n\in \Z,\ T(z,\tau_2^{n}(\omega))= \varphi(z^{-1}S^{(2n)}_\omega)\cdot \varphi(S^{(2n-1)}_\omega).
\end{equation}
The matrix $T(z,\tau^{n}(\omega))$ is the n-th transfer matrix associated to $\{\mathbb{U}_{\omega}\}_{\omega \in \Omega}$. Then $(T(z,\tau^{n}(\omega)))_{n\in \Z}$ is a sequence  of \emph{i.i.d.} random matrices in $\mathrm{GL}_{2L}(\C)$ because of the definition of the probability $\mathbb{P}$ as a tensor product of probability measures. 
\bigskip

Note that for $z\in \mathbb{S}^1$, $T(z,\tau^{n}(\omega))\in U(L,L)$.

\subsection{Lyapunov exponents.}

The transfer matrices $T(z,\cdot)$ generate a cocycle $\Phi(z,\cdot,\cdot): \Omega \times \Z \to \mathrm{GL}_{2L}(\C) $ on the ergodic dynamical system $(\Omega, \mathcal{B}, \mathbb{P}, (\tau^n)_{n \in \Z} )$ defined by   
$$\forall \omega \in\Omega,\ \forall n\in \Z,\ \Phi(z,\omega,n) \;=\; \left\lbrace 
\begin{array}{lcl}
T(z,\tau_2^{n-1}(\omega)) \dots T(z,\omega) & \mbox{ if } & n >0 \\
\mathrm{I}_{2\mathrm{L}} & \mbox{ if } & n=0 \\
(T(z,\tau_2^{n}(\omega)))^{-1} \dots (T(z,\tau_2^{-1}(\omega)))^{-1} & \mbox{ if } & n <0.
\end{array} \right.
$$

From this cocycle we define the Lyapunov exponents associated to the ergodic family $\{ \mathbb{U}_{\omega} \}_{\omega \in \Omega}$. Let $z\in \C$. For $\mathbb{P}$-almost every $\omega \in \Omega$, the following limits exists and are equal:
\begin{equation}\label{eq_lim_Lyap}
 \Psi(z,\omega)  :=  \lim_{n\to +\infty} ((\Phi(z,\omega,n))^*\Phi(z,\omega,n))^{1/2n}   =  \lim_{n\to -\infty} ((\Phi(z,\omega,n))^*\Phi(z,\omega,n))^{1/2|n|}.
\end{equation}
For every $k\in \llbracket 1, 2L\rrbracket$, let $\lambda_k(z,\omega)$ the $k$-th  eigenvalue of $\Psi(z,\omega)$, the eigenvalues being ordered in increasing order. There are then real numbers $\lambda_k(z) \geq 0$ such that, for $\mathbb{P}$-almost every $\omega \in \Omega$, $\lambda_k(z,\omega)=\lambda_k(z)$. We then define the Lyapunov exponents associated to the ergodic family $\{ \mathbb{U}_{\omega} \}_{\omega \in \Omega}$ as being the real numbers $\gamma_k(z)$ defined by : 
$$\forall z\in \C,\ \forall \ k\in  \llbracket 1, 2L\rrbracket,\ \gamma_k(z) := \log( \lambda_k(z) ).$$

Note that when $z\in \mathbb{S}^1$, the fact that the transfer matrices belong to $U(L,L)$ implies a symmetry relation on the Lyapunov exponents which is the same as in the case where the transfer matrices are in the symplectic group:
\begin{equation}\label{lyap_sym_relation}
\forall k\ \in  \llbracket 1, L\rrbracket,\ \gamma_{2L-k+1}(z)=-\gamma_{k}(z).    
\end{equation}

In \cite{boumaza2015absence}, we proved, for a random scattering zipper which can be seen as a particular case of \eqref{def_scatzipp1} that, 
\begin{equation}\label{eq_pos_Lyap_BM15}
\forall z\in \mathbb{S}^1,\     \gamma_1(z)>  \gamma_2(z)> \cdots >  \gamma_L(z)> 0.
\end{equation}
Using Kotani's theory (see \cite{boumaza2015absence} which adapts results of  \cite{kotani1988stochastic}), this implies the absence of almost-sure  absolutely continuous spectrum of $\{\mathbb{U}_{\omega}\}_{\omega\in \Omega}$.

The proof of \eqref{eq_pos_Lyap_BM15} relies on the study of the so-called Furstenberg group associated to $\{\mathbb{U}_{\omega} \}_{\omega\in \Omega}$. This is the group generated by the common law of the transfer matrices. Let $z\in \C$. If $\mu_z$ is the common law of all the transfer matrices $T(z,\tau^n(\omega))$, we set
\begin{equation}
 G_{\mu_z} \; = \; \overline {\langle  \text{supp} \mu_z  \rangle} \subset  \mathrm{GL}_{2L}(\C). 
\end{equation}
The closure is taken for the topology on $ \mathrm{GL}_{2L}(\C) $ which is induced by the usual topology on $\mathcal{M}_{2L}(\C)$.

In \cite{boumaza2015absence} we proved that for every $z\in \mathbb{S}^1$, $G_{\mu_z}=U(L,L)$. Actually, the algebraic construction made in the proof of \cite[Proposition 2]{boumaza2015absence} is valid for any $z\in \C$ and we get
\begin{equation}\label{eq_ULL_subset_G}
   \forall z \in \mathbb{C}, \ U(L,L) \subseteq G_{\mu_z}.
\end{equation}
Hence one gets, using a similar proof as the one of \cite[Theorem 1]{boumaza2015absence}, that
\begin{equation}\label{eq_separate_Lyap}
\forall z\in \C,\     \gamma_1(z)>   \cdots >  \gamma_L(z)> \gamma_{L+1}(z)> \cdots > \gamma_{2L}(z).
\end{equation}
But for $z\in \C \setminus \mathbb{S}^1$ the transfer matrices are no longer in $U(L,L)$ and we no longer have the symmetry relation \eqref{lyap_sym_relation}. Hence there is no reason for $\gamma_L(z)$ to be postive if $z$ is not in $\mathbb{S}^1$. Still, what remains true is following equality:
\begin{equation}\label{mat_trans_relation_inverse}
     \forall n\in \Z,\ \forall \omega\in \Omega,\ \forall z\in \mathbb{C}^*: \ (T_{\omega}^{(n)}(z))^{-1}=\mathcal{L} \left(T_{\omega}^{(n)}(\Bar{z}^{-1})\right)^*\mathcal{L}.
\end{equation}
The proof of \eqref{eq_separate_Lyap} is based upon the Cayley transform, whose matrix in the canonical basis of $\C^{2L}$ is
\begin{eqnarray}\label{def_Cayley}
C=\frac{1}{\sqrt{2}} \left( \begin{matrix}
                              \one & -i\one \\
			      \one & i\one
                              \end{matrix}  \right)\in \mathcal{M}_{2L}(\C)
\end{eqnarray}
and which maps the group $U(L,L)$ onto the complex symplectic group $\SpC$. It also use the following transformation which separates the real and imaginary parts of a matrix with complex coefficients and place them in blocks: 
 \begin{equation}\label{def_pi_C_R}
    \pi\ :\ \begin{array}{cll}
           \mathcal{M}_{2L}(\C) & \to & \mathcal{M}_{4L}(\R) \\[2mm]
	    A+iB & \mapsto & \left( \begin{smallmatrix}
				      A & -B \\
				      B & A
				      \end{smallmatrix} \right) .
				\end{array}
\end{equation}

It allows to apply directly the results of \cite{bougerol2012products} to get separability of the Lyapunov exponents and their integral representation from the properties of $p$-contractivity and $L_p$-strong irreducibility of the Furstenberg group as defined in \cite{bougerol2012products}.

\subsection{Continuity and strict positivity of the Lyapunov exponents.}

Among the key elements for the development of the fractional moment method, the continuity and positivity of Lyapunov exponents occupy a prominent place. As we have just explained, Theorem 1 of \cite{boumaza2015absence} gives us the positivity of these exponents, but this assertion remains limited to the circle \(\mathbb{S}^1\).

In order to manipulate the resolvent efficiently, it is necessary to choose a spectral parameter \(z\) located outside of \(\mathbb{S}^1\). The objective of this subsection is therefore twofold. Firstly, we will demonstrate the continuity of Lyapunov exponents on $\C \setminus \{0\}$. Secondly, using \cite[Theorem 1]{boumaza2015absence}, we will highlight the positivity of these exponents in some  annulus $\mathbb{S}_{\epsilon}$.

Let $\omega\in \Omega$, $n\in \Z$ and $z\in \C$. In order to simplify the further computations, from now, we will denote by $T_{\omega}^{(n)}(z):=T(z,\tau_2^{n}(\omega))$ the $n$-th transfer matrix. 

We start by giving upper and lower bounds on the norm of $\rho(\alpha)$ which will be used throughout the rest of the article.

\begin{lemma}\label{lem_estime_rho}
Let \(L \geq 1\). Let \(\|.\|\) be a subordinate norm on \(\mathcal{M}_L(\mathbb{C})\). For \(\alpha \in \mathcal{M}_L(\mathbb{C})\) such that \(\|\alpha\| < 1\) and \(\rho(\alpha) =(\one-\alpha\alpha^*)^{\frac{1}{2}}\), the following inequalities hold:
\begin{itemize}
  \item[1.] \(\|{\rho(\alpha)}\| \geq \sqrt{1-\|\alpha\|^2} \).
  \item[2.] \(\|{\rho(\alpha)}\|\leq 2-\sqrt{1-\|\alpha\|^2} \).
  \item[3.] \(\|{\rho(\alpha)}^{-1}\| \leq \frac{1}{{ \sqrt{1-\|\alpha\|^2}}}\).
  \item[4.]  \(\|{\rho(\alpha)}^{-1}\| \geq \frac{1}{{ 2-\sqrt{1-\|\alpha\|^2}}}\).
    \end{itemize}
and the same for \(\tilde{\rho}(\alpha) =(\one-\alpha^*\alpha)^{\frac{1}{2}}\).
    \end{lemma}
    
   \begin{proof}
Point $(1)$ follows from the triangle inequality, points $(2)$ and $(3)$ are direct consequences of the power series expansion of $z\mapsto \sqrt{1-z}$ and $z\mapsto \frac{1}{\sqrt{1-z}}$ and the sub-multiplicativity of the matrix norm. Finally, to prove $(4)$ it suffices to note that for a subordinate norm, we have  
\(\frac{1}{\|A^{-1}\|} \leq \|A\|\). The proof is identical for $\tilde{\rho}(\alpha)$.
\end{proof}

We now prove several estimates on the transfer matrices which are ingredients of the proof of the continuity of the Lyapunov exponents.

\begin{lemma}\label{lemma1_cont_Lyap}
For any $\epsilon\in (0,1)$, there exists $C_1=C_1(\alpha, \epsilon)>0$ such that:
\begin{equation}\label{eq_lemma1_cont_Lyap}
\forall z\in \mathbb{S}_\epsilon,\ \forall n\in \mathbb{Z},\ \forall \omega\in \Omega, \  \| T_{\omega}^{(n)}(z) \|\leq C_1.
\end{equation}    
\end{lemma}

\begin{proof}
For simplicity, we denote $\rho=\rho(\alpha)$ and $\widetilde{\rho}=\widetilde{\rho}(\alpha)$.
    $$T_{\omega}^{(n)}(z) =\begin{pmatrix}
   V^{(2n)}_\omega & 0\\
    0 & (U^{(2n)}_\omega)^*
    \end{pmatrix}\begin{pmatrix}
     z^{-1}\widetilde{\rho}^{-1} & -\widetilde{\rho}^{-1}\alpha^*\\
    -\alpha\widetilde{\rho}^{-1} & z\rho^{-1}\end{pmatrix}\begin{pmatrix}
    V^{(2n-1)}_\omega & 0\\
    0 & (U^{(2n-1)}_\omega)^*
    \end{pmatrix}\begin{pmatrix}
    \widetilde{\rho}^{-1} & -\widetilde{\rho}^{-1}\alpha^*\\
    -\alpha\widetilde{\rho}^{-1} & \rho^{-1}\end{pmatrix}.$$
    Thus, for $z\in \mathbb{S}_\epsilon$,
\begin{align*}
       \|T_{\omega}^{(n)}(z) \| & \leq \left\| \left(\begin{smallmatrix}
     z^{-1}\widetilde{\rho}^{-1} & -\widetilde{\rho}^{-1}\alpha^*\\
    -\alpha\widetilde{\rho}^{-1} & z\rho^{-1}\end{smallmatrix}\right) \right\|.\left\| \left(\begin{smallmatrix}
    \widetilde{\rho}^{-1} & -\widetilde{\rho}^{-1}\alpha^*\\
    -\alpha\widetilde{\rho}^{-1} & \rho^{-1}\end{smallmatrix}\right) \right\| \\
    & \leq (|z^{-1}\||\widetilde{\rho}^{-1}\| +\|\widetilde{\rho}^{-1}\alpha^*\|+
    \|\alpha\widetilde{\rho}^{-1}\|+|z\||\rho^{-1}\|).(\| 
    \widetilde{\rho}^{-1}\|+\|\widetilde{\rho}^{-1}\alpha^*\|
    +\|\alpha\widetilde{\rho}^{-1}\|+\|\rho^{-1}\|) \\
    &  \leq \left( \frac{1}{1-\epsilon}\|{\rho}^{-1}\|+ 2 \| {\rho}^{-1}\|. \|\alpha\| +(1+\epsilon)\|\rho^{-1}\|\right)\left(\|{\rho}^{-1}\|+ 2 \| {\rho}^{-1}\| \|\alpha\| +\|\rho^{-1}\|\right).
\end{align*}    
We define $c_\epsilon :=\max\{\frac{1}{1-\epsilon}; 1+\epsilon\}>1$. Then, since $\frac{1}{c_\epsilon} <1$,
\begin{align*}
 \|T_{\omega}^{(n)}(z) \| &      \leq \left( c_\epsilon\|{\rho}^{-1}\|+ 2\| {\rho}^{-1}\| .\|\alpha\|+ c_\epsilon\|\rho^{-1} \|\right)^2\\
 &  \leq c_\epsilon^2 \|\rho^{-1}\|^2\left( 2+  2\|\alpha\|\right)^2 \\
 &  \leq 4c_\epsilon^2 \frac{1}{{1-\|\alpha\|^2}}\left( 1+  \|\alpha\|\right)^2:=C_1(\alpha,\epsilon).
\end{align*}
using Lemma \ref{lem_estime_rho} at the last inequality. 
\end{proof}
\begin{lemma} \label{lemma2_cont_Lyap}
  Let $\epsilon\in (0,1).$ Let $p\in \llbracket 1,L \rrbracket$. There exists $C_2=C_2(p,\alpha,\epsilon)>0$ such that: 
\begin{equation}\label{eq_lemma2_cont_Lyap}
\forall z\in \mathbb{S}_\epsilon, \forall n\in \mathbb{Z}, \forall \omega\in \Omega,\  \ \left\| \wedge^p T_{\omega}^{(n)}(z) \right\|\leq C_2.
\end{equation}
\end{lemma}

\begin{proof}
The inequality \eqref{eq_lemma2_cont_Lyap} comes from \eqref{eq_lemma1_cont_Lyap} and the general fact: if  $M\in GL_{2L}(\mathbb{C})$ then  for any $p$,  $\| \wedge^p M\| \leq \| M\|^p$ (see \cite{bougerol2012products}). Note that here, $\wedge^p$ denotes the $p^{\mathrm{th}}$ exterior power.
\end{proof}

\begin{lemma}\label{lemma3_cont_Lyap}
Let $\epsilon \in (0,1)$. There exists $C_3=C_3(\alpha,\epsilon)>0$ such that:  
\begin{equation}\label{eq_lemma3_cont_Lyap}
  \forall n\in \mathbb{Z}, \forall \omega\in \Omega, \forall z_1, z_2 \in \mathbb{S}_\epsilon,\ \| T_{\omega}^{(n)}(z_1)-  T_{\omega}^{(n)}(z_2) \|\leq C_3 \mid z_1-z_2\mid.   
\end{equation}
\end{lemma}

\begin{proof}
Let $n\in\Z$, $\omega\in \Omega$ and $z_1,z_2\in \mathbb{S}_\epsilon$.
{ \tiny \begin{align*}
\| T_{\omega}^{(n)}(z_1) - T_{\omega}^{(n)}(z_2) \| &   =\left\| \begin{pmatrix}
   V^{(2n)}_\omega & 0\\
    0 & (U^{(2n)}_\omega)^*
    \end{pmatrix} \left[\begin{pmatrix}
     z_1^{-1}\widetilde{\rho}^{-1} & -\widetilde{\rho}^{-1}\alpha^*\\
    -\alpha\widetilde{\rho}^{-1} & z_1\rho^{-1}\end{pmatrix}-\begin{pmatrix}
     z_2^{-1}\widetilde{\rho}^{-1} & -\widetilde{\rho}^{-1}\alpha^*\\
    -\alpha\widetilde{\rho}^{-1} & z_2\rho^{-1}\end{pmatrix} \right] \begin{pmatrix}
    V^{(2n-1)}_\omega & 0\\
    0 & (U_\omega^{(2n-1)})^*
    \end{pmatrix}\begin{pmatrix}
    \widetilde{\rho}^{-1} & -\widetilde{\rho}^{-1}\alpha^*\\
    -\alpha\widetilde{\rho}^{-1} & \rho^{-1}\end{pmatrix}\right\| \\
    & \leq \left\| \begin{pmatrix}
    \widetilde{\rho}^{-1} & -\widetilde{\rho}^{-1}\alpha^*\\
    -\alpha\widetilde{\rho}^{-1} & \rho^{-1}\end{pmatrix}\right\|.\left\|\begin{pmatrix}
     (z_1^{-1}-z_2^{-1})\widetilde{\rho}^{-1} & 0\\
    0 & (z_1-z_2)\rho^{-1}\end{pmatrix}  \right\|. 
\end{align*} }
    
\noindent Since $\|\alpha\|<1$ and using Lemma \ref{lem_estime_rho}, there exists $C_3'=C_3'(\alpha)>0$ such that:
\begin{align*} 
   \|T_{\omega}^{(n)}(z_1)- T_{\omega}^{(n)}(z_2)\| & \leq C_3' \left\| \left(\begin{smallmatrix}
     (z_1^{-1}-z_2^{-1})\widetilde{\rho}^{-1} & 0\\
    0 & (z_1-z_2)\rho^{-1}\end{smallmatrix}\right) \right\| \\
    & \leq C_3'\left(\| (z_1^{-1}-z_2^{-1})\widetilde{\rho}^{-1}\|+\| (z_1-z_2)\rho^{-1}\|\right) \\
    & \leq C_3'\|\rho^{-1}\|\left(\frac{|z_1-z_2|}{|z_1.z_2|}+ |z_1-z_2|\right).
   \end{align*}
Since $z_1$ and $z_2$ are in $\mathbb{S}_\epsilon$, $|z_i|>1-\epsilon$, $\frac{1}{|z_1.z_2|}<\frac{1}{(1-\epsilon)^2}$ and using lemma \ref{lem_estime_rho} we find,
$$ \|T_{\omega}^{(n)}(z_1)- T_{\omega}^{(n)}(z_2)\|\leq C_3'\frac{1}{{ \sqrt{1-\|\alpha\|^2}}} \left(1+\frac{1}{(1-\epsilon)^2}\right)|z_1-z_2| \leq C_3|z_1-z_2|.$$
\end{proof}

We now prove the last step before proving the continuity of the Lyapunov exponents.

\begin{lemma}\label{lemma4_cont_Lyap}
Let $\epsilon\in (0,1)$. Let $p\in \llbracket 1,L \rrbracket$. There exists $C_4=C_4(\alpha,\epsilon,p)>0$ such that:  
\begin{equation}\label{eq_lemma4_cont_Lyap}
 \forall n\in \mathbb{Z}, \forall \omega\in \Omega, \forall z_1, z_2 \in \mathbb{S}_\epsilon,\ \|\wedge^p T_{\omega}^{(n)}(z_1)- \wedge^p T_{\omega}^{(n)}(z_2)\|\leq C_4 \mid z_1-z_2\mid.     
\end{equation}
\end{lemma}

\begin{proof}
As the transfer matrices are invertible, we can reproduce the reasoning from the proof of \cite[Lemma 6.2.6]{boumaza2007exposants} and use the following inequality:
$$
\left\|\wedge^{p} M-\wedge^{p} N\right\| \leq\|N-M\|\left(\|N\|^{p-1}+\|M\| \cdot\|N\|^{p-2}+\ldots+\|M\|^{p-1}\right)
$$
Then it suffices to apply Lemma \ref{lemma2_cont_Lyap} and Lemma \ref{lemma3_cont_Lyap} to conclude.
\end{proof}

Having these four lemmas, one can now deduce the continuity of the Lyapunov exponents. 

\noindent For $ p \in \llbracket 1,L \rrbracket$, $\epsilon\in (0,1)$, $\bar{x}\in P(\mathbb{C}^{2L})$  and $z\in\mathbb{S}_\epsilon$, we define:
\begin{equation}
\Phi(z,\bar{x})=\mathbb{E}\left(\ln{\frac{\| \wedge^pT_{\omega}^{(n)}(z) \bar{x}\|}{\|\bar{x}\|}}\right)
\end{equation}
where $P(\mathbb{C}^{2L})$ is the projective space associated with $\mathbb{C}^{2L}$. Let $L_p$ be the $p$-Lagrangian of $\wedge^p \C^{2L}$.

\begin{lemma}\label{lemma_cont_Lyap}
 The map $(z,\bar{x})\mapsto  \Phi(z,\bar{x})$ has the following properties: 
    \begin{enumerate}
        \item  The map $\bar{x}\mapsto  \Phi(z,\bar{x})$ is continuous on  $P({L}_{p})$.
        \item There exists a constant $C>0$ such that for all $z_1$, $z_2 \ \in \mathbb{S}_\epsilon$: 
        $$\sup_{\bar{x}\in P({L}_{p})} |\Phi(z_1,\bar{x})-\Phi(z_2,\bar{x})| < C|z_1-z_2|. $$
        \item The function 
        $$\begin{aligned}
          &\mathbb{S}_\epsilon \times P({L}_{p})\to\mathbb{R} & \ \\
       \Phi:&(z,\bar{x})\mapsto\Phi(z,\bar{x})
        \end{aligned}$$
        is continuous.
        \end{enumerate}
\end{lemma}
\begin{proof}
We start by demonstrating the first point. According to Lemma \ref{lemma2_cont_Lyap}, $\wedge^pT_{\omega}^{(n)}(z)$ is uniformly bounded in $z$ and $\omega$ and  
   $$ \forall z\in\mathbb{S}_\epsilon, \ \forall \omega \in \Omega, \ \ln{\frac{\| \wedge^p T_{\omega}^{(n)}(z) \bar{x}\|}{\|\bar{x}\|}} \leq \ln{\| \wedge^p T_{\omega}^{(n)}(z) \|} \leq  \ln(C_2) $$
   hence if $ \bar{x}_m \to \bar{x}$ in $ \mathbb{P}({L}_{p}) $, it suffices to apply the dominated convergence theorem to show that: 
   $$\Phi (z,\bar{x}_m)= \mathbb{E}\left(\frac{\|(\wedge^p T_{\omega}^{(n)}(z)\bar{x}_m\|}{\|\bar{x}_m\|}\right)\underset{m\to +\infty}{\xrightarrow{\hspace{2.5em}}} \mathbb{E}\left(\frac{\|(\wedge^p T_{\omega}^{(n)}(z) \bar{x}\|}{\|\bar{x}\|}\right)=\Phi(z,\bar{x}).$$\\
     For the second assertion, we begin by taking $z_1$ and $z_2$ in $\mathbb{S}_\epsilon$. We have:
 $$\Phi(z_1,\bar{x})-\Phi(z_2,\bar{x})= \mathbb{E}\left(\ln{\frac{\|\wedge^pT_{\omega}^{(n)}(z_1) \bar{x}\|}{\|\wedge^pT_{\omega}^{(n)}(z_2) \bar{x}\|}}\right).$$
Writing 
\begin{align*}
    \| \wedge^pT_{\omega}^{(n)}(z_1))\bar{x} \| & =
\| \wedge^pT_{\omega}^{(n)}(z_1)\cdot \wedge^p(T_{\omega}^{(n)}(z_2))^{-1}\cdot \wedge^pT_{\omega}^{(n)}(z_2)\bar{x}\| \\
 & \leq \| \wedge^pT_{\omega}^{(n)}(z_1)\cdot \wedge^p (T_{\omega}^{(n)}(z_2)))^{-1}\|\cdot \|\wedge^pT_{\omega}^{(n)}(z_2)\bar{x}\|
\end{align*}
and the same exchanging $z_1$ and $z_2$, one gets: 
$$ \left| \ln \left(\frac{\| \wedge^pT_{\omega}^{(n)}(z_1)\bar{x} \|}{\|{ \wedge^pT_{\omega}^{(n)}(z_2)\bar{x} 
} \|}\right) \right| \leq \ln{(\max\{\|\wedge^pT_{\omega}^{(n)}(z_1)  (\wedge^pT_{\omega}^{(n)}(z_2))^{-1} \|;
\|\wedge^p T_{\omega}^{(n)}(z_2)  (\wedge^p T_{\omega}^{(n)}(z_1))^{-1} \|\}
)}$$ 
Moreover, 
\begin{align*}
\|T_{\omega}^{(n)}(z_1)  (T_{\omega}^{(n)}(z_2))^{-1} \|
&\leq \| \left(\wedge^pT_{\omega}^{(n)}(z_1)- \wedge^pT_{\omega}^{(n)}(z_2)\right) \cdot (\wedge^pT_{\omega}^{(n)}(z_2))^{-1}+I_{2L}\|\\
&\leq \|(\wedge^p T_{\omega}^{(n)}(z_2))^{-1}\| \cdot \|\left(\wedge^pT_{\omega}^{(n)}(z_1)- \wedge^pT_{\omega}^{(n)}(z_2)\right)\|+1
\end{align*} 
and the same exchanging $z_1$ and $z_2$. This gives the following inequality: 
\small
\[ \left| \ln \left(\frac{\| \wedge^pT_{\omega}^{(n)}(z_1)\bar{x} \|}{\|{ \wedge^p T_{\omega}^{(n)}(z_2)\bar{x} } \|}\right) \right|  \leq \ln{\left[ \max_{i=1,2}\left\{\|(\wedge^p T_{\omega}^{(n)}(z_i))^{-1} \|\right\}\cdot\|\wedge^p T_{\omega}^{(n)}(z_1)- \wedge^pT_{\omega}^{(n)}(z_2)\|+1\right]} \\
. \]
\normalsize
Since for every $x>0$, $ \ln(x) \leq x-1$, using \eqref{mat_trans_relation_inverse} to be able to apply Lemma \ref{lemma2_cont_Lyap} to the inverse of the transfer matrices, we get: 
\[ \left|\mathbb{E}\left(\ln{\frac{\|\wedge^pT_{\omega}^{(n)}(z_1)\bar{x}\|}{\|\wedge^pT_{\omega}^{(n)}(z_2)\bar{x}\|}}\right)\right| \leq C_2 \|\wedge^pT_{\omega}^{(n)}(z_1)-\wedge^pT_{\omega}^{(n)}(z_2)\|. \]
Lemma \ref{lemma4_cont_Lyap} ensures the existence of a constant $C_4>0$ such that:  
\[ \|\wedge^pT_{\omega}^{(n)}(z_1)- \wedge^pT_{\omega}^{(n)}(z_2)\|\leq C_4 \mid z_1-z_2\mid . \]
Since $C_2$ and $C_4$ do not depend on $\bar{x}$, we get the second point.

For the third point, it suffices to combine the two previous points. 
Let $\epsilon\in (0,1)$ and take $(z_1,\bar{x}_1)$ and $(z_2,\bar{x}_2)$ close enough to have: 
\[ \begin{aligned}
    \mid \Phi(z_1,\bar{x}_1)-\Phi(z_2,\bar{x}_2)\mid &\leq \mid \Phi(z_1,\bar{x}_1)-\Phi(z_2,\bar{x}_1)\mid +  \mid \Phi(z_2,\bar{x}_1)-\Phi(z_2,\bar{x}_2)\mid\\
    & \leq C_2 C_4 \mid z_1-z_2\mid +\epsilon \\
    & \leq C \epsilon .
\end{aligned} \]
This completes the proof. 
\end{proof}

Using the inclusion \eqref{eq_ULL_subset_G} and the applications defined at \eqref{def_Cayley} and \eqref{def_pi_C_R} one gets, following the strategy of   \cite{asch2010localization}, that the group $\pi(C^*G_{\mu_z}C)$ is $p$-contracting and $L_p$-strongly irreducible  for every $p\in \llbracket 1,L\rrbracket$ and every $z\in \C \setminus \{ 0\}$ since $\pi(C^*U(L,L)C)$ is. This implies, for each  $p\in \llbracket 1,L\rrbracket$, the existence of a unique measure $\nu_{p,z}$ on $P(\pi^{-1}(C L_p C^*))\subset P(\mathbb{C}^{2L})$ which is $\mu_z$-invariant and such that :
\begin{equation}\label{formule_int_Lyap}
 \forall z\in \C \setminus \{ 0\}, \ \gamma_1(z)+\hdots+ \gamma_p(z)=\int_{G_{\mu_z}\times P(\C^{2L})} \ln\frac{\|\wedge^p M \bar{x}\|}{\|\bar{x}\|}\  {\textrm{d}\mu_z}(M) {\textrm{d}\nu_{p,z}}(\bar{x}).   
\end{equation}

Following the proof of \cite{hamza2007localization} or \cite{boumaza2007exposants}, one deduce from the integral representation \eqref{formule_int_Lyap}, from Lemma \ref{lemma_cont_Lyap} and the use of Banach-Alaoglu's theorem the continuity of the sums of the Lyapunov exponents hence of the Lyapunov exponents themselves, as functions of the parameter $z$.

\begin{proposition}\label{prop_cont_Lyap}
Let $\epsilon\in (0,1)$ and  $p\in \llbracket 1,L\rrbracket$. The function $z \mapsto \gamma_p(z)$ is continuous on $\C \setminus \{ 0\}$.
\end{proposition}

We deduce from this continuity results a result of positivity of the Lyapunov exponents on some annulus.

\begin{corollary}\label{corollaire_pos_Lyap_couronne}
There exists $\epsilon_0\in (0,1)$ such that for
every $\epsilon \in (0,\epsilon_0]$, the $L$ first Lyapunov exponents are positive on $\mathbb{S}_{\epsilon}$.
\end{corollary} 
  
\begin{proof}
Let $\epsilon\in (0,1)$ and  $p\in \llbracket 1,L\rrbracket$. It was shown in \cite{boumaza2015absence} that for every $z\in \mathbb{S}^1$, $\gamma_p(z)>0$. Hence, by continuity of $z\mapsto \gamma_p(z)$ on $\C \setminus \{ 0\}$, this function remains positive in a neighborhood of $\mathbb{S}^1$ hence on some $\mathbb{S}_{\epsilon_0}$ for an $\epsilon_0\in (0,1)$.
\end{proof}

\subsection{Exponential decay of transfer matrices.}
The positivity of the Lyapunov exponents does not imply directly the exponential decaying of the eigenfunctions of the operators $\mathbb{U}_\omega$, as discussed for example in \cite{boumaza2023localization}, but it already implies the exponential decaying of some power of the inverse of products of transfer matrices. 




\begin{lemma}\label{lemma_exp_decay_trans_mat_inv}
Let \(K \subset \mathbb{C}\) a compact set. There exist constants \(C>O\), \(c>0\) and \(s\in (0,1)\), such that: 
\begin{equation}\label{eq_lemma_exp_decay_trans_mat_inv}
 \mathbb{E}\left( \|\left(T_{\omega}^{(m)}(z)\dots T_{\omega}^{(n)}(z)\right)^{-1}v\|^s\right)\leq C e^{-c|n-m|}
 \end{equation}
for all \( z \in K\), for all  unit vector \( v\) , and for $n,m\in \Z$ such that $|n-m|$ is sufficiently large.
\end{lemma}
\begin{proof}
One can follow the proof of  \cite[Lemma 7.1]{hamza2009dynamical} or \cite[Lemma 3]{boumaza2009localization}. It suffices to use \eqref{mat_trans_relation_inverse} and the fact that 
$$\lim_{|n-m|\rightarrow \infty}\frac{1}{|n-m|} \mathbb{E}\left( \ln \|\left(T_{\omega}^{(m)}(z)\dots T_{\omega}^{(n)}(z) \right)^{-1}v\|\right)= -\gamma_L(z)<0$$
since for any $A\in \mathrm{GL}_{\mathrm{L}}(\C)$, the Hilbert-Schmidt norm of its inverse satisfies $\|A^{-1} \|_{HS}=(\sigma_L(A))^{-1}$ where $\sigma_L(A)$ is the smallest singular value of $A$. 
\end{proof}

\section{Exponential decay of the fractional moments of the Green's function}\label{sec_exp_decay_Green}

In this Section we provide a proof of the exponential decay of fractional moments of the Green's matrix, which is the result of Theorem \ref{thm_greengeneral}. We will begin our analysis by reducing the problem to the proof of the exponential decay of some particular coefficients of the Green kernel, following the approaches described in previous works such as \cite{hamza2009dynamical, asch2010localization}. Then, we will adopt a new strategy to establish exponential decay in the reduced case.

The usual scheme to establish exponential decay of fractional moments is to find an explicit expression of the Green's function in terms of the solutions of the equation \((\mathbb{U}_\omega-z)\phi=0\) as in \cite{hamza2009dynamical}. In our case, it is difficult to use this approach due to the matrix structure of the model which implies a formula (see \cite{damanik2008analytictheorymatrixorthogonal}) whose terms are difficult to control. Therefore, we proceed differently, in two steps: first, we bound the norm of the Green's function by multiplying it by a non-constant quantity and the norm of the inverses of the transfer matrices which, as we have already established, decay exponentially. Then, we demonstrate that the first quantity is bounded, under a condition on the norm of the matrix \(\alpha\).

\subsection{Finite scattering zippers.}\label{sec_finitezippers}
In this section, we start by explicitly compute the even-odd truncation for the scattering zipper operator. The importance of this step lies in the fact that truncating a unitary operator can lead to a non-unitary operator if the truncation is not done well. Let us first denote explicitly the coefficients of the scattering matrices where we omit the random paramter $\omega$ in order to symplify the latter computations.

Let $\omega\in \Omega$. Denote $S_\omega^{(n)}=\left( \begin{smallmatrix}
    \alpha_n & \beta_n\\
    \gamma_n & \delta_n \\
\end{smallmatrix} \right)  \in U(2L)_{\mathrm{inv}} $ and take $U,V \in U(L). $
 We define the restriction $\mathbb{U}^{[2n,2m+1]}_\omega$ of $\mathbb{U}_{\omega}$ to an interval $[2n,2m+1]$ with $n\leq m$ two elements of $\Z$.

We construct $ \mathbb{V}^{[2n,2m+1]}_\omega $ and $\mathbb{W}^{[2n,2m+1]}_\omega$ as follows: 
\begin{equation}\label{def_V_2n2m1}
    \mathbb{V}^{[2n,2m+1]}_\omega :=\bigoplus_{k\in \llbracket n,m\rrbracket} S_{\omega}^{(2k)} 
\end{equation}
\begin{equation}\label{def_W_2n2m1}
\mathbb{W}^{[2n,2m+1]}_\omega:= U \oplus\bigoplus_{k\in \llbracket n,m-1\rrbracket} S_{\omega}^{(2k+1)} \oplus V
\end{equation}
 
 The boundary conditions U and V are placed in $\mathbb{W}_{\omega}$ in such a way that we obtain exactly $2m+1-2n+1$ $L$-rows and $L$-columns starting from $2n$.
 Thus, $\mathbb{U}^{[2n,2m+1]}_\omega$ is given by:
{\footnotesize $$ \begin{matrix}
     2n\\
     2n+1\\
     \vdots\\
     \vdots\\
     \vdots\\
     \vdots\\
     2m\\
    2m+1\\
 \end{matrix} \ \begin{pmatrix}
 \text{\larger[2] $S$}_{\omega}^{(2n)}   & \begin{array}{cc}0 & 0\\ 0 & 0\end{array} 
& \cdots  & \cdots &\begin{array}{cc}0 & 0\\ 0 & 0\end{array}\\

\begin{array}{cc}0 & 0\\ 0 & 0\end{array}  & \text{\larger[2] $S$}_{\omega}^{(2n+2)}  & \cdots&  \cdots  & \begin{array}{cc}0 & 0\\ 0 & 0\end{array}  \\

\vdots & \vdots & \ddots & \ddots   &  \vdots  \hspace{0.5 cm}\vdots  \\

\vdots & \vdots  &  \vdots   &  \text{\larger[2] $S$}_{\omega}^{(2m-2)}  & \begin{array}{cc}0& 0\\0 &  0\end{array}  \\

\begin{array}{cc}0 & 0\\ 0 & 0\end{array}  & \cdots & \cdots & \begin{array}{cc} 0& 0\\ 0& 0\end{array} & \text{\larger[2] $S$}_{\omega}^{(2m)} \\

\end{pmatrix}\cdot
\begin{pmatrix}
 U   &  \begin{array}{cc}0&0\end{array} & \cdots & \cdots & 0 \\
 
\begin{array}{cc} 0 \\ 0 \end{array} & \text{\larger[2] $S$}_{\omega}^{(2n+1)}  & \cdots
&  \cdots & \begin{array}{cc}0 \\ 0 \end{array}  \\

\begin{array}{cc}0 \\  \vdots \end{array} & \cdots & \hspace{-1.2cm} \ddots & \cdots  & \vdots  \\

\vdots & \vdots  &  \hspace{1.2cm} \ddots   &  \cdots  & \vdots    \\

 \begin{array}{cc} 0 \\ 0 \end{array} & \cdots & \cdots &  \text{\larger[2] $S$}_{\omega}^{(2m-1)}  &\begin{array}{cc}0 \\ 0 \end{array}\\
 
  0  &  \cdots &  \cdots&  \begin{array}{cc}0 & 0 \end{array} & V
\end{pmatrix}$$}
Since we will mostly deal with the resolvent of  $\mathbb{U}^{[2n,2m+1]}_\omega$, we also give the expression of  $\mathbb{U}^{[2n,2m+1]}_\omega-z\mathrm{Id}$ for $z\in \C$.
$$\mathbb{U}^{[2n,2m+1]}_\omega-z\mathrm{Id}=\qquad \qquad \qquad \qquad \qquad \qquad \qquad \qquad \qquad \qquad \qquad \qquad \qquad \qquad \qquad\qquad \qquad \qquad \qquad \qquad$$
\scriptsize
$$\begin{pmatrix}
\alpha_{2n}  U-z & \beta_{2n} \alpha_{2n+1}  & \beta_{2n}\beta_{2n+1} & 0 & \cdots & \cdots & \cdots    &  0 \\
\gamma_{2n} U & \delta_{2n} \alpha_{2n+1}-z & \delta_{2n} \beta_{2n+1}  & 0 & \cdots & \cdots    & \cdots  & 0 \\
0 & \alpha_{2n+2}\gamma_{2n+1} &  \alpha_{2n+2}\delta_{2n+1}-z & \beta_{2n+2} \alpha_{2n+3} & \beta_{2n+2}\beta_{2n+3} & 0 & \cdots  & 0 \\
0 & \gamma_{2n+2}\gamma_{2n+1} &  \gamma_{2n+2} \delta_{2n+1} & \delta_{2n+2}\alpha_{2n+3} -z & \delta_{2n+2}\beta_{2n+3} & 0 & \cdots  &  0 \\
0 & 0 &   0 &  \alpha_{2n+4}\gamma_{2n+3} & \alpha_{2n+4}\delta_{2n+3}-z & * *  & \cdots  & 0\\
0 & 0 &   0 &  \gamma_{2n+4}\gamma_{2n+3} & \gamma_{2n+4}\delta_{2n+3} & * *  & \cdots  & 0\\
\vdots &   \ \ \ddots & \ \ \ddots  &     \ddots  &  \ddots  &    \ddots    & \ddots    & \vdots \\
\vdots &  \  \ddots & \ddots  &     \ddots    & \ddots  &    \ddots & \ddots  & \vdots \\
0 &  \cdots &  0 & \alpha_{2m-2}\gamma_{2m-3} &  \alpha_{2m-2}\delta_{2m-3}-z & \beta_{2m-2}\alpha_{2m-1} & \beta_{2m-2} \beta_{2m-1} & 0 \\
0 &  \cdots &  0 & \gamma_{2m-2}\gamma_{2m-3} &  \gamma_{2m-2}\delta_{2m-3} & \delta_{2m-2} \alpha_{2m-1} -z& \delta_{2}\beta_{2n+3} & 0 \\
0 &  \cdots &  0 & 0 & 0&   \alpha_{2m}\gamma_{2m-1}  &\alpha_{2m} \delta_{2m-1} -z& \beta_{2m}V  \\
0 &  \cdots &  0 & 0 & 0&   \gamma_{2m}\gamma_{2m-1}  &\gamma_{2m} \delta_{2m-1} & \delta_{2m}V -z\\
\end{pmatrix}$$
\normalsize
\bigskip

Let $z\in \mathbb{C}$ and $\Phi\in \ell^2([2n,2m+1],\mathbb{C}^L)$ satisfying $\mathbb{U}^{[2n,2m+1]}_\omega\Phi=z\Phi$.\\
If we define $\Psi=\mathbb{W}^{[2n,2m+1]}_\omega\Phi$, then $\mathbb{V}^{[2n,2m+1]}_\omega\Psi=z\Phi$, which leads to $\Phi=z^{-1}\mathbb{V}^{[2n,2m+1]}_\omega\Psi$.\\
In other words, in the formalism of scattering matrices (or S-matrices), we have:
{\small 
$$
\Psi_{2n}=U\Phi_{2n}, \
\begin{pmatrix}
\Phi_{2n}\\
\Phi_{2n+1}
\end{pmatrix}=z^{-1}S_{\omega}^{(2n)} 
\begin{pmatrix}
\Psi_{2n} \\
\Psi_{2n+1}
\end{pmatrix}, \ \begin{pmatrix}
\Psi_{2n+1}\\
\Psi_{2n+2}
\end{pmatrix}=S_{\omega}^{(2n+1)} 
\begin{pmatrix}
\Phi_{2n+1} \\
\Phi_{2n+2}
\end{pmatrix},  \ \ldots\  \ldots$$
$$\ldots \ \ldots \ ,
\begin{pmatrix}
\Phi_{2m-2}\\
\Phi_{2m-1}
\end{pmatrix}=z^{-1}S_{\omega}^{(2m-2)} 
\begin{pmatrix}
\Psi_{2m-2} \\
\Psi_{2m-1}
\end{pmatrix},
\begin{pmatrix}
\Psi_{2m-1}\\
\Psi_{2m}
\end{pmatrix}=S_{\omega}^{(2m-1)} 
\begin{pmatrix}
\Phi_{2m-1} \\
\Phi_{2m} 
\end{pmatrix}, \ V\Phi_{2m+1}= \Psi_{2m+1} . 
$$
}
This yields the following diagram:
\begin{figure}[H]
\centering
\includegraphics[width=0.76\textwidth]{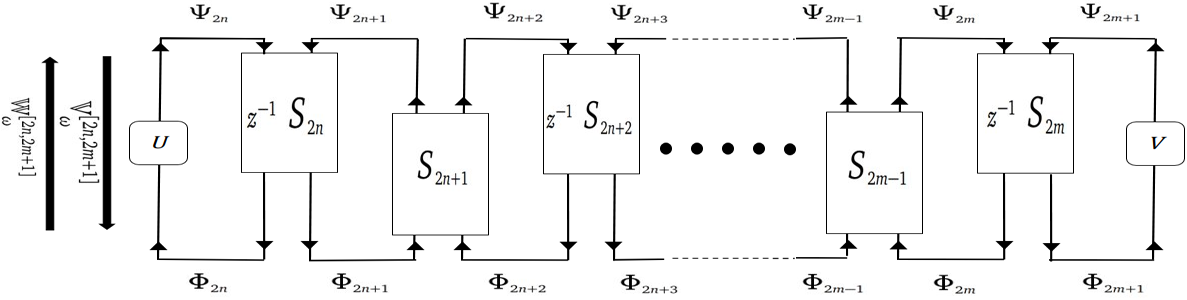}  
\end{figure}
Other cases of intervals ($[2n,2m]$, $[2n+1,2m]$ and $[2n+1,2m+1]$) are defined in the same manner. The cases where $a=-\infty$ or $b=+\infty$ are similar since we have only to put boundary conditions on one side of the interval.

\subsection{Proof of Theorem \ref{thm_unif_bounded}.}
The estimate of the fractional moments of the Green's function has played a prominent role in the proofs of localization for the self-adjoint Anderson model, notably illustrated in the works \cite{graf1994anderson,aizenman1993localization,aizenman1998localization}. Our approach to evaluate these moments begins by establishing an upper bound for their expectations. We will draw inspiration from the method presented by Hamza in \cite{hamza2007localization}, while adjusting the proofs to adapt the argument with the quasi-one-dimensional character of $\{ \mathbb{U}_{\omega} \}_{\omega \in \Omega}$. This upper bound will subsequently serve as a key tool to prove the exponential decay of the fractional moments of the resolvent.

\begin{proof} Of Theorem \ref{thm_unif_bounded}. Let \(s \in (0, \frac14)\), $\omega\in \Omega$ and $z\in \C \setminus \mathbb{S}^1$. Let $\epsilon\in (0,1)$. For $|z| \leq \epsilon$, the bound \eqref{eq_thm_unif_bounded} is clear. Hence one may assume that $|z|\geq \epsilon$, $|z|\neq 1$. Then, since $\left(\mathbb{U}_\omega - z\right)^{-1} = \frac{1}{2 z} \left[\left(\mathbb{U}_\omega + z\right) \left(\mathbb{U}_\omega - z\right)^{-1} - \mathrm{Id}\right]$, there exists $\widetilde{C}(s)>0$ such that
$$
\mathbb{E}\left[\left\|G_{\omega}(z, k, l)\right\|^{s}\right] \leq \widetilde{C}(s) \left(\mathbb{E}\left[\left\|\left(\left(\mathbb{U}_\omega + z\right) \left(\mathbb{U}_\omega - z\right)^{-1} \right)(k, l)\right\|^{s}\right] + \one \right)
$$
Therefore, it suffices to prove the existence of $\widehat{C}(s) >0$ such that
$$
\mathbb{E}\left[\left\|\left(\left(\mathbb{U}_\omega + z\right) \left(\mathbb{U}_\omega - z\right)^{-1} \right)(k, l)\right\|^{s}\right] \leq \widehat{C}(s)
$$
for all $z \in \mathbb{C} \setminus \mathbb{S}^1$ and for all $k, l \in \mathbb{Z}$ to prove \eqref{eq_thm_unif_bounded}. We use the finite-rank perturbation method to establish such a bound on the modified resolvent. 
Let $k\in \Z$. We use the notations introduced in \eqref{def_phasesUV} and \eqref{eq-Smatdef} which allows us to write:
$$
S^{(2k)}_\omega = \begin{pmatrix}
   \one & 0\\
   0 & \widehat{V}^{(2k)}_\omega e^{i\theta^{(2k)}_\omega} 
    \end{pmatrix} \begin{pmatrix}
    \alpha & \rho(\alpha) U^{(2k)}_\omega\\
    \widetilde{\rho}(\alpha) & -\alpha^* U^{(2k)}_\omega
    \end{pmatrix}.
$$
This implies a factorization of $\mathbb{U}_{\omega}$:
\begin{equation}\label{eq_factorization_proof_unif_bound}
    \mathbb{U}_\omega  = \mathbb{V}_\omega \mathbb{W}_\omega 
                   = \left(\bigoplus_{k \in \mathbb{Z}} \begin{pmatrix}
   \one & 0\\
   0 & \widehat{V}^{(2k)}_\omega e^{i\theta^{(2k)}_\omega}
    \end{pmatrix} \begin{pmatrix}
    \alpha & \rho(\alpha) U^{(2k)}_\omega\\
    \widetilde{\rho}(\alpha) & -\alpha^* U^{(2k)}_\omega
    \end{pmatrix}\right) \circ s^L \mathbb{W}_\omega := \mathbb{Y}_\omega \mathbb{D}_\omega \mathbb{V}_\omega' \mathbb{W}_\omega.
\end{equation}
with the unitaries $\mathbb{Y}_\omega = \bigoplus_{k \in \mathbb{Z}} \left( \begin{smallmatrix}
\one & 0\\
   0 & \widehat{V}^{(2k)}_\omega 
   \end{smallmatrix} \right) $, $\mathbb{D}_\omega = \bigoplus_{k \in \mathbb{Z}} \left( \begin{smallmatrix}
\one & 0\\
   0 & e^{i\theta^{(2k)}_\omega}
   \end{smallmatrix} \right) $ and $\mathbb{V}_\omega' = \bigoplus_{k \in \mathbb{Z}} \left( \begin{smallmatrix}
        \alpha & \rho(\alpha) U^{(2k)}_\omega\\
    \widetilde{\rho}(\alpha) & -\alpha^* U^{(2k)}_\omega
   \end{smallmatrix}  \right) \circ s^L$.
   
For $k \neq l$ even, $A = \{k+1, l+1\} \subset \mathbb{Z}$ and $a = \frac{1}{2} (\theta^{(k)}_{\omega} + \theta^{(l)}_{\omega})$, $b = \frac{1}{2} (\theta^{(k)}_{\omega} - \theta^{(l)}_{\omega})$, let
$$
\begin{gathered}
\eta^{(j)}_\omega = \begin{cases}a & j \in A \\
0 & j \notin A\end{cases} \quad 
\mbox{ and }\quad \xi^{(j)}_\omega = \begin{cases}b &
 j = k+1 \\
-b & j = l+1, \\
0 & j \notin A.\end{cases}
\end{gathered}
$$
If one of the $k,l$ is odd, we set $a =b = \frac{1}{2} \theta^{(k)}_{\omega}$. Let $E_j = (0, 0, \hdots,0,\one, 0, \hdots, 0)$ where $\one$ is in the $j^\text{th}$ position, and define the diagonal matrices by $L$-blocks, $D_{a}, D_{b}$, and $\widehat{D}$ by:
\begin{equation}
    D_{a} E_j = \widehat{V}^{(j)}_\omega e^{i\eta^{(j)}_\omega} E_j,\quad D_{b} E_j = \widehat{V}^{(j)}_\omega e^{i\xi^{(j)}_\omega} E_j, \quad    \mbox{and}\quad 
 \widehat{D} E_j = \begin{cases} E_j & j \text{ even} \\
\widehat{V}^{(j)}_\omega e^{i\theta^{(j)}_\omega} & j \text{ odd}\end{cases}.   
\end{equation}
Using these definitions, we factorize $\mathbb{D}_\omega$ as $\mathbb{D}_\omega=D_{a} D_{b} \widehat{D}$ and we have:
\begin{equation}
\mathbb{U}_\omega = \mathbb{Y}_\omega \mathbb{D}_\omega \mathbb{V}_\omega' \mathbb{W}_\omega = \mathbb{Y}_\omega D_{a} D_{b} \widehat{D} \mathbb{V}_\omega' \mathbb{W}_\omega := \mathbb{Y}_\omega D_{a} \mathbb{X}_\omega.    
\end{equation}
with the unitary operator $\mathbb{X}_\omega = D_{b} \widehat{D} \mathbb{V}_\omega' \mathbb{W}_\omega$ independent of $a$. Note that $D_a = \mathbb{Y}_\omega^{-1} \mathbb{U}_\omega \mathbb{X}_\omega^{-1}$ and that in view of the definitions of $ \mathbb{X}_\omega$ and $\mathbb{Y}_\omega$, $\mathbb{U}_\omega - \mathbb{Y}_\omega \mathbb{X}_\omega$ is a finite-rank operator.

Let $P_A$ be the orthogonal projection onto the subspace generated by $\{\mathbb{X}_\omega^{-1} e_j: j \in A\}$. Note that $\{\mathbb{X}_\omega^{-1} e_j: j \in \mathbb{Z}\}$ is an orthonormal basis of $l^2(\mathbb{Z})$.
One has
\begin{equation}
P_A \mathbb{X}_\omega^{-1} \mathbb{Y}_\omega^{-1} \mathbb{U}_\omega P_A = e^{-ia} \mathrm{I}_2.
\end{equation}
Indeed, for  $u\in \left\{\mathbb{X}_\omega^{-1} e_{j}: j \in A\right\}$, there exist real numbers $a$ and  $b$ such that $u = a \mathbb{X}_\omega^{-1} e_{k+1} + b \mathbb{X}_\omega^{-1} e_{l+1}.$
Thus, $\mathbb{X}_\omega^{-1} \mathbb{Y}_\omega^{-1} \mathbb{U}_\omega u = a \mathbb{X}_\omega^{-1} D_a e_{k+1} + b \mathbb{X}_\omega^{-1} D_a e_{l+1} = e^{-ia} (a \mathbb{X}_\omega^{-1} e_{k+1} + b \mathbb{X}_\omega^{-1} e_{l+1}) = e^{-ia} u.$
Moreover, $\mathbb{U}_\omega - \mathbb{Y}_\omega \mathbb{X}_\omega = \mathbb{Y}_\omega (D_a - \mathrm{Id}) \mathbb{X}_\omega = 0$ on $\mathrm{Ran}(\mathrm{Id} - P_A)$ (the vector subspace generated by $\left\{\mathbb{X}_\omega^{-1} e_{j}: j \in \mathbb{Z} \setminus A\right\}$), since $D_a - \mathrm{Id} = 0$ on $\mathrm{Ran}(\mathrm{Id} - P_A)$. As $\left(\mathbb{Y}_\omega \mathbb{X}_\omega\right)^{-1} \mathbb{U}_\omega = e^{-ia} \mathrm{I}_2$ on  $\mathrm{Ran}(P_A)$, 
{\small \begin{equation*}
      \mathbb{U}_\omega = \mathbb{Y}_\omega \mathbb{X}_\omega (\mathrm{Id} - P_A) + \mathbb{U}_\omega P_A = \mathbb{Y}_\omega \mathbb{X}_\omega (\mathrm{Id} - P_A) + \mathbb{Y}_\omega \mathbb{X}_\omega \left(\mathbb{Y}_\omega \mathbb{X}_\omega\right)^{-1} \mathbb{U}_\omega P_A = \mathbb{Y}_\omega \mathbb{X}_\omega (\mathrm{Id} - P_A) + e^{-ia} \mathbb{Y}_\omega \mathbb{X}_\omega P_A.
\end{equation*}}
For $z \in \mathbb{C} \setminus (\mathbb{S}^1 \cup \{0\})$, define
\begin{equation}
      F_z = P_{A} (\mathbb{U}_\omega + z) (\mathbb{U}_\omega - z)^{-1} P_{A}, \quad \mbox{ and } \quad    \widehat{F}_z = P_{A} (\mathbb{Y}_\omega \mathbb{X}_\omega + z) (\mathbb{Y}_\omega \mathbb{X}_\omega - z)^{-1} P_{A}.
\end{equation}    
One has: 
\begin{align*}
&\widehat{F}_{z}+\widehat{F}_{z}^{*}=P_{A}\left[(\mathbb{Y}_\omega\mathbb{X}_\omega+z)(\mathbb{Y}_\omega\mathbb{X}_\omega-z)^{-1}+\left((\mathbb{Y}_\omega\mathbb{X}_\omega+z)(\mathbb{Y}_\omega\mathbb{X}_\omega-z)^{-1}\right)^*\right]P_A.\\
&=P_A((\mathbb{Y}_\omega\mathbb{X}_\omega-z)^{-1})^*\left[(\mathbb{Y}_\omega\mathbb{X}_\omega-z)^* (\mathbb{Y}_\omega\mathbb{X}_\omega+z)+ (\mathbb{Y}_\omega\mathbb{X}_\omega+z)^*(\mathbb{Y}_\omega\mathbb{X}_\omega-z)\right](\mathbb{Y}_\omega\mathbb{X}_\omega-z)^{-1}P_A.\\
&=P_A(2\mathrm{Id}-2|z|^2)\left((\mathbb{Y}_\omega\mathbb{X}_\omega-z)^{-1}\right)^*(\mathbb{Y}_\omega\mathbb{X}_\omega-z)^{-1}P_A.
\end{align*}
On the image of \(P_A\), this implies that \(\widehat{F}_{z} + \widehat{F}_{z}^{*} < 0\) for \(|z| > 1\). Consequently, \(-i \widehat{F}_{z}\) is a dissipative operator. Similarly, \(-i \widehat{F}_{z}^{-1}\) is also a dissipative operator. In the case where \(|z| < 1\), we have \(i \widehat{F}_{z}\) and \(i \widehat{F}_{z}^{-1}\) that are dissipative operators.

\noindent Next, we use the fact that \((x+z)(x-z)^{-1} = 1 + 2z(x-z)^{-1}\), to obtain: 
\begin{align*}
F_{z}-\widehat{F}_{z}  &= P_{A}\left[(\mathbb{U}_\omega+z)\left(\mathbb{U}_\omega-z\right)^{-1}-(\mathbb{Y}_\omega\mathbb{X}_\omega+z)(\mathbb{Y}_\omega\mathbb{X}_\omega-z)^{-1}\right]  P_{A}.\\
&= P_{A}\left[\mathrm{Id}+2z(\mathbb{U}_\omega-z)^{-1}-(\mathrm{Id}+2z(\mathbb{Y}_\omega\mathbb{X}_\omega-z)^{-1})\right]  P_{A}.\\
&=-2 z P_{A}\left[(\mathbb{Y}_\omega\mathbb{X}_\omega-z)^{-1}-(\mathbb{U}_\omega-z)^{-1}\right] P_{A}.  
\end{align*}
Applying the resolvent identity, we get:
\begin{align*}
F_{z}-\widehat{F}_{z}=-2 z P_{A}\left[(\mathbb{Y}_\omega\mathbb{X}_\omega-z)^{-1}(\mathbb{U}_\omega-\mathbb{Y}_\omega\mathbb{X}_\omega)(\mathbb{U}_\omega-z)^{-1}\right] P_{A}.
\end{align*}
As  $\mathbb{U}_\omega - \mathbb{Y}_\omega \mathbb{X}_\omega$ is zero on $\mathrm{Ran}(\mathrm{Id} - P_A)$  and equals $\mathbb{Y}_\omega\mathbb{X}_\omega(e^{-ia}-1)\mathrm{I}_2$ on $\mathrm{Ran}(P_A)$, we have:
\begin{align*}
\mathbb{U}_\omega-\mathbb{Y}_\omega\mathbb{X}_\omega&=\mathbb{Y}_\omega\mathbb{X}_\omega(\mathrm{Id}-P_A)+e^{-ia}\mathbb{Y}_\omega\mathbb{X}_\omega P_A-\mathbb{Y}_\omega\mathbb{X}_\omega =\mathbb{Y}_\omega\mathbb{X}_\omega-\mathbb{Y}_\omega\mathbb{X}_\omega P_A+e^{-ia}\mathbb{Y}_\omega\mathbb{X}_\omega P_A-\mathbb{Y}_\omega\mathbb{X}_\omega\\
&=\mathbb{Y}_\omega\mathbb{X}_\omega\left[e^{-ia}P_A-P_A\right] =\mathbb{Y}_\omega\mathbb{X}_\omega\left[e^{-ia}P_AP_A-P_AP_A\right] =\mathbb{Y}_\omega\mathbb{X}_\omega P_A(e^{-ia}-1)P_A,
\end{align*}
hence $F_{z}-\widehat{F}_{z}=-2 z P_{A}\left[(\mathbb{Y}_\omega\mathbb{X}_\omega-z)^{-1}(\mathbb{Y}_\omega\mathbb{X}_\omega)P_A(e^{-ia}-1)P_A(\mathbb{U}_\omega-z)^{-1}\right]P_A.$ 

\noindent Thus, $ F_{z}-\widehat{F}_{z} =-P_A(\mathbb{Y}_\omega\mathbb{X}_\omega-z)^{-1}\mathbb{Y}_\omega\mathbb{X}_\omega P_A(e^{-ia}-1)P_A(2z(\mathbb{U}_\omega-z)^{-1})P_A.$
Then, using the equality $2z(\mathbb{U}_\omega-z)^{-1}=(\mathbb{U}_\omega+z)(\mathbb{U}_\omega-z)^{-1}-\mathrm{Id},$ it follows that
\begin{align*}
     F_{z}-\widehat{F}_{z}&=P_A(\mathbb{Y}_\omega\mathbb{X}_\omega-z)^{-1}\mathbb{Y}_\omega\mathbb{X}_\omega P_A(e^{-ia}-1)[F_z-\mathrm{Id}]P_A\\
     &=-P_A\left[\mathrm{Id}+z(\mathbb{Y}_\omega\mathbb{X}_\omega)\right](\mathbb{Y}_\omega\mathbb{X}_\omega-z)^{-1}\mathbb{Y}_\omega\mathbb{X}_\omega P_A(e^{-ia}-1)[F_z-\mathrm{Id}]P_A\\
     &=P_A\left[\frac{1}{2}(\mathrm{Id}+(\mathbb{Y}_\omega\mathbb{X}_\omega+z)(\mathbb{Y}_\omega\mathbb{X}_\omega-z)^{-1})P_A(e^{-ia}-1)[\mathrm{Id}-F_z]\right]P_A\\
     &=P_A\left[\frac{1}{2}(\mathrm{Id}+\hat{F}_z)(e^{-ia}-1)[\mathrm{Id}-F_z]\right] P_A.
\end{align*}
Therefore, \(P_A(F_z-\hat{F})P_A=P_A\left[\frac{1}{2}(\mathrm{Id}+\hat{F}_z)(e^{-ia}-1)(\mathrm{Id}-F_z)\right] P_A,\) and we can write:
\begin{equation}
  F_{z}-\widehat{F}_{z}=\frac{1}{2}(\mathrm{Id}+\hat{F}_z)(e^{-ia}-1)(\mathrm{Id}-F_z),  
\end{equation}
which is equivalent to $F_{z}\left(I+\frac{1}{2}(\mathrm{Id}+\hat{F}_z)(e^{-ia}-1)\right)=F_z+\frac{1}{2}(\mathrm{Id}+\hat{F}_z)(e^{-ia}-1)$. 

Hence $F_{z}\frac{1}{2}(e^{-ia}-1)\left(\frac{e^{-ia}+1}{e^{-ia}-1}\mathrm{Id}+\hat{F}_z)\right)=\frac{1}{2}(e^{-ia}-1)\left(\mathrm{Id}+\hat{F}_z\frac{e^{-ia}+1}{e^{-ia}-1}\right)$.

Let us define \( m_a:=-i\left( \frac{e^{-ia}+1}{e^{-ia}-1}\right) \). Then we have:
\begin{equation}
  F_z(i.m_a+\hat{F}_z)=\mathrm{Id}+i.m_a \hat{F}_z.   
\end{equation}
Hence, $F_z$ and $\hat{F}_z$ satisfy the same relation as the one found in the proof of \cite[Theorem 3.1]{hamza2009dynamical} and they have the same dissipative properties. Then, one can follow the proof of \cite[Theorem 3.1]{hamza2009dynamical} for $\tau \equiv \frac{1}{2\pi}$ to finish the proof of \eqref{eq_thm_unif_bounded}. 

Note that if at the beginning of the proof we choose both $k$ and $l$ odd instead of even, we have to write the same proof as in the even case but by writing the factorization \eqref{eq_factorization_proof_unif_bound} on the factor $\mathbb{W}_\omega$ and replacing  $\widehat{V}^{(2k)}_\omega $ by  $\widehat{U}^{(2k)}_\omega$ and $\theta^{(2k)}_\omega$ by $\Theta^{(2k)}_\omega$.
\end{proof}

\subsection{Reduction results.}
In this section, we aim to focus our study on structures of odd lengths, and then on those of finite lengths. Thess reductions imply the need to ensure the invertibility of \(\alpha\). In what follows, we fix \(a, b \in \mathbb{Z}\), \(|a-b|>4\).

\subsubsection{Reduction to appropriate elements.}
The goal of this part is to estimate, in expectation, elements of the type \(\left\|G_\omega^{[a, b]}(z,2n,2m)\right\|^s\) and \(\left\|G_\omega^{[a, b]}(z,2n+1,2m)\right\|^s\) by elements of the type \(\left\|G_\omega^{[a, b]}(z,2n,2m+1)\right\|^s\) and \(\left\|G_\omega^{[a, b]}(z,2n+1,2m+1)\right\|^s\), for which we know explicit expressions.  Challenge in proving the reduction to even elements consists of being able to bound the even columns (and rows) of the Green's matrix by the odd columns (and rows), or more easily, by a sum. We use the notations introduced in Section \ref{sec_finitezippers}. We start with a lemma that will allow us to have a useful recurrence relation for the rest of the proof. 

\begin{lemma}\label{lmm3.2.1}
Let \(l \in \llbracket a,b \rrbracket\). For every \(k \in \llbracket a+1,b-1 \rrbracket ,\  k \neq l\) and every \(z \in \mathbb{C}\setminus\mathbb{S}^1 \),
\begin{equation}\label{eq_rel_resolvante}
     \left[{\gamma_k^{-1}}^*\gamma_{k-1} \right]G_{\omega}^{[a,b]}(z,k-1,l)
    +{\gamma_k^{-1}}^*[\delta_k+z\alpha_k^*]G_{\omega}^{[a,b]}(z,k,l) 
    +zG_{\omega}^{[a,b]}(z,k+1,l)=0.
\end{equation}    
\begin{equation}
    zG_{\omega}^{[a,b]}(z,k,l)
     + { \beta_k ^{-1}}^* \left[\alpha_{k+1} +z \delta_k^* \right]G_{\omega}^{[a,b]}(z,k+1,l)
    +{\beta_k^{-1}}^*\beta_{k+1}G_{\omega}^{[a,b]}(z,k+2,l)=0.
\label{3.3}\end{equation}
Similar inequalities are applicable for the columns.
\end{lemma}
\begin{proof}
By the definition of the resolvent:
$$\left[(\mathbb{U}^{[a,b]}_\omega-z\mathrm{Id})G_{\omega}^{[a,b]}(z)\right](k,l)=\left\{\begin{array}{ll}
    \one & \text{if } k=l \\
     0 & \text{if } k\neq l
\end{array}
\right.$$
Thus, for $k$ even and $l \neq k, k+1$, we find: 
$$
\left[(\mathbb{U}^{[a,b]}_\omega-z\mathrm{Id}) G_{\omega}^{[a,b]}(z)\right](k,l)=(\mathbb{U}^{[a,b]}_\omega-z\mathrm{Id})(k,\cdot)\cdot G_{\omega}^{[a,b]}(z,\cdot,l)=0 $$
This, given the 5-diagonal by blocks structure of the operator, results in:
\footnotesize
\begin{equation}
\alpha_k \gamma_{k-1} G_{\omega}^{[a,b]}(z,k-1,l)+(\alpha_k \delta_{k-1}-z\one ) G_{\omega}^{[a,b]}(z,k,l)+\beta_k\alpha_{k+1} G_{\omega}^{[a,b]}(z,k+1,l)+\beta_k\beta_{k+1}G_{\omega}^{[a,b]}(z,k+2,l)=0, \label{1}
\end{equation}
\normalsize
and similarly for \(k+1\):
$$\left[(\mathbb{U}^{[a,b]}_\omega-z\mathrm{Id})\cdot G_{\omega}^{[a,b]}(z)\right](k+1,l)=(\mathbb{U}^{[a,b]}_\omega-z\mathrm{Id})(k+1,\cdot)\cdot G_{\omega}^{[a,b]}(z,\cdot,l)=0. $$
This results in:
\footnotesize
\begin{equation}
\gamma_k \gamma_{k-1} G_{\omega}^{[a,b]}(z,k-1,l)+\gamma_k \delta_{k-1} G_{\omega}^{[a,b]}(z,k,l)+(\delta_k\alpha_{k+1}-z\one) G_{\omega}^{[a,b]}(z,k+1,l)+\delta_k\beta_{k+1}G_{\omega}^{[a,b]}(z,k+2,l)=0\label{2}
\end{equation}
\normalsize
Multiplying \eqref{1} by $\delta_k \beta_{k}^{-1} $ on the left and taking the difference one gets
\begin{align*}
 &   \left[ \delta_k\beta_k^{-1}\alpha_k\gamma_{k-1} -\gamma_k \gamma_{k-1}\right]G_{\omega}^{[a,b]}(z,k-1,l)+[\delta_k\beta_k^{-1}(\alpha_k\delta-z\one )-\gamma_k\delta_{k-1}]G_{\omega}^{[a,b]}(z,k,l) \\
  &  +[\delta_k\beta_k^{-1}\beta_k\alpha_{k+1}-(\delta_k\alpha_k-z\one )]G_{\omega}^{[a,b]}(z,k+1,l)
    +[\delta_k\beta_k^{-1}\beta_k\beta_{k+1}-\delta_k\beta_{k+1})]G_{\omega}^{[a,b]}(z,k+2,l)=0.
\end{align*}
After simplification, this yields: 
{\small $$ \left[(\delta_k\beta_k^{-1}\alpha_k- \gamma_k)\gamma_{k-1} \right] G_{\omega}^{[a,b]}(z,k-1,l) +[\delta_k\beta_k^{-1}(\alpha_k\delta_{k-1}-z\one)-\gamma_k\delta_{k-1}] G_{\omega}^{[a,b]}(z,k,l) +zG_{\omega}^{[a,b]}(z,k+1,l)=0.$$}
\noindent We have:
$$ (\alpha_k\gamma_k^{-1}\delta_k-\beta_k)=(\gamma_k^{-1})^*  \text{ and }  \beta_k^*\alpha_k=-\delta_k^*\gamma_k \Leftrightarrow  \alpha_k\gamma_k^{-1}=-(\beta_k^{-1})^* \delta_k^* .$$
This allows us to write:
\begin{equation}
     \left[(\gamma_k^{-1})^*\gamma_{k-1} \right]G_{\omega}^{[a,b]}(z,k-1,l)
    +(\gamma_k^{-1})^*[\delta_k+z\alpha_k^*]G_{\omega}^{[a,b]}(z,k,l) 
    +zG_{\omega}^{[a,b]}(z,k+1,l)=0.
\label{3.4}\end{equation}
A similar equality can be found by multiplying \eqref{2} by $ \alpha_k\gamma_k^{-1}$ and taking the difference:
\begin{align*}
&\left[ \alpha_k\gamma_k^{-1}\gamma_k\gamma_{k-1} -\alpha_k \gamma_{k}\right]
G_{\omega}^{[a,b]}(z,k-1,l)+[(\alpha_k\gamma_k^{-1}\gamma_k\delta_{k-1})-(\alpha_k\delta_{k-1}-z\one )]G_{\omega}^{[a,b]}(z,k,l)\\
 &  +[\alpha_k\gamma_k^{-1}(\delta_{k}\alpha_{k+1}-z\one )-\beta_k\alpha_{k+1}] G_{\omega}^{[a,b]}(z,k+1,l)
    +[  \delta_k\gamma_k^{-1}\delta_k\beta_{k+1}-\beta_k\beta_{k+1}]G_{\omega}^{[a,b]}(z,k+2,l)=0.
\end{align*}
Hence,
{\small $$ zG_{\omega}^{[a,b]}(z,k,l) +\left[(\alpha_k\gamma_k^{-1}\delta_k-\beta_k) \alpha_{k+1}-z\alpha_k\gamma_k^{-1}\right]G_{\omega}^{[a,b]}(z,k+1,l)+[(\alpha_k \gamma_k^{-1} \delta_k-\beta_k)\beta_{k+1}]G_{\omega}^{[a,b]}(z,k+2,l)=0.$$}
For any \(k\), we have:
$$\alpha_k\gamma_k^{-1}\delta_k-\beta_k=(\beta_k^{-1})^* \quad \text{and} \quad  \beta_k^*\alpha_k=-\delta_k^*\gamma_k \Leftrightarrow  \alpha_k\gamma_k^{-1}=-(\beta_k^{-1})^*\delta_k^*.$$
This allows us to write:
\begin{equation}
    zG_{\omega}^{[a,b]}(z,k,l)
     + (\beta_k ^{-1})^* \left[\alpha_{k+1} +z \delta_k^* \right]G_{\omega}^{[a,b]}(z,k+1,l)+ (\beta_k^{-1})^*\beta_{k+1}G_{\omega}^{[a,b]}(z,k+2,l)=0.
\label{3.5}\end{equation}
To find the equalities for the columns $l$ (while fixing $k$), it suffices to repeat the same process but start with the multiplication in the reverse direction:
\[
\left[G_{\omega}^{[a,b]}(z)(\mathbb{U}^{[a,b]}_\omega-z\mathrm{Id})\right](k,l) = \left\{
\begin{array}{ll}
    \one & \text{if } k=l \\
    0 & \text{if } k\neq l
\end{array}
\right..
\]
\end{proof}
To exploit these relations, we need to be able to invert $\delta_k+z\alpha_k^*$. This is the subject of the following lemma.
\begin{lemma}\label{lem_invertibility}
If $\alpha\in \mathrm{GL}_{\mathrm{L}}(\C)$, then for any $\epsilon > 0$, $k \in \mathbb{Z}$, and $z \in \mathbb{S}_\epsilon$, $\delta_k+z\alpha_k^*$ is $\mathbb{P}$-almost surely invertible.
\end{lemma}  
\begin{proof}
    Let \(z \in \mathbb{S}_\epsilon\) and \(k \in \mathbb{Z}\). We define:
    \begin{align}
 \Omega_{\mathrm{inv}} &= \{\omega \in \Omega; \ \delta_k + z\alpha_k^* \in \mathrm{GL}_{\mathrm{L}}(\C)\} \nonumber \\
             &= \{\omega \in \Omega; \ \widetilde{\rho}^{-1}[zV^{(k)}_\omega\alpha^* - \alpha^*U^{(k)}_\omega] \in\mathrm{GL}_{\mathrm{L}}(\C)\} \nonumber \\
             &= \{\omega \in \Omega; \ \det(zV^{(k)}_\omega\alpha^* - \alpha^*U^{(k)}_\omega) \neq 0\} \label{def_omega_inv}
 \end{align}
To show that \(\mathbb{P}(\Omega_{\mathrm{inv}}) = 1\), it suffices to demonstrate that:
$$\exists (U_0;V_0) \in \Omega; \ \det (zV_0\alpha^* - \alpha^* U_0) \neq 0$$
since we are working with the Haar measure $\nu_L$ of the phases and since \(\alpha\) is invertible.\\
If \(z \neq 1\), we set \((U_0,V_0) = (\one, \one )\). In this case:
$$\det (zV_0\alpha^* - \alpha^* U_0) = \det (z\alpha^* - \alpha^*) = (1-z)^L(-1)^L\det\alpha^* \neq 0.$$
If \(z = 1\), it suffices to take \((U_0,V_0) = (i\one, \one)\). In this case:
$$\det (zV_0\alpha^* - \alpha^* U_0) = \det (\alpha^* - i\alpha^*) = \det ((1-i)\alpha^*) = (1-i)^L\det\alpha^* \neq 0.$$
    \end{proof}
 
With Lemma \ref{lem_invertibility} one can isolate the term $G_{\omega}^{[a,b]}(z,k,l)$ in \eqref{eq_rel_resolvante}.

We define \(\Omega_{\mathrm{inv}}\) as in \eqref{def_omega_inv}.For \(\omega \in \Omega_{\mathrm{inv}}\), the equality \eqref{eq_rel_resolvante} becomes: 
\begin{align}
     G_{\omega}^{[a,b]}(z,k,l) &=\left({\gamma_k^{-1}}^*[\delta_k+z\alpha_k^*]\right)^{-1} \left[ {\gamma_k^{-1}}^*\gamma_{k-1} G_{\omega}^{[a,b]}(z,k-1,l)
+zG_{\omega}^{[a,b]}(z,k+1,l)\right]. \nonumber \\
&=[\delta_k+z\alpha_k^*]^{-1} \left[ {\gamma_{k-1} G_{\omega}^{[a,b]}(z,k-1,l)
+z\gamma_k^{-1}}^*G_{\omega}^{[a,b]}(z,k+1,l)\right].
\label{3.2} \end{align}
Hence, to estimate the expectancy $\| G_{\omega}^{[a,b]}(z,k,l)  \|^s$ in terms of $G_{\omega}^{[a,b]}(z,k-1,l)$ and $G_{\omega}^{[a,b]}(z,k+1,l)$, one needs to estimate the expectancy of $\|(\delta_k + z \alpha_k^*)^{-1}\|^s$. This is exactly Lemma \ref{lem1alphainversible}.

\begin{proof} of Lemma \ref{lem1alphainversible}. One has:
$$\mathbb{E}\left[\|(\delta_k + z \alpha_k^*)^{-1}\|^s\right] = \int_{\Omega} \|(-z V^{(k)}_\omega \alpha U^{(k)}_\omega + \alpha)^{-1}\|^s d\omega 
= \int_{\Omega} \frac{1}{(\sigma_L(\alpha - z V^{(k)}_\omega \alpha U^{(k)}_\omega))^s} d\omega 
$$
Here, \(\sigma_L\) represents the \(L^\text{th}\) singular value, which is the smallest singular value. Then, 
$$ \forall k\in \Z,\ \sigma_L\left(\alpha-z V^{(k)}_\omega \alpha U^{(k)}_\omega \right) = \sigma_L\left(\alpha-z e^{i\theta^{(k)}_\omega} \widehat{V}^{(k)}_\omega  \alpha \widehat{U}^{(k)}_\omega e^{i \Theta^{(k)}_\omega }\right).
$$
We can simplify some expressions. Indeed, we know that \(\theta^{(k)}_\omega \), \(\Theta^{(k)}_\omega \), \(\widehat{U}^{(k)}_\omega \), and \(\widehat{V}^{(k)}_\omega \) are mutually independent. Furthermore, \(\theta^{(k)}_\omega =\text{diag}(\theta^{(k)}_{\omega,1}, \ldots, \theta^{(k)}_{\omega,L}) \), and the \((\theta^{(k)}_{\omega,j})_{j\in\llbracket1,L\rrbracket}\) are i.i.d.. The same goes for \((\Theta^{(k)}_{\omega,j})_{j\in\llbracket1,L\rrbracket}\). Consequently, we can  replace \(\theta^{(k)}_\omega \) with \(\theta^{(k)}_{\omega,1}\one\) and \(\Theta^{(k)}_\omega\) with \(\Theta^{(k)}_{\omega,1} \one\) in the expectancy (and not deterministically!), which gives: 
\begin{align*}
\mathbb{E}\left(\sigma_L\left(\alpha-z e^{i\Theta^{(k)}_\omega} \widehat{V}^{(k)}_\omega \alpha \widehat{U}^{(k)}_\omega e^{i \theta^{(k)}_\omega}\right)^{-s}\right)&=\mathbb{E}\left(\sigma_L\left(\alpha-z e^{i\Theta^{(k)}_{\omega,1}} \widehat{V}^{(k)}_\omega \alpha \widehat{U}^{(k)}_\omega e^{i \theta^{(k)}_{\omega,1}}\right)^{-s}\right)\\ &=\mathbb{E}\left(\sigma_L\left(\alpha-z e^{i(\Theta^{(k)}_{\omega,1}+\theta^{(k)}_{\omega,1})} \widehat{V}^{(k)}_\omega \alpha \widehat{U}^{(k)}_\omega\right)^{-s}\right).   
\end{align*}
Since $\widehat{U}^{(k)}_\omega $ and $\widehat{U}^{(k)}_\omega $ are in \(U(L)\cap H_L(\mathbb{C})\), they have a square root which is also an element of \(U(L)\cap H_L(\mathbb{C})\), hence one writes \(\widehat{U}^{(k)}_\omega =(\widehat{U}^{(k)}_\omega )^{\frac{1}{2}}(\widehat{U}^{(k)}_\omega )^{\frac{1}{2}}\) and \(\widehat{V}^{(k)}_\omega =(\widehat{V}^{(k)}_\omega )^{\frac{1}{2}}(\widehat{V}^{(k)}_\omega )^{\frac{1}{2}}\) to obtain:
\begin{align*}
   & \mathbb{E}\left(\sigma_L\left(\alpha-z e^{i(\Theta^{(k)}_{\omega,1}+\theta^{(k)}_{\omega,1})}\widehat{V}^{(k)}_\omega \alpha \widehat{U}^{(k)}_\omega \right)^{-s}\right)=\mathbb{E}\left(\sigma_L\left(\alpha-z e^{i(\Theta^{(k)}_{\omega,1}+\theta^{(k)}_{\omega,1})} (\widehat{V}^{(k)}_\omega )^{\frac{1}{2}}(\widehat{V}^{(k)}_\omega )^{\frac{1}{2}} \alpha (\widehat{U}^{(k)}_\omega )^{\frac{1}{2}}(\widehat{U}^{(k)}_\omega )^{\frac{1}{2}}\right)^{-s}\right)\\
    &=\mathbb{E}\left(\sigma_L\left(\left((\widehat{V}^{(k)}_\omega )^{\frac{1}{2}}\right)^*\alpha \left((\widehat{U}^{(k)}_\omega )^{\frac{1}{2}}\right)^*-z  e^{i(\Theta^{(k)}_{\omega,1}+\theta^{(k)}_{\omega,1})}(\widehat{V}^{(k)}_\omega)^{\frac{1}{2}} \alpha (\widehat{U}^{(k)}_\omega)^{\frac{1}{2}}\right)^{-s}\right).
\end{align*}
Since \((\widehat{U}^{(k)}_\omega)^{\frac{1}{2}}\) and \((\widehat{V}^{(k)}_\omega)^{\frac{1}{2}}\) belong in particular to \(H_L(\C)\), we find that
{\small \begin{align*}
    \mathbb{E}\left(\sigma_L\left(\alpha-ze^{i(\Theta^{(k)}_{\omega,1}+\theta^{(k)}_{\omega,1})} \widehat{V}^{(k)}_\omega  \alpha\widehat{U}^{(k)}_\omega \right)^{-s}\right)   &=\mathbb{E}\left(\sigma_L\left(\left((\widehat{V}^{(k)}_\omega )^{\frac{1}{2}} \alpha (\widehat{U}^{(k)}_\omega )^{\frac{1}{2}}-z e^{i(\Theta^{(k)}_{\omega,1}+\theta^{(k)}_{\omega,1})} (\widehat{V}^{(k)}_\omega )^{\frac{1}{2}} \alpha (\widehat{U}^{(k)}_\omega )^{\frac{1}{2}}\right)^{-s}\right)\right)\\
    &=\mathbb{E}\left(\sigma_L\left((\widehat{V}^{(k)}_\omega )^{\frac{1}{2}}\left( \alpha -z e^{i(\Theta^{(k)}_{\omega,1}+\theta^{(k)}_{\omega,1})}\alpha \right)(\widehat{U}^{(k)}_\omega )^{\frac{1}{2}}\right)^{-s} \right)\\
    & = \mathbb{E} \left(\sigma_L\left((1-z e^{i(\Theta^{(k)}_{\omega,1}+\theta^{(k)}_{\omega,1})} )\alpha\right)^{-s}\right).
\end{align*}}
\noindent Furthermore, for every $\omega\in \Omega$, every $z\in \C\setminus \mathbb{S}^1$ and every $k\in \Z$,
$$     \sigma_L\left((1-ze^{i(\Theta^{(k)}_{\omega,1}+\theta^{(k)}_{\omega,1})}) \alpha\right)= |1-ze^{i(\Theta^{(k)}_{\omega,1}+\theta^{(k)}_{\omega,1})} |\sigma_L(\alpha) .$$
Therefore, we obtain the inequality for the partial expectation:
\begin{align*}
\mathbb{E}_{\theta,\Theta}[\|(\delta_k+z\alpha_k^*)^{-1}\|^s]
&\leq \int_{[0,2\pi]\times [0,2\pi]}\frac{1}{|1-ze^{i(\Theta^{(k)}_{\omega,1}+\theta^{(k)}_{\omega,1})} |^s.(\sigma_L(\alpha))^s }d(\theta,\Theta)\\
&= \frac{1}{(\sigma_L(\alpha))^s }\int_{0}^{2\pi}\frac{1}{|1-ze^{i\tilde{\theta}}|^s}d\tilde{\theta},
 \end{align*}
 since $e^{i(\Theta^{(k)}_{\omega,1}+\theta^{(k)}_{\omega,1})}$ follows the uniform law on $\mathbb{S}^1$. According to \cite[Lemma A.2]{hamza2007localization}, for any $z\in \C \setminus \mathbb{S}^1$, \(\int_{0}^{2\pi}\frac{1}{|1-ze^{i\theta}|^s}d\theta=C(s)<\infty\). Hence, by taking the expectation over \(U(L)\times \left\{-1,1\right\}^L\times U(L)\times \left\{-1,1\right\}^L\) of this partial expectation, we obtain the desired result.
\end{proof}
By combining the results of Lemma \ref{lmm3.2.1} and Lemma \ref{lem1alphainversible}, we can establish the lemma of reduction to the even case.
\begin{lemma}\label{lem_reduc_paire}
For \(\alpha\in \mathrm{GL}_{\mathrm{L}}(\C) \) and \(s \in \left(0,\frac{1}{2}\right)\), there exist $C_1(s), C_2(s) >0$ such that,
\begin{equation}\label{eq_ineq_lignes_resolv}
\mathbb{E}\left(\|G_{\omega}^{[a,b]}(z,k+1,l)\|^s\right) ^2\leq C_1(s)\mathbb{E}\left(\| G_{\omega}^{[a,b]}(z,k,l)\|^{2s} \right)  +C_2(s)\mathbb{E}\left(\|G_{\omega}^{[a,b]}(z,k+2,l)\|^{2s} \right).
\end{equation}
   Similarly, for the columns, there exist $\tilde{C}_1(s), \tilde{C}_2(s)>0$ such that
\begin{equation}\label{eq_ineq_col_resolv} 
 \mathbb{E}\left( \|G_{\omega}^{[a,b]}(z,k,l+1)\|^{s}\right)^{2} \leq \widetilde{C}_1(s)\mathbb{E}\left(\| G_{\omega}^{[a,b]}(z,k,l)\|^{2s}\right) \\
    +\widetilde{C}_2(s)\mathbb{E}\left(\|G_{\omega}^{[a,b]}(z,k,l+2)\|^{2s}\right).
\end{equation}
\end{lemma}
    \begin{proof}
For \(\omega \in \Omega_{\mathrm{inv}}\), equality \eqref{eq_rel_resolvante} becomes:
        \begin{align*}
     G_{\omega}^{[a,b]}(z,k,l) &=\left((\gamma_k^{-1})^* [\delta_k+z\alpha_k^*] \right)^{-1} \left[ (\gamma_k^{-1})^*\gamma_{k-1} G_{\omega}^{[a,b]}(z,k-1,l)
+zG_{\omega}^{[a,b]}(z,k+1,l)\right].\\
&=[\delta_k+z\alpha_k^*]^{-1} \left[ \gamma_{k-1} G_{\omega}^{[a,b]}(z,k-1,l)
+z(\gamma_k^{-1})^*G_{\omega}^{[a,b]}(z,k+1,l)\right].
\label{3.2} \end{align*}
This gives:
$$ \|G_{\omega}^{[a,b]}(z,k,l)\| \leq \|(\delta_k+z\alpha_k^*)^{-1}\| \left(\|\gamma_{k-1}.G_{\omega}^{[a,b]}(z,k-1,l)\| \\
    +\mid z\mid \|\gamma_{k}^*G_{\omega}^{[a,b]}(z,k+1,l)\|\right).$$
Now, we just need to use H\"older's inequality, Lemma \ref{lem1alphainversible} and Jensen's inequality as in \cite[Lemma 7.1]{hamza2007localization} to conclude.
    \end{proof}
Lemma \ref{lem_reduc_paire} implies the first reduction result.    
\begin{proof}(of Proposition \ref{prop_reduc_paire}).
The statement of Proposition \ref{prop_reduc_paire} is just a combination of \eqref{eq_ineq_lignes_resolv} and \eqref{eq_ineq_col_resolv}. 
\end{proof}

\subsubsection{Reduction to a suitable finite interval.}

We will now prove that it suffices to bound the norm of the finite Green's function moments for a restriction to a suitable finite interval to bound the Green's function on the infinite volume. This reduction is based on comparing the blocks of the infinite matrix with those of the finite matrix. Let $a,b\in \Z \cup \{\pm \infty\}$, $a<b$.

\begin{proposition}\label{prop_reduc_suitable_volume}
For $s\in(0,\frac12)$ and $\epsilon\in (0,1)$, there exists $C(s,\epsilon)>0$ such that
\begin{equation}\label{eq_prop_reduc_suitable_volume}
    \mathbb{E}\left( \|G_{\omega}^{[a,b]}(z,k,l) \|^s\right)^2\leq C(s,\epsilon) \mathbb{E} \left( \|G_{\omega}^{[k,l]}(z,k,l)^{2s}\|\right) 
\end{equation} 
for every $z\in \mathbb{S}_{\epsilon}$ and all $k,l\in[a+1,b-1]$, $|k-l|>4$.
\end{proposition}

\begin{proof}
We only deal with the case \(k=2n, l=2m\) where \(|k-l| \geq 4\) and \(m, n \in \mathbb{Z}\), the other cases being treated the same way.\\
Using the definition of \(\mathbb{U}_\omega^{[a, b]}\), we have
\begin{equation}\label{eq_decomp_U_suitable}
 \mathbb{U}_\omega^{[a, b]}=\mathbb{U}_\omega^{[a, 2n-1]} \oplus \mathbb{U}_\omega^{[2n, b]}+\Gamma_{n}^{e}   
\end{equation}
where \(\Gamma_{n}^{e}\) is given by
{\small \[
\Gamma_{n}^{e}(k, l)= \begin{cases}
   -\beta_{2n-2}\alpha+\beta_{2n-2}V, & k=2n-2, \  l=2n-1 \\
      -\beta_{2n-2}\beta_{2n-1}, & k=2n-2, \  l=2n \\
   -\delta_{2n-2}\alpha+\delta_{2n-2}V, & k=2n-1, \  l=2n-1 \\
    -\delta_{2n-2}\beta_{2n-1}, & k=2n-1, \  l=2n \\
        -\alpha_{2n+2}\gamma_{2n+1}, & k=2n, \  l=2n-1 \\
           -\alpha_{2n+2}\delta_{2n+1}+\alpha_{2n+2}U, & k=2n, \  l=2n \\
    -\gamma_{2n+2}\gamma_{2n+1}, & k=2n+1, \  l=2n-1 \\
   -\gamma_{2n+2}\delta_{2n+1}+\gamma_{2n+2}U, & k=2n+1, \  l=2n\\
    0, & \textrm{otherwise}.
\end{cases}
\]}
Let \(G_{\omega}^{n}(z)=G_{\omega}^{[a, 2n-1]}(z) \oplus G_{\omega}^{[2n, b]}(z)\). By the first resolvent identity, we have
\[
G_{\omega}^{[a, b]}(z)-G_{\omega}^{n}(z)=-G_{\omega}^{[a, b]}(z) \Gamma_{n}^{e} G_{\omega}^{n}(z).
\]
Consequently, for every \(m \geq n+2\), 
{\small \begin{align*}
 G_{\omega}^{[a, b]}(z,2n, 2m)  & =\left[1+(    \beta_{2n-2}\beta_{2n-1}) G_{\omega}^{[a, b]}(z,2n, 2n-2)+(\alpha_{2n+2}\gamma_{2n+1}) G_{\omega}^{[a, b]}(z,2n, 2n-1)  \right. \numberthis \label{eq1elyes} \\
    &+\left(\alpha_{2n+2}\delta_{2n+1}+\alpha_{2n+2}U\right) e^{-i \theta_{2n}} G_{\omega}^{[a, b]}(z,2n, 2n)\\
    &\left. +(\gamma_{2n+2}\delta_{2n+1}+\gamma_{2n+2}U) G_{\omega}^{[a, b]}(z,2n, 2n+1)\right]\cdot G_{\omega}^{[2n, b]}(z,2n, 2m).
\end{align*}}
Similarly, we have $\mathbb{U}_\omega^{[2n, b]}=\mathbb{U}_\omega^{[2n, 2m]} \oplus \mathbb{U}_\omega^{[2m+1, b]}+\Gamma_{m}^{o}$ where \(\Gamma_{m}^{o}\) is given by
{\small \[
\Gamma_{n}^{o}(k, l)= \begin{cases}
   -\alpha_{2m-2}\gamma_{2m-1}+V\gamma_{2m-1}, & k=2m, \  l=2m-1 \\
      -\alpha_{2m-2}\delta_{2m-1}+V\delta_{2m-1}, & k=2m, \  l=2m \\
   -\beta_{2m}\alpha_{2m+1}, & k=2m, \  l=2m+1 \\
    -\beta_{2m}\beta_{2m+2}, & k=2m, \  l=2m+2 \\
        -\gamma_{2m+1}\gamma_{2m-1}, & k=2m+1, \  l=2m-1 \\
              -\gamma_{2m+1}\delta_{2m}+U\alpha_{2m+1}, &k=2m+1, \  l=2m \\
    -\delta_{2m+1}\alpha_{2m+1}+U\alpha_{2m+1}, & k=2m+1, \  l=2m+1 \\
   -\delta_{2m+1}\beta_{2m+2}  +U\beta_{2m+2}, & k=2m+1, \  l=2m+2\\
    0, & \textrm{otherwise.}
\end{cases}
\]}
Now, if we define \(G_{\omega}^{m}(z)=G_{\omega}^{[2n, 2m]}(z) \oplus G_{\omega}^{[2m+1, b]}(z)\), once again, we observe that
\[
G_{\omega}^{[2n, b]}(z)-G_{\omega}^{m}(z)=-G_{\omega}^{m}(z) \Gamma_{m}^{o} G_{\omega}^{[2n, b]}(z).
\]
Thus, for any \(m \geq n+2\)
{\small \begin{align*}
     G_{\omega}^{[2n, b]}(z,2n, 2m) \numberthis \label{eq2elyes} & = \left[ 1- [-(\alpha_{2m-2}\gamma_{2m-1}+V\gamma_{2m-1})G_{\omega}^{[2n, b]}(z,2m-1,2m) \right. \\
     & -\left(\alpha_{2m-2}\delta_{2m-1}+V\delta_{2m-1}\right) G_{\omega}^{[2n, b]}(z,2m, 2m) \\
    &  -( \gamma_{2m+1}\delta_{2m}+\alpha_{2n+2}U)G_{\omega}^{[2n, b]}(z,2m+1,2m) \\ & \left. - (\beta_{2m}\beta_{2m+2}) G_{\omega}^{[2n, b]}(z,2m+2,2m)]\right] G_{\omega}^{[2n, 2m]}(z,2n, 2m) 
\end{align*}}
Inserting \eqref{eq2elyes} into equation \eqref{eq1elyes}, then applying Hölder's inequality and Theorem \ref{thm_unif_bounded}, we obtain for all $s\in (0, \frac12)$ and $\epsilon\in (0,1)$,
$$\left(\mathbb{E}\left(\left\|G_\omega^{[a, b]}(z,2 n, 2 m)\right\|^s\right)\right)^2 \leq C(s, \epsilon) \mathbb{E}\left(\left\|G_z^{[2 n, 2 m]}(z,2 n, 2m+1)\right\|^{2 s}\right).$$
This yields the required result for all $n, m$ such that $m \geq n+2 $. The proof is analogous in the case where $n\geq m+2$.
\end{proof}

\noindent Lemma \ref{casfini} is now a direct application of Proposition \ref{prop_reduc_suitable_volume} with $a=-\infty$ and $b=+\infty$.

\noindent Note that one can take for example $\epsilon =\frac12$ to obtain in \eqref{eq_prop_reduc_suitable_volume} a constant $C(s)$ independent in $\epsilon$ as it is done in \cite{hamza2009dynamical}.

\subsection{Exponential decay of the reduced case.}
Now that we have simplified our study to the fractional moments of elements of the form \(G_{\omega}^{[2n,2m+1]}(z,2n,2m+1)\), we are ready to show that their expectations decrease exponentially, which is the statement of Theorem \ref{thm_exp_decay_even_odd}.

\subsubsection{An explicit expression of the Green's function in terms of transfer matrices.} Given the matching between the block dimensions in the Green's matrix and those in the transfer matrices, we formulate an expression of the Green's function. This expression relies directly on the dimensions of the blocks involved in the product of transfer matrices.

Inspired by the notations in \cite{marin2013scattering}, for $m,n \in \Z$, $n\leq m$, we define:
\begin{equation}\label{def_blocks_product_trans_mat}
   T_{2n}^{2m+1}(z):= T_{\omega}^{(2m+1)}(z)\cdots T_{\omega}^{(2n)}(z) = \begin{pmatrix}
    A^{2m+1}_{2n}(z) &  B^{2m+1}_{2n}(z)U \\
     C^{2m+1}_{2n}(z) &  D^{2m+1}_{2n}(z)U\end{pmatrix} 
\end{equation} 
where $A^{2m+1}_{2n}(z)$ denote the left upper $L\times L$ block of the product of transfer matrices and the same for the three other blocks. 

Note that we remove the $\omega$ dependency from the various quantities since the random character plays no role in the following discussion.
\begin{lemma}\label{lmm4.5.1}
For any $\epsilon\in (0,1)$ and any $z\in \mathbb{S}_{\epsilon}$, any $U,V\in U(L)$ and any $m,n\in \Z$ such that $|m-n|>2$:
  \begin{enumerate}
      \item  the following quantities are invertible: 
      $$\textrm{(a)}\ C^{2m+1}_{2n}(z)\pm VA^{2m+1}_{2n}(z),\ \textrm{(b)}\ B^{2m+1}_{2n}(z)U\pm A^{2m+1}_{2n}(z),\  \textrm{(c)}\ D^{2m+1}_{2n}(z)U\pm C^{2m+1}_{2n}(z).$$
\item The following quantities are inside the Siegel disk $\mathbb{D}_L=\left\{Z\in \mathcal{M}_L(\mathbb{C})\ |\ ZZ^*<1\right\}:$
\begin{enumerate}
     \item   $ (C^{2m+1}_{2n}(z)-VA^{2m+1}_{2n}(z))^{-1} (VB^{2m+1}_{2n}(z)-D^{2m+1}_{2n}(z)),$
      \item  $ (D^{2m+1}_{2n}(z) U-C^{2m+1}_{2n}(z))(B^{2m+1}_{2n}(z)U-A^{2m+1}_{2n}(z))^{-1},$
        \item  $ (D^{2m+1}_{2n}(z) U+C^{2m+1}_{2n}(z))(B^{2m+1}_{2n}(z)U+A^{2m+1}_{2n}(z))^{-1}.$
\end{enumerate}
  \end{enumerate}
  \end{lemma}
\begin{proof}
    $1(a)$ and $2(a)$ follow directly from Theorem 1 in \cite{marin2013scattering}. \\
 To prove $1(b)$ and $2(b)$, we apply Theorem 1 of \cite{marin2013scattering} to 
 $$ (T_{2n}^{2m+1}(z))^{-1}=\mathcal{L} (T_{2n}^{2m+1}(\tfrac{1}{\bar{z}}))^*\mathcal{L}
 =\begin{pmatrix}
    (A^{2m+1}_{2n}(\frac{1}{\Bar{z}}))^* &    -(C^{2m+1}_{2n}(\frac{1}{\Bar{z}}))^*  \\
  -(B^{2m+1}_{2n}(\frac{1}{\Bar{z}})U)^*  &  (D^{2m+1}_{2n}(\frac{1}{\Bar{z}})U)^*
    \end{pmatrix}.$$
Dropping the $2n$ and $2m+1$ indices, we find that $-V(A(\frac{1}{\Bar{z}}))^*-(B(\frac{1}{\Bar{z}})U)^*$ is invertible. Therefore, $A(\frac{1}{\Bar{z}})V^*+B(\frac{1}{\Bar{z}})U$ is invertible. Since the mapping $z \mapsto \frac{1}{\Bar{z}}$ is an involution of $\mathbb{C}\setminus \{0 \}$, we can deduce that $B(z')U+ A(z')V^*$ is invertible for all $z' \in \mathbb{C}\setminus \{0 \}$. Since this result holds for any initial conditions $U$ and unitary $V$, we take $V^*=\one$ or $V^*=-\one$ to prove $1(b)$. Similarly, $2(b)$ and $2(c)$ can be shown.

To prove $1(c)$ we follow the proof of \cite[Theorem 1]{marin2013scattering} but with changing the vector $\left( \begin{smallmatrix}
1 \\ 0    
\end{smallmatrix}\right)$ into the vector  $\left( \begin{smallmatrix}
0 \\ 1    
\end{smallmatrix}\right)$ in the defintion of $\Phi$. Then we apply \cite[Lemma 3]{marin2013scattering} with $\Phi^*\mathcal{L}\Phi <0$ instead of $\Phi^*\mathcal{L}\Phi >0$ which still implies the invertibility of $DU\pm C$. But since we are in the negative case, we have that $(BU\pm A)(DU\pm C)^{-1}$ is not in $\mathbb{D}_L$.

\end{proof}

\begin{lemma}\label{lemma_formula_resolv} Let $\epsilon\in (0,1)$. For all $z\in \mathbb{S}_{\epsilon}$ and all $m,n\in \Z$, $|m-n|>2$,
{\small \begin{equation}\label{eq_lemma_formula_resolv1}
G_\omega^{[2n,2m+1]}(z,2n,2m+1)=  \left( (C^{2m+1}_{2n}-VA^{2m+1}_{2n})+(D^{2m+1}_{2n}-VB^{2m+1}_{2n}) U \right)^{-1}(L^{2m+1}_{2m}-VK^{2m+1}_{2m})
\end{equation}
\begin{equation}\label{eq_lemma_formula_resolv2}
G_\omega^{[2n+1,2m+1]}(z,2n+1,2m+1) =  \left( (C^{2m+1}_{2n+1}-VA^{2m+1}_{2n+1})+z(D^{2m+1}_{2n+1}-VB^{2m+1}_{2n+1}) U^* \right)^{-1}(L^{2m+1}_{2m}-VK^{2m+1}_{2m})
\end{equation}}
with 
\begin{equation}\label{eq_def_K2m12n}
K^{2m+1}_{2m} =\frac{1}{z}V^{(2m+1)}_{\omega}   \widetilde{\rho}^{-1} V^{(2m)}_{\omega}  \widetilde{\rho}^{-1} +V^{(2m+1)}_{\omega}    \widetilde{\rho}^{-1} \alpha^*  (U^{(2m)}_{\omega})^*   \alpha  \widetilde{\rho}^{-1}    
\end{equation}
\begin{equation}\label{eq_def_L2m12n}
L^{2m+1}_{2m} = (U^{(2m+1)}_{\omega})^* \alpha \widetilde{\rho}^{-1} V^{(2m)}_{\omega} \widetilde{\rho}^{-1} +z  (U^{(2m+1)}_{\omega})^*   \widetilde{\rho}^{-1}  (U^{(2m)}_{\omega})^*  \alpha  \widetilde{\rho}^{-1}    
\end{equation}
with $U$, $V$ $ \in \mathrm{U(L)}$, the two boundary conditions at $2n$ and $2m+1$ as defined in \eqref{def_W_2n2m1}.
\end{lemma}

\begin{proof}
For \eqref{eq_lemma_formula_resolv1}, we know that the quantity $G_\omega^{[2n,2m+1]}(z,2n,2m+1)$ corresponds to the $\phi_{2n}$ component of the solution $\Phi$ of the equation $(\mathbb{U}^{[2n,2m+1]}_{\omega}-z)\Phi = \xi $ for $\xi_k =\delta_{2m+1,k}\one. $ We apply Lemma 4 of \cite{marin2013scattering}:
\begin{equation}\label{eq_expr_green_mat_trans1}
  \begin{pmatrix}
\phi_{2m+1}\\
\psi_{2m+1}
\end{pmatrix}=\begin{pmatrix}
\phi_{2m+1}\\
V\phi_{2m+1}
\end{pmatrix}=  T_{2n}^{2m+1}(z) \begin{pmatrix}
    \one & 0\\
    0 &  U
    \end{pmatrix}\begin{pmatrix}
\phi_{2n}\\
\phi_{2n}
\end{pmatrix}+T_{2m}^{2m+1}(z)\begin{pmatrix}
    -z^{-1}\one \\
    0 
    \end{pmatrix}.    
\end{equation}
  
We compute $T_{2m}^{2m+1}(z)\left(\begin{smallmatrix}
    -z^{-1}\one \\
    0 
    \end{smallmatrix}\right)$ to retrieve the vector $\left( \begin{smallmatrix}
K^{2m+1}_{2m}\\
L^{2m+1}_{2m}
\end{smallmatrix} \right)$ given by \eqref{eq_def_K2m12n} and \eqref{eq_def_L2m12n} and using the notations introduced at \eqref{def_blocks_product_trans_mat}, \eqref{eq_expr_green_mat_trans1} one gets
\begin{equation}\label{eq_expr_green_mat_trans2}
\begin{pmatrix}
\phi_{2m+1}\\
V\phi_{2m+1}
\end{pmatrix}=\begin{pmatrix}
A^{2m+1}_{2n} &  B^{2m+1}_{2n}\\
 C^{2m+1}_{2n} &  D^{2m+1}_{2n}
\end{pmatrix} \begin{pmatrix}
\one & 0\\
0 &  U
\end{pmatrix} \begin{pmatrix}
\phi_{2n}\\
\phi_{2n}
\end{pmatrix}
+\begin{pmatrix}
K^{2m+1}_{2m}\\
L^{2m+1}_{2m}
\end{pmatrix}.
\end{equation}
We multiply the first row by $V$ and calculate the difference of the two lines to find:
\begin{equation}\label{eq_phi2n}
\phi_{2n}= \left( (C^{2m+1}_{2n}-VA^{2m+1}_{2n})+(D^{2m+1}_{2n}-VB^{2m+1}_{2n}) U \right)^{-1}(L^{2m+1}_{2m}-VK^{2m+1}_{2m}). 
\end{equation}
 For \eqref{eq_lemma_formula_resolv2}, we follow the same steps. Note that $U\psi_{2n+1}=z\phi_{2n+1}+\xi_{2n+1}=z\phi_{2n+1}$ implies that $\psi_{2n+1}=zU^*\phi_{2n+1}$. This yields:
$$\begin{pmatrix}
\phi_{2m+1}\\
V\phi_{2m+1}
\end{pmatrix}= \begin{pmatrix}
    A^{2m+1}_{2n+1} &  B^{2m+1}_{2n+1}\\
     C^{2m+1}_{2n+1} &  D^{2m+1}_{2n+1}
    \end{pmatrix} \begin{pmatrix}
    I & 0\\
    0 &  zU^*
    \end{pmatrix}\begin{pmatrix}
\phi_{2n+1}\\
\phi_{2n+1}
\end{pmatrix}+\begin{pmatrix}
K^{2m+1}_{2m}\\
L^{2m+1}_{2m}
\end{pmatrix}.$$
We multiply by $V$ and then take the difference of the two lines to find the expression given in the statement.
\end{proof}

\begin{remark}
For the other cases, $G_\omega^{[2n+1,2m]}(z,2n+1,2m)$ and $G_\omega^{[2n,2m]}(z,2n,2m)$, the equations $\mathbb{V}_{\omega}\Psi=z\Phi+\xi$ and $\mathbb{W}_{\omega}\Phi=\Psi,$ yield $V\psi_{2m+1}=z\phi_{2m+1}+\one$ or $\psi_{2m+1}=zV^*\phi_{2m+1}+V^*.$ So, we need to multiply the first row by $zV^*$, add $V^*$, and take the difference to obtain:
$$ G_\omega^{[2n,2m]}(z,2n,2m)= [C^{2m}_{2n}-zV^*A^{2m}_{2n}-V^*+(D^{2m}_{2n}-zV^*B^{2m}_{2n}-V^*)U]^{-1}(L^{2m}_{2m}-V^*K^{2m}_{2m}-V^*)$$
with  $\left( \begin{smallmatrix}
K^{2m}_{2m}\\
L^{2m}_{2m}
\end{smallmatrix}\right) =  T^{(2m)}_{\omega}(z)\left(\begin{smallmatrix}
    -z^{-1}\one \\
    0 
    \end{smallmatrix}\right)$. As we will see later, the presence of the three extra terms $-V^*$ will prevent to prove directly the exponential decay for these terms which is why we prove before the Lemma \ref{lem_reduc_paire} so that we don't have to estimate directly these terms.
\end{remark}

\begin{lemma}\label{lem_kappa}
Let $\epsilon\in (0,1)$. There exists $\kappa_{\epsilon} >0$ such that for every $z\in \mathbb{S}_{\epsilon}$:
\begin{equation}\label{eq_lem_kappa}
 \|G_\omega^{[2n,2m+1]}(z,2n,2m+1)\|\leq \kappa_{\epsilon} \| \left( (C^{2m+1}_{2n}-VA^{2m+1}_{2n})+(D^{2m+1}_{2n}-VB^{2m+1}_{2n}) U \right)^{-1}\|.     
\end{equation}
\end{lemma}

\begin{proof}
It suffices to prove that $ \|L^{2m+1}_{2m}-VK^{2m+1}_{2m}\|$ is uniformly bounded in $m$. One has
  
\begin{align*}
  \|L^{2m+1}_{2m}-VK^{2m+1}_{2m}\|&= \left\| (U^{(2m+1)}_{\omega})^* \alpha \widetilde{\rho}^{-1} V^{(2m)}_{\omega} \widetilde{\rho}^{-1} +z  (U^{(2m+1)}_{\omega})^*   \widetilde{\rho}^{-1}  (U^{(2m)}_{\omega})^*  \alpha  \widetilde{\rho}^{-1} \right.  \\
  & \left. -V\left( \frac{1}{z}V^{(2m+1)}_{\omega}   \widetilde{\rho}^{-1} V^{(2m)}_{\omega}  \widetilde{\rho}^{-1} +V^{(2m+1)}_{\omega}    \widetilde{\rho}^{-1} \alpha^*  (U^{(2m)}_{\omega})^*   \alpha  \widetilde{\rho}^{-1}     \right)\right\|  \\
          &  \leq \|   \widetilde{\rho}^{-1} \|^2 \left[ \|\alpha\| .\|\rho^{-1}\|+\frac{1}{\mid z\mid} +\|\alpha\|\| \widetilde{\rho}^{-1}\| ^2 (\mid z \mid +\|\alpha\| )\right]  \end {align*}
            
            For $\|\alpha\|<1$ and $\epsilon\in (0,1)$, by Lemma \ref{lem_estime_rho} and since $z\in \mathbb{S}_\epsilon$:
$$
 \|L^{2m+1}_{2m}-VK^{2m+1}_{2m}\| \leq \frac{1}{1-\| \alpha\| ^2} \left[ \frac{{\| \alpha\|}}{\sqrt{1-\|\alpha\|^2}} +\frac{1}{1-\epsilon}+(1-\epsilon)\|\alpha\|+\|\alpha\|^2\right] := \kappa_{\epsilon}.$$
\end{proof}

\begin{remark}
Note that since in Theorem \ref{thm_exp_decay_even_odd} $\epsilon \in (0,\epsilon_0]$ and $\epsilon_0 \in (0,1)$ is fixed by Corollary \ref{corollaire_pos_Lyap_couronne}, the term $\frac{1}{1-\epsilon}$ in $\kappa_{\epsilon}$ is in the interval $(1,\frac{1}{1-\epsilon_0})$ and thus is bounded.    
\end{remark}

To estimate the Green matrix using transfer matrices presents notable complexity: in the scalar context, this process is made possible through the property \(\mid a \cdot b \mid = \mid a \mid \cdot \mid b \mid\). However, in the matrix context, this direct relationship does not hold, as we only have sub-multiplicativity of norms. Thus, we opt for a different method.

\subsubsection{Estimate of the Green kernel by products of transfer matrices.}
Our first step is to establish a lower bound on the norms of the inverses of transfer matrices, which corresponds to the norms of the transfer matrices themselves. Hence we have to prove that the norm of a new term needs to be uniformly bounded by a constant.

\begin{lemma}\label{lemma_lower_bounded_inv_trans_mat}
Let $\epsilon \in (0,1)$. There exists $\kappa_{\epsilon} > 0$ such that for all $z\in \mathbb{S}_\epsilon \setminus \mathbb{S}^1$ and all $m,n\in \mathbb{Z}$ such that $|m-n|>2$ : 
\begin{equation}\label{eq_lemma_lower_bounded_inv_trans_mat}
    \|(T_{2n}^{2m+1}(z))^{-1}\| \geq \kappa_{\epsilon} \frac{\|G_\omega^{[2n,2m+1]}(z,2n,2m+1)\|}{\|H^{2m+1}_{2n}(F^{2m+1}_{2n})^{-1}-G^{2m+1}_{2n} (E^{2m+1}_{2n})^{-1}\|} 
\end{equation}
where \begin{align}
E^{2m+1}_{2n}= & (C^{2m+1}_{2n}-VA^{2m+1}_{2n}) +(D^{2m+1}_{2n}-VB^{2m+1}_{2n})U. \label{eq_def_E}\\
F^{2m+1}_{2n}= & (D^{2m+1}_{2n}-VB^{2m+1}_{2n})U-(C^{2m+1}_{2n}-VA^{2m+1}_{2n}). \label{eq_def_F}\\
G^{2m+1}_{2n}= & A^{2m+1}_{2n}+V^*C^{2m+1}_{2n} + (B^{2m+1}_{2n}+V^*D^{2m+1}_{2n})U.\label{eq_def_G} \\
H^{2m+1}_{2n}= & (B^{2m+1}_{2n}+V^*D^{2m+1}_{2n})U-(A^{2m+1}_{2n}+V^*C^{2m+1}_{2n}). \label{eq_def_H}
\end{align}
\end{lemma}
\begin{remark}
Since in Theorem \ref{thm_exp_decay_even_odd} $\epsilon \in (0,\epsilon_0]$, once $\epsilon_0\in (0,1)$ is fixed, $\kappa_{\epsilon}$ can be taken uniform in $\epsilon$. Thus $\kappa$ can be taken such that it only depends on the norm of the matrix $\alpha$.
\end{remark}

\begin{proof}
Throughout this proof, we will omit the indices $2n$ and $2m+1$ for brevity. With a direct computation, we get with the expressions given in \eqref{eq_def_E}, \eqref{eq_def_F}, \eqref{eq_def_G} and \eqref{eq_def_H}, 
\begin{align*}
(T_{2n}^{2m+1}(z))^{-1}
   & =\left(\begin{smallmatrix}
    A &  BU\\
     C &  DU
    \end{smallmatrix}\right)^{-1}\
    =\left(\tfrac{1}{\sqrt{2}} \left(  \begin{smallmatrix}
    -V &  \one\\
    \one  &  V^*
    \end{smallmatrix}\right)\right)^{-1}
\left(\begin{smallmatrix}
(C-VA) +(D-VB)U     &  (D-VB)U-(C-VA)\\
A+V^*C + (B+V^*D)U  & (B+V^*D)U-(A+V^*C)  
    \end{smallmatrix}\right)^{-1}
\left( \tfrac{1}{\sqrt{2}}\left(\begin{smallmatrix}
   \one &  -\one\\
   \one &  \one
    \end{smallmatrix}\right)\right)^{-1}
\\
& = \left(\tfrac{1}{\sqrt{2}} \left(  \begin{smallmatrix}
    -V &  \one\\
    \one  &  V^*
    \end{smallmatrix}\right)\right)^{-1} \left( \begin{matrix}
    E&  F\\
    G & H
    \end{matrix}\right)^{-1} \left( \tfrac{1}{\sqrt{2}}\left(\begin{smallmatrix}
   \one &  -\one\\
   \one &  \one
    \end{smallmatrix}\right)\right)^{-1}.   
\end{align*}
Using Lemma \ref{lmm4.5.1} $1(a)$ and $2(a)$ and the Neumann power serie lemma,
$E = (C - VA) + (D - VB)U$ is invertible, so we can compute the inverse of $\left(\begin{smallmatrix} E & F \\ G & H \end{smallmatrix}\right)$ using the Schur complement :
$$
\begin{pmatrix}
E & F \\
G & H
\end{pmatrix}^{-1} =
\begin{pmatrix}
E^{-1} + E^{-1}F\left(\frac{M}{E}\right)^{-1}GE^{-1} & E^{-1}F\left(\frac{M}{E}\right)^{-1} \\
\left(\frac{M}{E}\right)^{-1}GE^{-1} & \left(\frac{M}{E}\right)^{-1}
\end{pmatrix},
$$
where $\left(\frac{M}{E}\right) = H - GE^{-1}F$ is the Schur complement of $E$. Thus, we obtain:
  $$ (T_{2n}^{2m+1}(z))^{-1} = \left(\tfrac{1}{\sqrt{2}} \left(  \begin{matrix}
    -V &  \one\\
    \one  &  V^*
    \end{matrix}\right)\right)^{-1}\begin{pmatrix}
    E^{-1}+E^{-1}F(\tfrac{M}{E})^{-1}GE^{-1}&  E^{-1}F(\tfrac{M}{E})^{-1}\\
    (\tfrac{M}{E})^{-1}GE^{-1} & (\tfrac{M}{E})^{-1}
    \end{pmatrix}\left( \tfrac{1}{\sqrt{2}}\left(\begin{matrix}
   \one &  -\one\\
   \one &  \one
    \end{matrix}\right)\right)^{-1}.$$ 
By taking the Frobenius norm, we obtain:
\begin{equation}\label{eq_norm_inverse_schur}
 \| (T_{2n}^{2m+1}(z))^{-1}\| = \left\| \begin{pmatrix}
E^{-1}+E^{-1}F\left(\frac{M}{E}\right)^{-1}GE^{-1} & E^{-1}F\left(\frac{M}{E}\right)^{-1} \\
\left(\frac{M}{E}\right)^{-1}GE^{-1} & \left(\frac{M}{E}\right)^{-1}
\end{pmatrix} \right\|   
\end{equation}
since the matrices $\tfrac{1}{\sqrt{2}} \left(  \begin{smallmatrix}
    -V &  \one\\
    \one  &  V^*
    \end{smallmatrix}\right)$ and $ \tfrac{1}{\sqrt{2}}\left(\begin{smallmatrix}
   \one &  -\one\\
   \one &  \one
    \end{smallmatrix}\right)$ are unitary. We know that the Frobenius norm of a square matrix is always greater than or equal to the sum of the norms of its blocks. Therefore : 
    $$
  \left\|\begin{pmatrix}
    E^{-1}+E^{-1}F(\tfrac{M}{E})^{-1}GE^{-1}&  E^{-1}F(\tfrac{M}{E})^{-1}\\
    (\tfrac{M}{E})^{-1}GE^{-1} & (\tfrac{M}{E})^{-1}
    \end{pmatrix}\right\| \geq \|  E^{-1}F(\tfrac{M}{E})^{-1} \|.
    $$
Furthermore, by the sub-multiplicativity of the Frobenius norm, 
\begin{equation}\label{eq_frobenius_submult}
  \|E^{-1}F\left(\tfrac{M}{E}\right)^{-1}\|\geq \frac{\|E^{-1}F\left(\frac{M}{E}\right)^{-1}\left(F\left(\frac{M}{E}\right)^{-1}\right)^{-1}\|}{\left\|\left(F\left(\frac{M}{E}\right)^{-1}\right)^{-1}\right\|} = \frac{\|E^{-1}\|}{\left\|\left(\frac{M}{E}\right)F^{-1}\right\|}.  
\end{equation}
Combining \eqref{eq_norm_inverse_schur} and \eqref{eq_frobenius_submult}, we have the following inequality:
$$
\|(T_{2n}^{2m+1}(z))^{-1}\|\geq \frac{\|E^{-1}\|}{\left\|\left(\frac{M}{E}\right)F^{-1}\right\|}.
$$
To conclude, it suffices to observe that $\left(\tfrac{M}{E}\right)F^{-1}= H F^{-1}-G E^{-1}$  and, using \eqref{eq_def_E} and Lemma \ref{lem_kappa}, $\kappa \|G_\omega^{[2n,2m+1]}(z,2n,2m+1)\|\leq  \|E^{-1} \|,
$ where $\kappa$ is the constant obtained in Lemma \ref{lem_kappa} for $\epsilon = \epsilon_2$.
\end{proof}

With Lemma \ref{lemma_lower_bounded_inv_trans_mat}, to control the norm of the Green kernel by the norm of $(T_{2n}^{2m+1}(z))^{-1}$, it suffices to prove that there exists a constant $C_{\alpha}>0$ such that: 
\begin{equation}\label{eq_bound_HF-GE}
\forall m,n\in \Z,\ |m-n|\geq 2,\ \|H^{2m+1}_{2n}(F^{2m+1}_{2n})^{-1}-G^{2m+1}_{2n} (E^{2m+1}_{2n})^{-1}\|\leq C_{\alpha}.
\end{equation}
To achieve this, we propose a strategy which consists in constructing a sequence of matrices which will follow a sub-arithmetico-geometric progression. This particular progression eventually reaches a maximum threshold, meaning it does not exceed a certain limit, provided that the value of  the norm of \(\alpha\) meet a specific condition. We set $\|.\|$ as the Frobenius norm on $\mathcal{M}_{\mathrm{L}}(\C)$.
\begin{proposition}\label{prop_Lambda_Mu}
There exists $r_0\in (0,1)$, $p_0>2$ and $C_{r_0,p_0}>0$ such that for every $\| \alpha \|\leq r_0$, every $\epsilon\in (0,1)$, every $z\in \mathbb{S}_\epsilon$, and every $m,n\in \Z$, $|m-n|>p_0$, 
\begin{equation}\label{eq_lemma_bound_HF-GE}
\|H^{2m+1}_{2n}(F^{2m+1}_{2n})^{-1}-G^{2m+1}_{2n} (E^{2m+1}_{2n})^{-1}\|\leq C_{r_0,p_0}.
\end{equation}
\end{proposition}


\begin{proof}\textbf{Step 1} During the next computations, we drop the indices $2n$, $2m+1$. Using \eqref{eq_def_E}, \eqref{eq_def_F}, \eqref{eq_def_G} and \eqref{eq_def_H}, 
\begin{align*}
HF^{-1}-GE^{-1} & =\left((B+V^*D)U-(A+V^*C)\right) \left((D-VB)U-(C-VA)\right)^{-1} \\
& \qquad -\left( A+V^*C + (B+V^*D)U \right)\left((D-VB)U+(C-VA) \right)^{-1} \\
& =\left(BU-A + V^*(DU-C) \right)\left(DU-C-V(BU-A)\right)^{-1} \\
& \qquad -\left(BU+A + V^*(DU+C) \right)\left(DU+C-V(A+BU) \right)^{-1}.
\end{align*}
According to Lemma \ref{lmm4.5.1}, $BU-A$, $BU+A$, $DU-C$ and $DU+C$ are invertible and we find:
\begin{align}
HF^{-1}-GE^{-1} & = - V^* \left(\one - (DU-C)(BU-A)^{-1}V^*\right)^{-1} + V^*\left( \one - V(BU-A)(DU-C)^{-1} \right)^{-1} \nonumber \\
& + V^* \left(\one - (DU+C)(BU+A)^{-1}V^* \right)^{-1} - V^*\left(\one - V(BU+A)(DU+C)^{-1} \right)^{-1}. \label{eq_simplify_HFGE1}
\end{align}
If one sets $M=(DU-C)(BU-A)^{-1}V^*$ and $N=(DU+C)(BU+A)^{-1}V^*$, by Lemma \ref{lmm4.5.1} they are both invertible and they both lie in $\mathbb{D}_L$. Therefore $\one-M$ and $\one-N$ are invertible and the same for $\one-M^{-1}$ and $\one-N^{-1}$. Moreover, one has $(\one-M^{-1})^{-1}=-M(\one-M)^{-1}$ and  $(\one-N^{-1})^{-1}=-N(\one-N)^{-1}$.
Rewriting \eqref{eq_simplify_HFGE1}, one gets
\begin{equation}\label{eq_simplify_HFGE2}
HF^{-1}-GE^{-1} = V^* \left( -(\one +M) (\one - M)^{-1}  + (\one + N) (\one - N)^{-1}  \right).
\end{equation}
Applying the Frobenius norm and since $V^*$ is unitary and $M,N\in \mathbb{D}_L$,
\begin{equation}\label{eq_simplify_HFGE_norm1}
\| HF^{-1}-GE^{-1}\| \leq  2\left( \| (\one - M)^{-1} \| + \| (\one - N)^{-1} \| \right) .   
\end{equation}
\textbf{Step 2} We will begin by bounding $\| (\one - M_{2n}^{2m+1})^{-1} \|$. Since by Lemma \ref{lmm4.5.1}, $M_{2n}^{2m+1}\in \mathbb{D}_L$, 
$$\| (\one - M_{2n}^{2m+1})^{-1} \|  \leq \frac{1}{1-\|M_{2n}^{2m+1}\|}.$$
To show that the upper bound $\frac{1}{1-\|M_{2n}^{2m+1}\|}$ is bounded, it suffices to prove the existence of a constant \( C \in (0,1) \) independent of \( n \) and \( m \), such that for every  $m,n\in \Z$, $|m-n|$ large enough,
\begin{equation}\label{eq_borne_unif_M}
 \|M_{2n}^{2m+1}\|=\|(D^{2m+1}_{2n}U-C^{2m+1}_{2n})(B^{2m+1}_{2n}U-A^{2m+1}_{2n})^{-1}V^*\| \leq C.    
\end{equation}
We start by finding a reccurence relationship between the ranks $2m$ and $2m+1$ for $n$ fixed. Writing 
\begin{align*}
\begin{pmatrix}
 A^{2m+1}_{2n}(z) &  B^{2m+1}_{2n}(z)U\\
 C^{2m+1}_{2n}(z) &  D^{2m+1}_{2n}(z)U
\end{pmatrix}& =T_{2n}^{2m+1}(z)=T_{2m}^{2m+1}(z)T_{2n}^{2m-1}(z) \\
&=\begin{pmatrix}
  W_{2m}^{2m+1}(z) &  X_{2m}^{2m+1}(z)\\
     Y_{2m}^{2m+1}(z) &  Z_{2m}^{2m+1}(z)
    \end{pmatrix}.\begin{pmatrix}
    A^{2m-1}_{2n}(z) &  B^{2m-1}_{2n}(z)U\\
     C^{2m-1}_{2n}(z) &  D^{2m-1}_{2n}(z)U
    \end{pmatrix} 
\end{align*}     
    with
\begin{align*}
W_{2m}^{2m+1}(z)= & \frac{1}{z} V_{\omega}^{(2m+1)}  (\widetilde{\rho}(\alpha))^{-1} V_{\omega}^{(2m)}  (\widetilde{\rho}(\alpha))^{-1}  +V_{\omega}^{(2m+1)}     (\widetilde{\rho}(\alpha))^{-1}  \alpha^* (U_{\omega}^{(2m)})^*  (\rho(\alpha))^{-1}, \\
X_{2m}^{2m+1}(z)= & -\frac{1}{z} V_{\omega}^{(2m+1)}    (\widetilde{\rho}(\alpha))^{-1}  V_{\omega}^{(2m)}   (\widetilde{\rho}(\alpha))^{-1} \alpha^* -V_{\omega}^{(2m+1)}   (\widetilde{\rho}(\alpha))^{-1} \alpha^*  (U_{\omega}^{(2m)})^*  (\rho(\alpha))^{-1},\\
Y_{2m}^{2m+1}(z)= &- (U_{\omega}^{(2m+1)})^*\alpha   (\widetilde{\rho}(\alpha))^{-1}  V_{\omega}^{(2m)}  (\widetilde{\rho}(\alpha))^{-1} -z (U_{\omega}^{(2m+1)})^*(\widetilde{\rho}(\alpha))^{-1} (U_{\omega}^{(2m)})^*  \alpha (\widetilde{\rho}(\alpha))^{-1},\\
Z_{2m}^{2m+1}(z)=&  (U_{\omega}^{(2m+1)})^* \alpha   (\widetilde{\rho}(\alpha))^{-1} V_{\omega}^{(2m)} (\widetilde{\rho}(\alpha))^{-1}  \alpha^*+z  (U_{\omega}^{(2m+1)})^* (\widetilde{\rho}(\alpha))^{-1}    (U_{\omega}^{(2m)})^*\alpha (\rho(\alpha))^{-1},
\end{align*}
yields the following relationships (dropping the $z$ dependence):
 \begin{align*}
A^{2m+1}_{2n} &= W_{2m}^{2m+1}A^{2m-1}_{2n} + X_{2m}^{2m+1} C^{2m-1}_{2n},\ B^{2m+1}_{2n}U =W_{2m}^{2m+1}B^{2m-1}_{2n}U+X_{2m}^{2m+1}D^{2m-1}_{2n}U \\
C^{2m+1}_{2n}& = Y_{2m}^{2m+1}A^{2m-1}_{2n}+Z_{2m}^{2m+1}C^{2m-1}_{2n},\ D^{2m+1}_{2n}U =Y_{2m}^{2m+1}B^{2m-1}_{2n}U+Z_{2m}^{2m+1} D^{2m-1}_{2n}U.
\end{align*}
Summing and subtracting these, we find:  
\begin{align*}
 A^{2m+1}_{2n} \pm B^{2m+1}_{2n}U =& W_{2m}^{2m+1}(A^{2m-1}_{2n} \pm B^{2m-1  }_{2n}U) + X_{2m}^{2m+1} (C^{2m-1}_{2n} \pm D^{2m-1}_{2n}U). \\
 C^{2m+1}_{2n} \pm D^{2m+1}_{2n}U =& Y_{2m}^{2m+1}(A^{2m-1  }_{2n} \pm B^{2m-1}_{2n}U) + Z_{2m}^{2m+1} (C^{2m-1}_{2n} \pm D^{2m-1}_{2n}U).
\end{align*}
This establishes the relationship:
\begin{align*}
(D^{2m+1}_{2n} & U-C^{2m+1}_{2n}) (B^{2m+1}_{2n}U-A^{2m+1}_{2n})^{-1}  =
 \left[Y_{2m}^{2m+1}(A^{2m-1}_{2n}-B^{2m-1}_{2n}U)+ Z_{2m}^{2m+1} (C^{2m-1}_{2n}- D^{2m-1}_{2n}U)\right]    \\
 & \qquad \times  \left[W_{2m}^{2m+1}(A^{2m-1}_{2n}- B^{2m-1}_{2n}U)+ X_{2m}^{2m+1}(C^{2m-1}_{2n}- D^{2m-1}_{2n}U)\right]^{-1} 
\end{align*}
Under the condition of the invertibility of $W_{2m}^{2m+1}$, we factorize \( W_{2m}^{2m+1}(A^{2m-1}_{2n}- B^{2m-1}_{2n}U) \) on the right in the inverse to get:
{\small \begin{align*}
 (D^{2m+1}_{2n}  U-C^{2m+1}_{2n}) (B^{2m+1}_{2n}U-A^{2m+1}_{2n})^{-1} & = \left[Y_{2m}^{2m+1}+Z_{2m}^{2m+1}(C^{2m-1}_{2n}-D^{2m-1}_{2n}U)(A^{2m-1}_{2n}-B^{2m-1}_{2n}U)^{-1}\right] (W_{2m}^{2m+1})^{-1} \\
  & \quad \times\left[\one + X_{2m}^{2m+1}(C^{2m-1}_{2n}-D^{2m-1}_{2n}U)(A^{2m-1}_{2n}-B^{2m-1}_{2n}U)^{-1} (W_{2m}^{2m+1})^{-1}\right]^{-1}
\end{align*}}
Under the condition of the invertibility of $W_{2m}^{2m+1}$ and assuming that 
\begin{equation}\label{eq_cond_XW}
  \|X_{2m}^{2m+1}\|\cdot \|(W_{2m}^{2m+1})^{-1}\| <1,   
\end{equation}
using Neumann power serie, we can write: 
{\footnotesize \begin{equation}
\|(D^{2m+1}_{2n}U-C^{2m+1}_{2n})  (B^{2m+1}_{2n}U-A^{2m+1}_{2n})^{-1}\|\leq \frac{(\|Y_{2m}^{2m+1}\|+\|Z_{2m}^{2m+1}\|\cdot \|(C^{2m-1}_{2n}- D^{2m-1}_{2n}U)(A^{2m-1}_{2n}- B^{2m-1  }_{2n}U)^{-1}\|)\cdot \|(W_{2m}^{2m+1})^{-1}\|}{1-\| X_{2m}^{2m+1}\|.\|(C^{2m-1}_{2n}- D^{2m-1}_{2n}U)(A^{2m-1}_{2n}- B^{2m-1  }_{2n}U)^{-1}\| \cdot\|(W_{2m}^{2m+1})^{-1}\|}. \end{equation}}
Therefore, we find:
  \begin{align*}
\|\left(\one-M_{2n}^{2m+1}\right)^{-1}\| & \leq \frac{1}{1-\|(D^{2m+1}_{2n}U-C^{2m+1}_{2n})(B^{2m+1}_{2n}U-A^{2m+1}_{2n})^{-1}\|}\\
      &\leq \frac{1}{1- \frac{(\|Y_{2m}^{2m+1}\|+\|Z_{2m}^{2m+1}\|).\|(W_{2m}^{2m+1})^{-1}\|}{1-\| X_{2m}^{2m+1}\|.\|(W_{2m}^{2m+1})^{-1}\|.\|(C^{2m}_{2n}- D^{2m}_{2n}U)(A^{2m}_{2n}- B^{2m}_{2n}U)^{-1}\|} } 
\end{align*}
under the condition
\begin{equation}\label{eq_cond_croissance_11-x1}
     \frac{(\|Y_{2m}^{2m+1}\|+\|Z_{2m}^{2m+1}\|\cdot \|(C^{2m-1}_{2n}- D^{2m-1}_{2n}U)(A^{2m-1}_{2n}- B^{2m-1  }_{2n}U)^{-1}\|)\cdot \|(W_{2m}^{2m+1})^{-1}\|}{1-\| X_{2m}^{2m+1}\|\cdot \|(C^{2m-1}_{2n}- D^{2m-1}_{2n}U)(A^{2m-1}_{2n}- B^{2m-1  }_{2n}U)^{-1}\| \cdot\|(W_{2m}^{2m+1})^{-1}\|} < 1.
\end{equation}
 which is implied by the condition
 \begin{equation}\label{eq_cond_croissance_11-x2}
    (\|Y_{2m}^{2m+1}\|+\|Z_{2m}^{2m+1}\|+\| X_{2m}^{2m+1}\|)\cdot \|(W_{2m}^{2m+1})^{-1}\| \leq 1.
\end{equation} 
since $\|(C^{2m-1}_{2n}- D^{2m-1}_{2n}U)(A^{2m-1}_{2n}- B^{2m-1  }_{2n}U)^{-1}\| <1$.
  We set, with $n$ still fixed, 
\begin{equation}\label{def_xm_ym_zm_wm}
    x_{m}:=\|X_{2m}^{2m+1}\|\ ,\ y_{m}:=\|Y_{2m}^{2m+1}\|\ ,\ z_{m}:=\|Z_{2m}^{2m+1}\|  \ , \  w_{m}:=\|(W_{2m}^{2m+1})^{-1}\|
\end{equation}  
and for all \(m \geq n+1\),
\begin{equation}\label{def_fm}
    f_{m}^{(n)}=\|(D^{2m+1}_{2n}U-C^{2m+1}_{2n})(B^{2m+1}_{2n}U-A^{2m+1}_{2n})^{-1}\|.
\end{equation}
With these notations, we write:
\begin{align*}
 \|\left(\one-M_{2n}^{2m+1}\right)^{-1}\| & \leq \frac{1}{1-f_{m}^{(n)}}\leq \frac{1}{1-\frac{y_m w_m + z_m w_m}{1-x_m w_m f_{m-1}^{(n)}}}\leq \frac{1-x_m w_m f_{m-1}^{(n)}}{1-y_m w_m-z_m w_m-x_m w_m f_{m-1}^{(n)}} \\ 
 & \leq  \frac{1}{1-y_m w_m-z_m w_m-x_m w_m f_{m-1}^{(n)}}   
\end{align*}
Thus, we obtain a recurrence relationship for all \(m \geq n+1\):
\begin{equation}
    f_m^{(n)} \leq y_m w_m + z_m w_m+ x_m w_m f_{m-1}^{(n)}. \label{eq_recursive_relation_fm}
\end{equation}

\textbf{Step 3} Our goal now is to show that the sequence \((f_m^{(n)})_{m\geq n+1}\) is uniformly bounded in both $m$ and $n$, by a constant strictly smaller than $1$, at least for $|m-n|$ large enough. Iterating the sub-arithmetico-geometric relationship \eqref{eq_recursive_relation_fm} and still under the assumption $x_m w_m <1$, one finds
\begin{align}
\forall m\geq n+1,\ f_m^{(n)} & \leq (y_m w_m + z_m w_m)(1+x_m w_m + \cdots + (x_m w_m )^{m-n-1}) + (x_m w_m)^{m-n} f_{m-n}^{(n)}\\
 & =  \frac{y_m w_m + z_m w_m}{1-x_m w_m}  + \left(f_{m-n}^{(n)} -  \frac{y_m w_m + z_m w_m}{1-x_m w_m}\right) (x_m w_m)^{m-n} \\
  & \leq \frac{y_m w_m + z_m w_m}{1-x_m w_m}  + \left(1 -  \frac{y_m w_m + z_m w_m}{1-x_m w_m}\right) (x_m w_m)^{m-n}  \label{eq_rec_somme_sub_arith}
\end{align}
since we know that for every $m,n$, $f_m^{(n)}<1$, using Lemma \ref{lmm4.5.1}. It remains to estimate the term $\frac{y_m w_m + z_m w_m}{1-x_m w_m}$. Using Lemma \ref{lem_estime_rho}, one gets
\begin{equation}\label{eq_estim_xm_ym_zm}
\forall m\in \Z,\ y_m \leq (1+|z|)\frac{\| \alpha\|}{1- \| \alpha\|^2}  , \quad x_m \leq \left(1 + \frac{1}{|z|} \right)  \frac{\| \alpha\|}{1- \| \alpha\|^2} \quad \mbox{and} \quad z_m \leq \| \alpha \| \frac{|z| + \| \alpha\|}{1- \| \alpha\|^2} .
\end{equation}
Then, to estimate $w_m$, we write :
{\scriptsize \begin{align*}
& (W_{2m}^{2m+1}(z))^{-1}=  \left(\frac{1}{z} V_{\omega}^{(2m+1)}  (\widetilde{\rho}(\alpha))^{-1} V_{\omega}^{(2m)}  (\widetilde{\rho}(\alpha))^{-1}  +V_{\omega}^{(2m+1)}     (\widetilde{\rho}(\alpha))^{-1}  \alpha^* (U_{\omega}^{(2m)})^*  (\rho(\alpha))^{-1}\right)^{-1}, \\
& = z\left( V_{\omega}^{(2m+1)}  (\widetilde{\rho}(\alpha))^{-1} V_{\omega}^{(2m)}  (\widetilde{\rho}(\alpha))^{-1} \right)^{-1} \left( \one + z V_{\omega}^{(2m+1)}     (\widetilde{\rho}(\alpha))^{-1}  \alpha^* (U_{\omega}^{(2m)})^*  (\rho(\alpha))^{-1} \left( V_{\omega}^{(2m+1)}  (\widetilde{\rho}(\alpha))^{-1} V_{\omega}^{(2m)}  (\widetilde{\rho}(\alpha))^{-1} \right)^{-1} \right)^{-1}.
\end{align*}}
Using Lemma \ref{lem_estime_rho},
{\footnotesize \begin{equation}\label{eq_cond_w_inversible}
\| V_{\omega}^{(2m+1)}     (\widetilde{\rho}(\alpha))^{-1}  \alpha^* (U_{\omega}^{(2m)})^*  (\rho(\alpha))^{-1} \left( V_{\omega}^{(2m+1)}  (\widetilde{\rho}(\alpha))^{-1} V_{\omega}^{(2m)}  (\widetilde{\rho}(\alpha))^{-1} \right)^{-1} \| \leq |z|\|\alpha\| \frac{(2-\sqrt{1-\| \alpha\|^2})^2}{1-\| \alpha\|^2}.  
\end{equation}}
For any $\epsilon>0$ and any $z\in \mathbb{S}_{\epsilon}$, the upper bound in \eqref{eq_cond_w_inversible} tends to $0$ when $\| \alpha \|$ tends to $0$. Hence 
there exists $r_1\in (0,1)$ independent of $\epsilon$ such that for any $\alpha$ such that $\|\alpha \| \leq r_1$ and any  $z\in \mathbb{S}_{\epsilon}$, $|z|\|\alpha\| \frac{(2-\sqrt{1-\| \alpha\|^2})^2}{1-\| \alpha\|^2} <1$. For such $\alpha$, $W_{2m}^{2m+1}(z)$ is invertible and using again Lemma \ref{lem_estime_rho},
\begin{equation}\label{eq_cond_w_inversible_estimate}
\forall m\in \Z,\  w_m = \| (W_{2m}^{2m+1}(z))^{-1} \| \leq \frac{|z| (2-\sqrt{1-\| \alpha\|^2})^2}{1- |z|\|\alpha\| \frac{(2-\sqrt{1-\| \alpha\|^2})^2}{1-\| \alpha\|^2} } = \frac{|z|(1-\| \alpha\|^2) (2-\sqrt{1-\| \alpha\|^2})^2}{1-\| \alpha\|^2- |z|\|\alpha\| (2-\sqrt{1-\| \alpha\|^2})^2}.
\end{equation}
Combining \eqref{eq_estim_xm_ym_zm} and \eqref{eq_cond_w_inversible_estimate}, for any $\alpha$ such that $\|\alpha \| \leq r_1$, any $\epsilon>0$ and any $z\in \mathbb{S}_{\epsilon}$, 
\begin{equation}\label{eq_estim_mu}
\forall m\in \Z,\ x_m w_m \leq   \frac{(1+|z|)\| \alpha\| (2-\sqrt{1-\| \alpha\|^2})^2}{1-\| \alpha\|^2- |z|\|\alpha\| (2-\sqrt{1-\| \alpha\|^2})^2} := \mu_{z,\alpha}.
\end{equation}
and
\begin{align}\label{eq_estim_lambda}
 \forall m\in \Z,\ \frac{y_m w_m + z_m w_m}{1-x_m w_m} & \leq \frac{\left((1+|z|)\frac{\| \alpha\|}{1- \| \alpha\|^2}  +  \| \alpha \| \frac{|z| + \| \alpha\|}{1- \| \alpha\|^2} \right) \frac{|z|(1-\| \alpha\|^2) (2-\sqrt{1-\| \alpha\|^2})^2}{1-\| \alpha\|^2- |z|\|\alpha\| (2-\sqrt{1-\| \alpha\|^2})^2} }{1- \left(1 + \frac{1}{|z|} \right)  \frac{\| \alpha\|}{1- \| \alpha\|^2} \frac{|z|(1-\| \alpha\|^2) (2-\sqrt{1-\| \alpha\|^2})^2}{1-\| \alpha\|^2- |z|\|\alpha\| (2-\sqrt{1-\| \alpha\|^2})^2} } \nonumber \\
 & = \frac{ \| \alpha \|\left(1+2|z| + \| \alpha\| \right)|z|(2-\sqrt{1-\| \alpha\|^2})^2 }{1-\| \alpha\|^2- (2|z|+1)  \| \alpha\| (2-\sqrt{1-\| \alpha\|^2})^2 } := \lambda_{z,\alpha} 
\end{align}
For any $\epsilon>0$ and any $z\in \mathbb{S}_{\epsilon}$, both $\lambda_{z,\alpha}$ and $\mu_{z,\alpha}$ tends to $0$ when $\| \alpha \|$ tends to $0$. Hence one can find  $r_2\in (0,1)$ independent of $\epsilon$ such that for any $\alpha$ such that $\|\alpha \| \leq r_2$ and any  $z\in \mathbb{S}_{\epsilon}$, $\mu_{z,\alpha} \in (0,1)$. Then, for such $\alpha$, $\mu_{z,\alpha}^{|m-n|}$ tends to $0$ when $|m-n|$ tends to infinity. Moreover, one can also find  $r_3\in (0,1)$ independent of $\epsilon$ such that for any $\alpha$ such that $\|\alpha \| \leq r_3$ and any  $z\in \mathbb{S}_{\epsilon}$, $\lambda_{z,\alpha} \in (0,\frac{1}{2}]$. 

From these two facts and inequality \eqref{eq_rec_somme_sub_arith}, one deduces that for every $\epsilon\in (0,1)$, for any $\alpha$ such that $\| \alpha \| \leq r_0:= \min(r_1,r_2,r_3)$, there exists  $p_0>2$, and $\tilde{C}_{r_0,p_0} \in (0,1)$ such that for any $m,n\in \Z$, $|m-n|\geq p_0$ and any $z\in \mathbb{S}_{\epsilon}$,
\begin{equation}\label{eq_estim_fm}
0< f_m^{(n)} \leq \lambda_{z,\alpha} + (1-\lambda_{z,\alpha})\mu_{z,\alpha}^{|m-n|} \leq \tilde{C}_{r_0,p_0}.
\end{equation}
Following the exact same procedure, we obtain the same bound $\tilde{C}_{r_0,p_0}$ for the second term $\| (\one-N)^{-1}\|$, since the change of sign in front of $A$ and $C$ does not change all the estimates in norm. We finally deduce \eqref{eq_lemma_bound_HF-GE} with $C_{r_0,p_0}=\frac{2}{1-\tilde{C}_{r_0,p_0}}$.

\textbf{Step 4} To finish the proof it remains to discuss the assumptions made during the previous steps. In particular we have to give conditions on $\alpha$ to satisfy the conditions of invertibility of $W_{2m}^{2m+1}$  and to satisfy conditions \eqref{eq_cond_XW} and \eqref{eq_cond_croissance_11-x2}. Actually, invertibility of $W_{2m}^{2m+1}$ was discussed at \eqref{eq_cond_w_inversible} and is implied by the condition $\| \alpha \| \leq r_1$. Condition \eqref{eq_cond_XW} is implied by $\lambda_{z,\alpha}<1$ and in particular is satisfied for  $\| \alpha \| \leq r_2$. Condition \eqref{eq_cond_croissance_11-x2} is also implied by $\lambda_{z,\alpha}<1$ and leads to the same condition on $\alpha$. 
\end{proof}

\begin{proof}(of Theorem \ref{thm_exp_decay_even_odd}). Using Proposition \ref{prop_Lambda_Mu}, Lemma \ref{lemma_lower_bounded_inv_trans_mat} implies that 
$$ \|G_\omega^{[2n,2m+1]}(z,2n,2m+1)\| \leq \frac{1}{\kappa} C_{r_0,p_0} \|(T_{2n}^{2m+1}(z))^{-1}\|$$
with $\kappa$ uniform in $\epsilon\in (0,\epsilon_0]$. Taking the power $s$ and the expectancy and applying Lemma \ref{lemma_exp_decay_trans_mat_inv}, one gets the statement of Theorem \ref{thm_exp_decay_even_odd}.    
\end{proof}

\subsection{Exponential decay in the general case.}\label{sec_exp_decay_general}
To prove the decay estimate for the general Green's function, it suffices to combine all the previous results.
\begin{proof}(of Theorem \ref{thm_greengeneral}).
     For $|k-l|>4$, it is clear that there exist $m, n \in \mathbb{Z}$ such that $k \in\{2 n, 2 n+1\}$, $l \in\{2 m, 2 m+1\}$, and $|m-n|>1$. Thus, using Proposition \ref{prop_reduc_paire} and Proposition \ref{prop_reduc_suitable_volume}, we obtain that there exists $0<\kappa(\alpha, s)<\infty$ such that
\begin{align*}
\left(\mathbb{E}\left[\left\|G_\omega(z,k, l)\right\|^s\right]\right)^2 & \leq \kappa(\alpha, s) \sum_{i, j=0}^1\left(\mathbb{E}\left[\left\|G_\omega(z,2 n+2 i, 2 m-2 j+1)\right\|^{4 s}\right]\right)^{1 / 2} \\
& \leq \kappa(\alpha, s) \sum_{i, j=0}^1\left(\mathbb{E}\left[\left\|G_z^{[2 n+2 i, 2 m-2 j+1] }(2 n+2 i, 2 m-2 j+1)\right\|^{4 s}\right]\right)^{1 / 2},
\end{align*}
Then, Theorem \ref{thm_exp_decay_even_odd} gives that there exists $s_0 \in (0,1)$, $C_1>0$ and $\gamma>0$ such that
$$
\left(\mathbb{E}\left[\left\|G_\omega(z,k, l)\right\|^{s_0}\right]\right)^2 \leq C_{1} \kappa(\alpha, s_0) \sum_{i, j=0}^1 e^{-\gamma|2 n-2 m+2 i+2 j|}.
$$
By the triangle inequality we have
$$|2 n-2 m-1+2 i+2 j|  \geq\| k-l|-|(2 n-k)-(2 m+1-l)+2 i+2 j\|  \geq|k-l|-6. $$
Thus, by taking $C_2=4{C}_{1} \kappa(\alpha, s) e^{5 \gamma }$, we conclude that $\mathbb{E}\left[\left\|G_\omega(z,k, l)\right\|^{s_0}\right] \leq {C}_2 e^{- \gamma|k-l|}$.
For $|k-l| \leq 4$, we use Theorem \ref{thm_unif_bounded} to show that there exists $C_3>0$ such that
$$
\mathbb{E}\left[\left\|G_\omega(z,k, l)\right\|^{s_0}\right] \leq C_3 e^{4  \gamma} e^{-\gamma |k-l|} .
$$
By choosing $C=\max \left\{{C}_2, C_3 e^{4  \gamma}\right\}$, we obtain the desired result.
\end{proof}

\section{Dynamical Localization}\label{sec_DL}
In this Section we prove dynamical localization for $\{ \mathbb{U}_{\omega} \}_{\omega \in \Omega}$ and thus demonstrate the main result of the article, Theorem \ref{thm_DL}. To do this, we will rely on two important aspects: firstly, the use of the exponential decay of the moments of the resolvent, already demonstrated in Section \ref{sec_exp_decay_general}, to obtain an estimate of the seond order moments of the resolvent; secondly, we connect the powers of the operator to its resolvent through a general lemma about unitary operators.

\subsection{Proof of Theorem \ref{thm_second_order}.}

\begin{proof}
Let $\omega \in \Omega$. We use here the factorization \eqref{eq_factorization_proof_unif_bound} of $\mathbb{U}_{\omega}= \mathbb{Y}_\omega \mathbb{D}_\omega \mathbb{V}_\omega' \mathbb{W}_\omega$, 
with the unitaries $\mathbb{Y}_\omega = \bigoplus_{k \in \mathbb{Z}} \left( \begin{smallmatrix}
\one & 0\\
   0 & \widehat{V}^{(2k)}_\omega 
   \end{smallmatrix} \right) $, $\mathbb{D}_\omega = \bigoplus_{k \in \mathbb{Z}} \left( \begin{smallmatrix}
\one & 0\\
   0 & e^{i\theta^{(2k)}_\omega}
   \end{smallmatrix} \right) $ and $\mathbb{V}_\omega' = \bigoplus_{k \in \mathbb{Z}} \left( \begin{smallmatrix}
        \alpha & \rho(\alpha) U^{(2k)}_\omega\\
    \widetilde{\rho}(\alpha) & -\alpha^* U^{(2k)}_\omega
   \end{smallmatrix}  \right) \circ s^L$. 

We set $\mathbb{X}_\omega=\mathbb{V}_\omega'\mathbb{W}_\omega$ and define, for $\delta=(\delta_1,\hdots,\delta_L)\in\mathbb{C}^L$,
\begin{equation}
\mathbb{Y}_\omega^\delta:=\bigoplus_{k\in \mathbb{Z}} \begin{pmatrix}
   \one & 0 \\
   0 & \widehat{V}^{(2k)}_\omega e^{i(\theta^{(2k)}_\omega-\delta)}
    \end{pmatrix} \ \mbox{ and } \  \mathbb{U}_\omega^\delta:= \mathbb{Y}_\omega^\delta \mathbb{X}_\omega.
\end{equation}
Recall that we denote by \((e_{\{n,l\}})_{n \in \mathbb{Z}, l \in \{1, \ldots, L\}}\) the canonical basis of $\ell^2(\Z)\otimes \C^L$ and let $P_{2n+1}$ be the projection onto $\mathrm{Span}(e_{\{2n+1,1\}},..,e_{\{2n+1,L\}})$. Identifying $P_{2n+1}$ with its matrix in the canonical basis, one has $P_{2n+1} = \mathrm{diag}(0,\ldots, 0, \one, 0, \ldots, 0)$ with  $\one$ in the $(2n+1)^\text{th}$ position. The projection $P_{2n+1}$ corresponds to the projection on the $(2n+1)^\text{th}$ block. We define \( \eta_{2n+1} = e^{i(\theta^{(2n+1)}_\omega-\delta)}- e^{i\theta^{(2n+1)}_\omega} \) to get:
\begin{equation}\label{eq_decomp_U_delta}
      \mathbb{U}_\omega^\delta = \mathbb{Y}_\omega(\mathbb{D}_\omega + \eta_{2n+1}P_{2n+1})\mathbb{X}_\omega = \mathbb{U}_\omega + \mathbb{Y}_\omega\eta_{2n+1}P_{2n+1}\mathbb{X}_\omega.
\end{equation}
By the geometric resolvent:
  $$\forall z \in \mathbb{C}\setminus\mathbb{S}^1, \ (\mathbb{U}_\omega^\delta-z)-(  \mathbb{U}_\omega-z)^{-1}=-(\mathbb{U}_\omega^\delta-z)^{-1}\left[\mathbb{Y}_\omega\eta_{2n+1}P_{2n+1}\mathbb{X}_\omega\right](\mathbb{U}_\omega-z)^{-1}.$$
We multiply on the left by \( \mathbb{X}_\omega \) and on the right by \( \mathbb{Y}_\omega \) to get:
$$\mathbb{X}_\omega (\mathbb{U}_\omega^\delta - z) \mathbb{Y}_\omega - \mathbb{X}_\omega (\mathbb{U}_\omega - z)^{-1} \mathbb{Y}_\omega = -\mathbb{X}_\omega (\mathbb{U}_\omega^\delta - z)^{-1} \mathbb{Y}_\omega [\eta_{2n+1} P_{2n+1}] \mathbb{X}_\omega (\mathbb{U}_\omega - z)^{-1} \mathbb{Y}_\omega.$$
We define \( F_\delta(z) := \mathbb{X}_\omega (\mathbb{U}_\omega^\delta - z)^{-1} \mathbb{Y}_\omega \) and \( F(z) = \mathbb{X}_\omega (\mathbb{U}_\omega - z)^{-1} \mathbb{Y}_\omega \).
Then,
  $$F_\delta(z)-F(z)=-F_\delta(z)(\eta_{2n+1} P_{2n+1})F(z).$$
For \(\delta = (0, \ldots, \delta_k, 0, \ldots, 0)\), $\eta_{2n+1} = \text{diag}(0, \ldots, e^{i(\theta^{(2n+1)}_{\omega,k} - \delta_k)} - e^{i\theta^{(2n+1)}_{\omega,k}}, \ldots, 0)$ and we set  $\eta_{2n+1,k}=e^{i(\theta^{(2n+1)}_{\omega,k} - \delta_k)} - e^{i\theta^{(2n+1)}_{\omega,k}}$.
Furthermore, 
$$\langle e_{\{2n+1,l\}}|(F_\delta(z)-F(z))e_{\{2n+1,k\}}\rangle = \langle e_{\{2n+1,l\}}|(-F_\delta(z)(\eta_{2n+1} P_{2n+1})F(z))e_{\{2n+1,k\}}\rangle.$$
Noting that, when focusing on block \(2n+1\), we define \(F_\delta(z,l,k) := F_\delta(z,\{2n+1,l\},\{2n+1,k\})\). Thus, for all $\delta$:
\begin{align*}
    |F_\delta(z,k,k)| &= \left|\langle e_{\{2n+1,k\}}|F_\delta(z)e_{\{2n+1,k\}}\rangle\right| = \left|\langle e_{\{2n+1,k\}}|\mathbb{X}_\omega(\mathbb{U}_\omega^\delta-z)^{-1}\mathbb{Y}_\omega e_{\{2n+1,k\}}\rangle\right|\\
      &\leq \| e_{\{2n+1,k\}}\|\|\mathbb{X}_\omega(\mathbb{U}_\omega^\delta-z)^{-1}\mathbb{Y}_\omega \| \|e_{\{2n+1,k\}}\|=\|(\mathbb{U}_\omega^\delta-z)^{-1}\|\leq \frac{1}{1-|z|}.
\end{align*}
Then : 
\begin{align}
    F_\delta(z,l,k)&=\frac{F(z,l,k)}{1+\eta_{2n+1,k}F(z,k,k)} =\frac{1}{\eta_{2n+1,k}+(F(z,k,k))^{-1}}.\frac{F(z,l,k)}{F(z,k,k)}.\label{3mercihakim}
\end{align}
But: $F(z,k,k)=\frac{1}{\eta_{2n+1,k}+(F(z,k,k))^{-1}}\times1$ and $1-|z|\leq |\eta_{2n+1,k}+(F(z,k,k))^{-1}|.$ Thus, we also have $1-|\eta_{2n+1,k}+(F(z,k,k))^{-1}|\leq |z|.$ 

Recall that $\eta_{2n+1,k}=e^{i(\theta^{(2n+1)}_{\omega,k}-\delta_k)}- e^{i\theta^{(2n+1)}_{\omega,k}}$
and introduce $\delta_k \in [0,2\pi]$ such that:
$$e^{-i\delta_k}=\frac{1-e^{-i\theta^{(2n+1)}_{\omega,k}}(F(z,k,k))^{-1}}{|1-e^{i\theta^{(2n+1)}_{\omega,k}}(F(z,k,k))^{-1}|}.$$
Thus,
\begin{align*}
    &1-\left|e^{-i\theta^{(2n+1)}_{\omega,k}}\frac{1-e^{-i\theta^{(2n+1)}_{\omega,k}}(F(z,k,k))^{-1}}{|1-e^{i\theta^{(2n+1)}_{\omega,k}}(F(z,k,k))^{-1}|}-e^{-i\theta^{(2n+1)}_{\omega,k}}+(F(z,k,k))^{-1}\right|\leq |z|\\
    &\Leftrightarrow 1-\left|\frac{1-e^{-i\theta^{(2n+1)}_{\omega,k}}(F(z,k,k))^{-1}}{|1-e^{i\theta^{(2n+1)}_{\omega,k}}(F(z,k,k))^{-1}|}-(1-e^{-i\theta^{(2n+1)}_{\omega,k}}(F(z,k,k))^{-1})\right|\leq |z|\\
&\Leftrightarrow
1-\left|1-|1-e^{i\theta^{(2n+1)}_{\omega,k}}(F(z,k,k))^{-1}|\right| \leq |z|.\end{align*}
Since $|1-e^{i\theta^{(2n+1)}_{\omega,k}}(F(z,k,k))^{-1}|\leq 1$, we get:
$$1-\left|1-|1-e^{-i\theta^{(2n+1)}_{\omega,k}}(F(z,k,k))^{-1}| \right|=1-(1-|e^{-i\theta^{(2n+1)}_{\omega,k}}(F(z,k,k))^{-1}|)=|e^{-i\theta^{(2n+1)}_{\omega,k}}(F(z,k,k))^{-1}|$$
and as $1-|z|^2\leq 1-|1-e^{-i\theta^{(2n+1)}_{\omega,k}}F(z,k,k)|^2$ and by $\eqref{3mercihakim}$ : 
\begin{equation}
    (1-|z|^2)|F_\delta(l,k,z)|^2\leq \frac{1-\left|1-e^{-i\theta^{(2n+1)}_{\omega,k}}(F(z,k,k))^{-1} \right|^2}{ |\eta_{2n+1,k}+(F(z,k,k))^{-1}|^2}\frac{|F(z,l,k)|^2}{|F(z,k,k)|^2} . \label{mercihakimviolet}
    \end{equation}
   By \eqref{mercihakimviolet} and since for every $x\in \R$ and every $s\in (0,1)$, $\min(1,|x|^s)\leq|x|^s,$
$$(1-|z|^2)|F_\delta(l,k,z)|^2 \leq \frac{1-\left|1-e^{-i\theta^{(2n+1)}_{\omega,k}}(F(z,k,k))^{-1} \right|^2}{|\eta_{2n+1,k} + (F(z,k,k))^{-1}|^2}\frac{|F(z,l,k)|^s}{|F(z,k,k)|^s}.$$
Given the definition of $\eta_{2n+1,k}$: 
   $$  (1-|z|^2)|F_\delta(l,k,z)|^2\leq \frac{1-\left|1-e^{-i\theta^{(2n+1)}_{\omega,k}}(F(z,k,k))^{-1} \right|^2}{ |e^{-i\delta} -(1-e^{-i\theta^{(2n+1)}_{\omega,k}}(F(z,k,k))^{-1})|^2}\frac{|F(z,l,k)|^s}{|F(z,k,k)|^s} $$
  Let \( y = 1 - e^{-i \theta^{(2n+1)}_{\omega,k}} (F(z,k,k))^{-1} \). Then \( |y| < 1 \), and we can write:
$$(1-|z|^2)|F_\delta(z,l,k)|^s = |e^{i\theta^{(2n+1)}_{\omega,k}}(1-y)|^s = |1-y|^s.$$
We obtain:
$$(1-|z|^2)|F_\delta(z,l,k)|^2 \leq \frac{(1-|y|^2) |1-y|^s}{|e^{-i\delta_k} - y|^2} |F(z,l,k)|^s.$$
By taking the expectation:
$$ \mathbb{E}\left((1-|z|^2)|F(z,l,k)|^2\right) = \frac{1}{2\pi} \int_0^{2\pi} \mathbb{E}\left((1-|z|^2)|F(z,l,k)|^2\right) d\delta.$$
Moreover, since the Haar measure on the unit circle is invariant by rotations, the random variables \( e^{i\theta^{(2n+1)}_{\omega,k}-\delta_k} \) and \( e^{i\theta^{(2n+1)}_{\omega,k}} \) have the same distribution. 
Hence:
   $$\frac{1}{2\pi}\int_0^{2\pi}\mathbb{E}\left( (1-|z|^2)|F(z,l,k)|^2\right)d\delta=\frac{1}{2\pi}\int_0^{2\pi}\mathbb{E}\left( (1-|z|^2)|F_\delta(z,l,k)|^2\right)d\delta.$$
 By Fubini's theorem, we have:
{\small \begin{align*}
 & \frac{1}{2\pi}\int_0^{2\pi}\mathbb{E}\left((1-|z|^2)|F_\delta(z,l,k)|^2\right) d \delta=\mathbb{E}\left((1-|z|^2)\frac{1}{2\pi}\int_0^{2\pi}|F_\delta(z,l,k)|^2d\delta\right) \\
 & \leq\mathbb{E}\left(\frac{1}{2\pi}\sup_{|y|<1}\int_0^{2\pi} \frac{(1-|y|^2)|1-y|^s}{|e^{-i\delta}-y|^2}|F(z,l,k)|^s d\delta\right) \leq  2^s\mathbb{E}\left( |F(z,l,k)|^s\times\sup_{|y|<1}\frac{1}{2\pi}\int_0^{2\pi} \frac{(1-|y|^2)}{|e^{-i\delta}-y|^2} d\delta\right)
\end{align*}}
since $|1-y|^s\leq 2^s$ for $y\in \mathbb{C}$ such that $|y|<1$. Using \cite[Eq. (5.24)]{hamza2009dynamical},
$$\mathbb{E}\left((1-|z|^2)|F(z,k,l)|^2\right)\leq 2^s\mathbb{E}\left(|F(z,k,l)|^s\right) \text{ for all } s \in (0,1).$$
  Using the band structure of $\mathbb{X}_\omega^*$, we have:
\begin{align*}
    \mathbb{E}\left(|F(z,l,k)|^s\right)&= \mathbb{E}\left(\left|\langle e_{\{m,l\}}\mid \mathbb{X}_\omega(\mathbb{U}_\omega-z)^{-1}\mathbb{Y}_\omega e_{\{2n+1,k\}}  \rangle\right|^s\right)\\
    &= \mathbb{E}\left(\left|\langle \mathbb{X}_\omega^* e_{\{m,l\}}\mid (\mathbb{U}_\omega-z)^{-1}\mathbb{Y}_\omega e_{\{2n+1,k\}}  \rangle\right|^s\right).
\end{align*}
We then replace $e_{\{m,j\}}$ with its original expression in the canonical basis (\emph{i.e.} $e_{m+Lj}$), apply Minkowski's inequality and use the fact that the coefficients of $\mathbb{X}_\omega$ are bounded to write:
\begin{align*}
     \mathbb{E}\left(|F(z,l,k)|^s\right)\leq C_{1}(s)\sum_{|j-l|\leq 4L}\mathbb{E}\left(\left|\langle \mathbb{X}_\omega^* e_{m+Lj}| (\mathbb{U}_\omega-z)^{-1}\mathbb{Y}_\omega e_{\{2n+1,k\}}  \rangle\right|^s\right).
\end{align*}
Since $\mathbb{Y}_\omega$ is block diagonal with unitary blocks in $U(L)$,
$$\forall m,n \in \mathbb{Z}, \forall j,k\in \{1,\hdots,L\},\ \left|\langle e_{\{mL+j\}} \mid (\mathbb{U}_\omega-z)^{-1}\mathbb{Y}_\omega e_{\{2n+1,k\}}\rangle\right| \leq \|G_\omega(z,\Tilde{m},2m+1)\|^s$$
where $\Tilde{m}$ is the index of the block containing $mL+j$ and equals $\left[\frac{mL+j}{L}\right].$
Thus
      \begin{align*}
           \mathbb{E}\left(|F(z,l,k)|^s\right)\leq C_{1}(s)\sum_{|j-l|\leq 4L}\mathbb{E}\left(\|G_\omega(z,\left[\tfrac{mL+j}{L}\right],2m+1)\|^s\right).
 \end{align*}
Now, we need to connect $\ \mathbb{E}\left(|F(z,l,k)|^s\right)$ and $ \mathbb{E}\left( (1-|z|^2) \left|  \langle e_{\{m,l\}}| \left(\mathbb{U}_\omega-z\right)^{-1} e_{\{2n+1,k\}} \rangle\right|^2 \right).$\\
Since $\mathbb{Y}_\omega$ is a direct sum of matrices of the form $\left(\begin{smallmatrix}
 \one & 0\\
   0& \widehat{V}^{(2k)}_\omega e^{i\theta^{(2k)}_\omega}
    \end{smallmatrix}\right),$ we consider two cases. The first case is when $k$ is an even block. Then $ \mathbb{Y}_\omega e_k=e_k$ and since $ \mathbb{X}_\omega $ is unitary, we find:
\begin{align*}
   \mathbb{E} & \left( (1-|z|^2) \left|  \langle e_{l}| \left(\mathbb{U}_\omega-z\right)^{-1} e_{k} \rangle\right|^2 \right)=\mathbb{E}\left( (1-|z|^2) \left|  \langle e_{l}| \left(\mathbb{U}_\omega-z\right)^{-1}\mathbb{Y}_\omega e_{k} \rangle\right|^2 \right)\\
    =&\mathbb{E}\left( (1-|z|^2) \left|  \langle \mathbb{X}_\omega e_{l}|  \mathbb{X}_\omega \left(\mathbb{U}_\omega- z\right)^{-1}\mathbb{Y}_\omega e_{k} \rangle\right|^2 \right)        =\mathbb{E}\left( (1-|z|^2) \left|  \langle \mathbb{X}_\omega e_{l}| F(z) e_{k} \rangle\right|^2 \right) .
    \end{align*}
  In the case where $k$ is an odd block,  $\mathbb{Y}_\omega e_k = \widehat{V}^{(2k)}_\omega e^{i\theta^{(2k)}_\omega} e_k$, and we have:
      \begin{align*}
   &  \left|  \langle e_{l}|  \left(\mathbb{U}_\omega -z\right)^{-1} e_{k} \rangle\right|^2 =\left|  \langle \mathbb{X}_\omega e_{l}| \mathbb{X}_\omega \left(\mathbb{U}_\omega -z\right)^{-1}\mathbb{Y}_\omega\mathbb{Y}_\omega^* e_{\{m,k\}} \rangle\right|^2
      = \left|  \langle \mathbb{X}_\omega e_{l}| F(z)\mathbb{Y}_\omega^* e_{\{m,k\}} \rangle\right|^2 \\
      & = \left| \sum_{j=1}^{L} \langle \mathbb{X}_\omega e_{l}| F(z)(( \widehat{V}^{(m)}_\omega e^{i\theta^{(m)}_\omega})^*)_{j,k} e_{\{m,j\}} \rangle\right|^2
       = \left| \sum_{j=1}^{L} (( \widehat{V}^{(m)}_\omega e^{i\theta^{(m)}_\omega})^*)_{j,k}\langle \mathbb{X}_\omega e_{l}| F(z) e_{\{m,j\}} \rangle\right|^2
     \end{align*}
    By the Cauchy-Schwarz inequality, we obtain:
\begin{align*}
   \left|  \langle e_{l}| \left(\mathbb{U}_\omega -z\right)^{-1} e_{k} \rangle\right|^2 & \leq \left( \sum_{j=1}^{L} \left|(( \widehat{V}^{(m)}_\omega)^*)_{j,k}\right|^2 \right)\left( \sum_{j=1}^{L}\left| \langle \mathbb{X}_\omega e_{l}| F(z) e_{\{m,j\}} \rangle\right|^2 \right) 
\end{align*}
Since $\widehat{V}^{(m)}_\omega$ is a unitary matrix, its $k^{\mathrm{th}}$ column  is of norm $1$ and $\sum_{j=1}^{L} \left|(( \widehat{V}^{(m)}_\omega)^*)_{j,k}\right|^2= 1.$ Hence,
      $$  \left|  \langle e_{l}| \left(\mathbb{U}_\omega -z\right)^{-1} e_{k} \rangle\right|^2 \leq \sum_{j=1}^{L}\left| \langle \mathbb{X}_\omega e_{l}| F(z) e_{\{m,j\}} \rangle\right|^2.$$
Therefore, for \( k = mL + j \),
\begin{align*}
       \mathbb{E}\left( (1-|z|^2) \left|  \langle e_{l}| \left(\mathbb{U}_\omega-z\right)^{-1} e_{k} \rangle\right|^2 \right)&\leq \sum_{j=1}^{L}\mathbb{E}\left( (1-|z|^2) \left| \langle \mathbb{X}_\omega e_{l}| F(z) e_{\{m,j\}} \rangle\right|^2 \right) 
\end{align*}
Due to the band structure of the operator \( \mathbb{X}_\omega \), there exists a constant \( C_2(s) \) such that:
    \begin{align*}
 \mathbb{E}\left( (1-|z|^2) \left|  \langle e_{l}| \left(\mathbb{U}_\omega-z\right)^{-1} e_{k} \rangle\right|^2 \right)& \leq C_2(s) \sum_{|p-q|\leq4L}^{}\sum_{j=1}^{L} \mathbb{E}\left( (1-|z|^2) \left|F(z,p,j)\right|^2\right) \\
 & \leq C_2(s)\  2^s  \sum_{|p-q|\leq4L}^{}\sum_{j=1}^{L} \mathbb{E}\left(  \left|F(z,p,j)\right|^s\right) \\
  \leq C_1(s)C_2(s)&\ 2^s \sum_{|p-l|\leq4L}^{}\sum_{j=1}^{L} \sum_{|q-p|\leq4L}^{}\mathbb{E}\left(\left\| G_\omega\left(z,\left[\frac{q}{L}\right],\left[\frac{mL+j}{L}\right]\right) \right\|^s\right).
   \end{align*}
We conclude the proof with the exponential decay of \(\mathbb{E}( \|G_\omega(z,q,mL+j)\|^s) \) given by Theorem \ref{thm_greengeneral} and the fact that each of the three sums has a finite number of terms.
\end{proof}

\subsection{Proof of Theorem \ref{thm_DL}.}

First of all, recall that for \(U\) a unitary operator,
\begin{equation}\label{formula_unitary_spectral_thm}
\forall n \in \Z,\  U^{n} = \lim _{r \rightarrow 1^{-}} \frac{1-r^{2}}{2 \pi} \int_{0}^{2 \pi}\left(U-r e^{i \theta}\right)^{-1}\left(U^{-1}-r e^{-i \theta}\right)^{-1} e^{i n \theta} d \theta.   
\end{equation}
Using this formula, of which we find an elementary proof in \cite{hamza2007localization},  we can prove Theorem \ref{thm_DL} following exactly the proof of \cite[Theorem 8.2]{hamza2007localization}, since we already have obtained the main input of this proof, estimate \eqref{eq_thm_greengeneral}.

\bigskip

\noindent \emph{Data availability statement:} We do not analyse or generate any datasets, because our work proceeds within a theoretical and mathematical approach.

\noindent \emph{Conflict of interest statement:} all authors certify that they have no affiliations with or involvement in any organization or entity with any financial interest or non-financial interest in the subject matter or materials discussed in this manuscript.

\bigskip

\addcontentsline{toc}{chapter}{Bibliography}
\bibliographystyle{alpha}
\bibliography{2024_DLSZ_Boumaza_Khouildi}

\end{document}